%% file: main-full-version.tex
\documentclass[a4paper,UKenglish,cleveref, autoref, thm-restate]{lipics-v2021}

\usepackage{subcaption}
\usepackage[rgb]{xcolor}
\usepackage{environ}
\usepackage{lipsum}
\usepackage{amsfonts}
\usepackage{cite}
\usepackage{xspace}
\usepackage{array}
\usepackage{xparse}
\usepackage{amsmath}
\usepackage{complexity}
\usepackage{url}[spaces, hyphens]
\usepackage{ulem}
\usepackage[scr=kp]{mathalfa} 
\usepackage{outlines}
\usepackage{svg}
\usepackage{tabularx}
\usepackage{booktabs}
\usepackage{ltablex}

\usepackage{tikz}

\usetikzlibrary{mindmap,positioning}
\usetikzlibrary{graphs}
\usetikzlibrary[graphs]
\usetikzlibrary{patterns}
\usetikzlibrary{arrows,automata,positioning}
\usetikzlibrary{decorations.pathreplacing}
\usetikzlibrary{decorations}
\usetikzlibrary{decorations.pathmorphing}
\usetikzlibrary{shapes.geometric}
\usetikzlibrary{quotes}
\usetikzlibrary{shapes.callouts}
\usetikzlibrary{arrows.meta}
\usetikzlibrary{calc}
\usetikzlibrary{fit}

\tikzset{faded/.style={gray,very thin}}
\tikzset{vertex/.style={draw,circle,minimum size=5pt,inner sep=0pt}}
\tikzset{novertex/.style={circle,minimum size=5pt,inner sep=0pt}}
\tikzset{blackvertex/.style={draw,circle,minimum size=5pt,inner sep=0pt, fill=black}}
\tikzset{redvertex/.style={draw,circle,minimum size=5pt,inner sep=0pt, fill=red}}
\tikzset{redvertexfaded/.style={draw,circle,faded,minimum size=5pt,inner sep=0pt, fill=red!50}}
\tikzset{greenvertex/.style={draw,circle,minimum size=5pt,inner sep=0pt, fill=green}}
\tikzset{greenvertexfaded/.style={draw,circle,faded,minimum size=5pt,inner sep=0pt, fill=green!50}}
\tikzset{bluevertex/.style={draw,circle,minimum size=5pt,inner sep=0pt, fill=blue}}
\tikzset{bluevertexfaded/.style={draw,circle,faded,minimum size=5pt,inner sep=0pt, fill=blue!50}}
\tikzset{yellowvertex/.style={draw,circle,minimum size=5pt,inner sep=0pt, fill=yellow}}
\tikzset{yellowvertexfaded/.style={draw,circle,faded,minimum size=5pt,inner sep=0pt,
fill=yellow!50}}
\tikzset{arrow/.style={-{Latex[scale=1]},shorten >= 0pt}}

\tikzset{snake it/.style={decorate, decoration=snake}}
\tikzset{big snake/.style={decorate, decoration={snake,segment length=7mm, amplitude=3mm}}}
\tikzset{ultra snake/.style={decorate, decoration={snake,segment length=7mm, amplitude=6mm}}}

\newcommand{\lipItem}[1]{\textcolor{lipicsGray}{\sffamily\bfseries\upshape\mathversion{bold}#1}}

\pdfoutput=1 
\hideLIPIcs  

\title{Revisiting Directed Disjoint Paths on tournaments (and relatives)} 

\titlerunning{Revisiting Directed Disjoint Paths on tournaments (and relatives)}

\author{Guilherme C. M. Gomes}{LIRMM, Université de Montpellier, CNRS, Montpellier,
France\\
Universidade Federal de Minas Gerais, Belo Horizonte, Brazil}{gcm.gomes@dcc.ufmg.br}{ttps://orcid.org/0000-0002-7487-3475}{Funded by the European Union, project PACKENUM, grant number 101109317. Views and opinions expressed are however those of the author only and do not necessarily reflect those of the European Union. Neither the European Union nor the granting authority can be held responsible for them.}

\author{Raul Lopes}{LIRMM, Université de Montpellier, CNRS, Montpellier,
France\\
Hamburg University of Technology, Institute for Algorithms and Complexity, Hamburg, Germany}{rtlopes@protonmail.com}{https://orcid.org/0000-0002-7487-3475}{French project ELIT (ANR-20-CE48-0008-01) and HIDSS-0002 DASHH (Data Science in Hamburg - Helmholtz Graduate School for the Structure of Matter).}

\author{Ignasi Sau}{LIRMM, Université de Montpellier, CNRS, Montpellier,
France}{ignasi.sau@lirmm.fr}{https://orcid.org/0000-0002-8981-9287}{French project ELIT (ANR-20-CE48-0008-01).}

\authorrunning{ }

\Copyright{ }



\category{} 

\relatedversion{A conference version of this article will appear in the proceedings of \emph{ICALP 2025}. We would like to thank the reviewers for their helpful remarks.} 

\nolinenumbers 

\EventEditors{Inge Li Gørtz,
  Simon J. Puglisi, and
Martin Farach-Colton,}
\EventNoEds{2}
\EventLongTitle{ESA 2023?}
\EventShortTitle{ESA 2023?}
\EventAcronym{ESA}
\EventYear{2023}
\EventDate{September 4--6, 2023}
\EventLocation{Amsterdam, The Netherlands}
\EventLogo{}
\SeriesVolume{XXX}
\ArticleNo{ }

\DeclareMathOperator{\dpw}{\text{\sf dpw}\xspace}

\DeclareMathOperator{\Ocal}{\mathcal{O}\xspace}
\DeclareMathOperator{\List}{\mathsf L}
\DeclareMathOperator{\Mat}{\mathsf M}
\newtheorem{property}[theorem]{Property}

\renewcommand{\mid}{\bigm|}

\makeatletter

\newcolumntype{\expand}{}
\long\@namedef{NC@rewrite@\string\expand}{\expandafter\NC@find}

\NewEnviron{boxproblem}[2][]{%
  \def\boxproblem@arg{#1}%
  \def\boxproblem@framed{framed}%
  \def\boxproblem@lined{lined}%
  \def\boxproblem@doublelined{doublelined}%
  \ifx\boxproblem@arg\@empty%
  \def\boxproblem@hline{}%
  \else%
  \ifx\boxproblem@arg\boxproblem@doublelined%
  \def\boxproblem@hline{\hline\hline}%
  \else%
  \def\boxproblem@hline{\hline}%
  \fi%
  \fi%
  \ifx\boxproblem@arg\boxproblem@framed%
  \def\boxproblem@tablelayout{|>{\bfseries}lX|c}%
  \def\boxproblem@title{\multicolumn{2}{|l|}{%
      \raisebox{-\fboxsep}{\textsc{\normalsize #2}}%
  }}%
  \else
  \def\boxproblem@tablelayout{>{\bfseries}lXc}%
  \def\boxproblem@title{\multicolumn{2}{l}{%
      \raisebox{-\fboxsep}{\textsc{\normalsize #2}}%
  }}%
  \fi%
  \bigskip\par\noindent%
  \begin{tabularx}{\textwidth}{\expand\boxproblem@tablelayout}%
    \boxproblem@hline%
    \boxproblem@title\\[2\fboxsep]%
    \BODY\\\boxproblem@hline%
  \end{tabularx}%
  \medskip\par%
}
\makeatother

\definecolor{mid-green}{rgb}{0.15,0.65,0.15}
\definecolor{dark-green}{rgb}{0.15,0.25,0.15}
\definecolor{dark-red}{rgb}{0.7,0.15,0.15}
\definecolor{dark-blue}{rgb}{0.15,0.15,0.9}
\definecolor{medium-blue}{rgb}{0,0,0.5}
\definecolor{gray}{rgb}{0.5,0.5,0.5}
\definecolor{color-Ig}{rgb}{0.15,0.7,0.15}
\definecolor{darkmagenta}{rgb}{0.30, 0.0, 0.30}

\hypersetup{
  colorlinks, linkcolor={dark-blue},
citecolor={mid-green}, urlcolor={dark-blue}}

\renewcommand{\NP}{{\sf NP}\xspace}

\renewcommand{\FPT}{{\sf FPT}\xspace}
\renewcommand{\XP}{{\sf XP}\xspace}

\usepackage{tkz-berge}
\usepackage{bm}
\usetikzlibrary{decorations,arrows,shapes}

\newcommand{\inners}{1.2pt}
\newcommand{\outers}{1pt}
\newcommand{\angled}[1]{\left\langle{#1}\right\rangle}
\newcommand{\pname}[1]{{\sc #1}}
\newfunc{\YES}{YES}
\newfunc{\NOi}{NO}
\newfunc{\truet}{true}
\newfunc{\falset}{false}
\newcommand{\ol}[1]{\overline{#1}}
\newcommand{\bigO}[1]{\mathcal{O}\!\left(#1\right)}
\newcommand{\congestion}{c}

\definecolor{goodred}{HTML}{CC6677}
\definecolor{goodblue}{HTML}{332288}
\definecolor{goodyellow}{HTML}{DDCC77}
\definecolor{goodgreen}{HTML}{117733}
\definecolor{goodcyan}{HTML}{88CCEE}
\definecolor{goodwine}{HTML}{882255}
\definecolor{goodteal}{HTML}{44AA99}
\definecolor{goodolive}{HTML}{999933}
\definecolor{goodpurple}{HTML}{AA4499}

\begin{document}

\maketitle

\begin{abstract}
  In the {\sc Directed Disjoint Paths} problem ($k$-{\sc DDP}), we are given a digraph
  and $k$ pairs of terminals, and the goal is to find $k$ pairwise vertex-disjoint paths
  connecting each pair of terminals. Bang-Jensen and Thomassen [SIAM J. Discrete Math.
  1992] claimed that $k$-{\sc DDP} is \NP-complete on tournaments, and this result
  triggered a very active line of research about the complexity of the problem on
  tournaments and natural superclasses. We identify a flaw in their proof, which has been
  acknowledged by the authors, and provide a new \NP-completeness proof. From an
  algorithmic point of view, Fomin and Pilipczuk~[J. Comb. Theory B 2019] provided an
  \FPT algorithm for the edge-disjoint version of the problem on semicomplete digraphs,
  and showed that their technique cannot work for the vertex-disjoint version. We
  overcome this obstacle by showing that the version of $k$-{\sc DDP} where we allow
  congestion $c$ on the vertices is \FPT on semicomplete digraphs provided that $c$ is
  greater than $k/2$. This is based on a quite elaborate irrelevant vertex argument inspired by the edge-disjoint version, and we show that our choice of $c$ is best possible for this technique, with a counterexample with no irrelevant vertices when $c \leq k/2$. We also prove that $k$-{\sc DDP} on digraphs that can be
  partitioned into $h$ semicomplete digraphs is $\W[1]$-hard parameterized by $k+h$,
  which shows that the \XP algorithm presented by Chudnovsky, Scott, and Seymour~[J.
  Comb. Theory B 2019] is essentially optimal.
\end{abstract}



\input{sections/intro}

\input{sections/prelim}
\input{sections/hardnesses}

\input{sections/fpt}
\input{sections/future}

\bibliography{main}

\newpage
\appendix
\input{sections/flaw}

\end{document}

%% file: sections/intro.tex
\section{Introduction}\label{sec:intro}
The {\sc Disjoint Paths} problem is one of the most well-studied classical \NP-complete
graph problems~\cite{GJ79}. It consists in, given an undirected graph $G$ and $k$
\emph{requests}, which are pairs
of vertices $(s_1,t_1), \ldots, (s_k,t_k)$ known as \emph{terminals}, deciding whether $G$
contains $k$ pairwise vertex-disjoint paths connecting each $s_i$ to $t_i$, for each $i \in
[k]$. As a crucial ingredient of their Graph Minors project, Robertson and
Seymour~\cite{RobertsonS95b} proved that {\sc Disjoint Paths} is \emph{fixed-parameter
tractable} (\FPT) parameterized by $k$, that is, it can be solved in time $f(k) \cdot
n^{\Ocal(1)}$ on $n$-vertex graphs for some computable function~$f$. In this article we
focus on its directed counterpart, defined analogously for an input digraph $D$ and
called {\sc Directed Disjoint Paths}, which is known to be much harder from a
computational point of view: it is already \NP-complete for a fixed number $k=2$ of
terminals~\cite{FortuneHW80}.
A number of approaches have been proposed to cope with this
intractability, ranging from approximation algorithms\cite{ChekuriKS06,000124a},
heuristics~\cite{BrandesSNWW99,MartinsGT17}, parameterized
algorithms~\cite{FominP19,Slivkins2010}, restricting the input
digraph~$D$~\cite{FortuneHW80,thomassen_tournament_nph,Disjoint-paths-in-tournaments,FradkinS15,FradkinS13},
or relaxing the problem by allowing
congestion on
vertices~\cite{CamposC0S23,Edwards2017,GiannopoulouKKK22,KawarabayashiKK14,KawarabayashiK15,Amiri2019,Lopes2022}.
In this work we consider and combine the latter three approaches,
which we proceed to discuss.

Let us start by restricting the graph class to which the input graph of {\sc Directed
Disjoint Paths} belongs.
Two relevant classes of digraphs are typically considered when one seeks to improve the
tractability of a problem: directed acyclic graphs (DAGs) and \emph{tournaments} (that
is, digraphs that can be obtained from a complete graph by orienting each edge). For the
former, the (parameterized) complexity of {\sc Directed Disjoint Paths} is well
understood: the problem is \NP-complete and  solvable in time $n^{\Ocal(k)}$ (hence, in
the class \XP)~\cite{FortuneHW80}, and $\W[1]$-hard~\cite{Slivkins2010} (even to
approximate within a constant factor~\cite{000124a}), thus unlikely to
be \FPT. For the latter, the landscape is more murky, and the goal of this article is to
contribute to understanding it a bit better.

Bang-Jensen and Thomassen~\cite{thomassen_tournament_nph} {\sl claimed} that {\sc
Directed Disjoint Paths} is \NP-complete on tournaments when $k$ is part of the input,
and this result triggered a long line of research. In the same paper, Bang-Jensen and
Thomassen~\cite{thomassen_tournament_nph} showed that the problem can be solved in
polynomial time for $k=2$. In fact, their algorithm works on the larger class of
\emph{semicomplete digraphs},  where each pair of distinct vertices has at least one arc between them, instead of exactly one as in tournaments.
Chudnovsky, Scott, and Seymour~\cite{Disjoint-paths-in-tournaments} showed that the
problem can be solved in polynomial time for every fixed $k$ on semicomplete digraphs, by
providing an \XP algorithm. Fradkin and Seymour~\cite{FradkinS15} proved that the
edge-disjoint version of the problem, for which all the results discussed above also
apply, is also polynomial-time solvable on semicomplete digraphs for fixed $k$. It is
worth mentioning that Chudnovsky, Fradkin, Kim, Scott, and
Seymour~\cite{ChudnovskyFS12,Disjoint-paths-in-tournaments,ChudnovskyS11,FradkinS15,FradkinS13,KimS15}
built a containment theory on semicomplete digraphs for the study of minor-related
problems, such as {\sc Directed Disjoint Paths} (see also the work of  Barbero, Paul, and
Pilipczuk~\cite{BarberoPP19}). One of the key notions in this theory is a width measure
of digraphs called \emph{directed pathwidth} -- a generalization of undirected
pathwidth (see \autoref{sec:prelim} for the
definition) -- that also plays a role in the current article.

Back to the {\sc Directed Disjoint Paths} problem, the \XP algorithm on
tournaments~\cite{Disjoint-paths-in-tournaments} has been generalized by Chudnovsky,
Scott, and Seymour~\cite{chudnovsky_union_of_tournaments} to the class
${\cal C}_h$ of digraphs whose vertex set can be partitioned into a bounded number $h$
of semicomplete digraphs (thus, semicomplete digraphs correspond to the class ${\cal
C}_1$). More precisely, the running time of their algorithm is $n^{\bigO{(hk)^5}}$. The
authors asked whether the result can be further generalized to the class ${\cal A}_h$ of
digraphs whose underlying graph has independence number at most $h$, and this is still
open. For the edge-disjoint version of the problem, an affirmative answer was given by
Fradkin and Seymour~\cite{FradkinS15}.

Concerning the fixed-parameter tractability of the problem, it is open whether  {\sc
Directed Disjoint Paths} on tournaments (or semicomplete digraphs) is \FPT
parameterized by $k$. Interestingly, the edge-disjoint version of the problem was shown
to be \FPT on semicomplete digraphs by Fomin and Pilipczuk~\cite{FominP19}, by solving
a more general problem called {\sc Rooted Immersion}. In a nutshell (more details are
given below when discussing our techniques in \autoref{sec:overview}),
theirs is a win/win approach based on the directed pathwidth of the input
digraph. If the pathwidth is small (as a function of $k$), then a dynamic programming
algorithm is used to solve the problem. Otherwise, it is shown that the digraph must
contain a large obstruction to directed pathwidth called a \emph{triple}~\cite{FradkinS13}, which is used to
find an \emph{irrelevant vertex} for the problem, that is, a vertex whose removal
does not affect the existence of a solution.

Note that the above win/win approach needs, as a first step, to compute the directed pathwidth
of the input digraph in \FPT-time. While this problem is open on general digraphs, an
\FPT algorithm on semicomplete digraphs is given in the same article~\cite{FominP19}.
This \FPT algorithm to compute directed pathwidth was generalized by Kitsunai, Kobayashi, and
Tamaki~\cite{KitsunaiKT15} to another superclass of semicomplete digraphs called
\emph{$h$-semicomplete digraphs} and denoted by ${\cal S}_h$ for some fixed integer
$h \geq 0$. Digraphs in this class satisfy the property that each vertex has at most $h$
non-neighbors (in any direction). If we denote by ${\cal T}$ and ${\cal S}$ the classes of tournaments and semicomplete
digraphs, respectively, it holds that ${\cal S} = {\cal S}_0$, and for every $h \geq
0$ (see the discussion in the beginning of \autoref{sec:tournaments-relatives}),
\begin{equation}\label{eq:containment-classes}
  {\cal T} \subseteq {\cal S}  \subseteq {\cal S}_h \subseteq  {\cal C}_{h+1}
  \subseteq {\cal A}_{h+1}.
\end{equation}

Another transversal strategy to try to overcome the inherent hardness of {\sc
Directed Disjoint Paths} is to relax the problem by allowing every vertex of $D$ to
be used by at most $c$ of the $k$ paths of the solution, for some integer $c \geq 1$
that is also part of the input and is called the \emph{congestion}; we call the
corresponding problem {\sc Directed Disjoint Paths with Congestion}. Note that if $c \geq k$ the problem is trivial, so we may
assume that $c \leq k -1$. It is open whether the case $k=3$ and $c=2$ (which is the
first non-trivial one with congestion greater than one) can be solved in polynomial
time on general digraphs. There are, however, some positive results. For instance,
Edwards, Muzi, and Wollan~\cite{Edwards2017} showed that the problem for $c=2$ can be
solved in polynomial time if the input graph is sufficiently connected as a function
of $k$, which is not the case for the truly disjoint version~\cite{Thomassen91}.
See~\cite{CamposC0S23} for recent improvements of this result. A popular variant of
the congested problem is an \emph{asymmetric} version, where the goal is to either
find a congested solution or to provide a {\sf no}-answer for the disjoint version.
This problem has been proved to be \XP parameterized by $k$ on general digraphs for
some small values of $c$~\cite{GiannopoulouKKK22,KawarabayashiKK14,KawarabayashiK15},
usually exploiting the celebrated Directed Grid
Theorem~\cite{Campos2022,KawarabayashiK15,HatzelKMM24}. On the other hand, other
articles~\cite{Amiri2019,Lopes2022} study the parameterized complexity of {\sc
Directed Disjoint Paths with Congestion} by considering parameters stronger than $k$.


Finally, it is worth mentioning that Cavallaro, Kawarabayashi, and
Kreutzer~\cite{CavallaroKK24} recently proved that the edge-disjoint version of {\sc
Directed Disjoint Paths} is \FPT on Eulerian digraphs parameterized by $k$; this is
one of the rare examples where (some variant of) the problem is \FPT. We refer to the
book of Bang-Jensen and Gutin~\cite{BangJensen2018} for a thorough introduction to
algorithms on digraphs, in particular on tournaments and related superclasses.

In the sequel, for notational conciseness we may use the abbreviations $k$-{\sc DDP}
and $(k,c)$-{\sc DDP} to refer to the {\sc Directed Disjoint Paths } and {\sc
Directed Disjoint Paths with Congestion} problems, respectively. We permit ourselves
to slightly abuse notation by including the integers $k$ and $c$ in the abbreviated
names of the problems, even if they are part of the corresponding  inputs.

\subparagraph*{Our contributions.} As mentioned above, the \NP-completeness proof of
    Bang-Jensen and Thomassen~\cite{thomassen_tournament_nph} of $k$-{\sc DDP} on
    tournaments triggered intensive research in this
    area~\cite{FominP19,chudnovsky_union_of_tournaments,Disjoint-paths-in-tournaments}.
    Unfortunately, we realized that their proof has a flaw that does not seem to be
    easily fixable, as acknowledged by the authors~\cite{personal-communication}; see \autoref{sec:flaw} for an explanation of this flaw. Our
    first contribution is to provide a new (correct) proof of this result.

\begin{restatable}{theorem}{tournamentnph}\label{thm:tournament_nph}
    \pname{Directed Disjoint Paths} on tournaments is  \NP-complete.
\end{restatable}

As mentioned above, it is open whether $k$-{\sc DDP} on tournaments is \FPT
parameterized by $k$. Recall that the win/win approach of Fomin and
Pilipczuk~\cite{FominP19} for the edge-disjoint version has two main ingredients
(other than computing the directed pathwidth): a dynamic programming algorithm and an irrelevant
vertex argument. While the former is claimed to exist for the vertex-disjoint
version~\cite{FominP19}, the latter one is doomed to fail: Fomin and
Pilipczuk~\cite{FominP19} provide a counterexample even for $k=2$ consisting of a
family of tournaments containing arbitrarily large triples (that are the structures
where irrelevant vertices are found), but in which each vertex is relevant. Our next
contribution is to prove that this obstacle disappears if we allow for a large
congestion.

\begin{restatable}{theorem}{FPTinSemicomplete}
\label{thm:FPTinSemicomplete}
  $(k,c)$-{\sc DDP} on semicomplete digraphs is \FPT parameterized by $k$ restricted
  to instances satisfying $c > k/2$.
\end{restatable}

Note that since we can assume that $c \leq k -1$, the result of \autoref{thm:FPTinSemicomplete}
covers roughly ``half'' of the range of values of the congestion $c$. It is natural
to ask whether the problem remains \NP-complete for this range of values of $c$, that
is, when the congestion is lower-bounded by a {\sl linear} function of $k$.
This question is still open, but we provide the following hardness result, where the
congestion $c$ is {\sl almost linear} in $k$ (as $\varepsilon$ approaches one).

\begin{restatable}{theorem}{epsilonddp}\label{thm:epsilon_ddp}
     \pname{$(k,c)$-DDP} remains \NP-complete on tournaments even restricted to instances satisfying $c = k^{\varepsilon}$, for every $\varepsilon \in [0,1)$.
\end{restatable}

As discussed before, Chudnovsky, Scott, and
Seymour~\cite{chudnovsky_union_of_tournaments} showed that {\sc Directed Disjoint
Paths} on the class ${\cal C}_h$ can be solved in time  $n^{\Ocal((hk)^5)}$. Our next
result is to show that this algorithm is somehow optimal, in the sense that it is
unlikely to get rid of both parameters in the exponent of $n$, even restricted to
digraphs of bounded directed pathwidth.

\begin{restatable}{theorem}{whch}\label{thm:w1h_ch}
     The \pname{$k$-DDP} problem on ${\cal C}_h$ is $\W[1]$-hard when parameterized by $k+h$, even if restricted to input digraphs of directed pathwidth two.
\end{restatable}

Moreover, \autoref{thm:w1h_ch} can be generalized to show hardness for \pname{$(k,c)$-DDP} when $c > k/2$ (cf. \autoref{thm:w1h_congestion}, which, alongside \autoref{eq:containment-classes}, implies that the win/win strategy cannot be extended beyond $h$-semicomplete graphs).
we summarize of our main contributions in \autoref{table:summary-of-results}.

 \subparagraph*{Organization.} 
    In \autoref{sec:overview} we give an overview of the techniques used to obtain our results.
  In \autoref{sec:prelim} we provide preliminaries about general digraphs,
  tournaments, related classes, directed pathwidth, and parameterized complexity. 
  In \autoref{sec:hardnesses} (resp. \autoref{sec:algo}) we provide our negative (resp. positive) results. We conclude the article in \autoref{sec:conclusions} with some directions for further research.
  


\section{Overview of our techniques}
\label{sec:overview}


The reduction that we use to prove \autoref{thm:tournament_nph} is novel and versatile enough so
that we can build upon it and modify it appropriately to prove the \NP-completeness of
several variants of the problem. Intuitively, in the proof of \autoref{thm:tournament_nph} we
reduce from a variant of \pname{3-SAT} where each variable has a bounded number of occurrences (namely, exactly three positive and one negative). This
allows us to build an instance of \pname{$k$-DDP} where a variable's truth value is determined by
two requests (cf. \autoref{fig:directed_butterfly}), while the satisfaction of the clauses is encoded using one additional
request per clause; interestingly, the paths that fulfill the requests have length at
most five. By extending this construction using a single long path, named the \emph{critical path} (cf. the black path in \autoref{fig:queued_butterfly2}), which is the unique way to fulfill several additional requests, we show how to prove \NP-hardness for \pname{$(k,c)$-DDP} for digraphs
in $\mathcal{C}_2$ even if $c = \varepsilon k$ for $\varepsilon \in [0,1)$ (cf.
\autoref{thm:c2_hardness}).
With this approach, however, we are unable to prove hardness for \pname{$(k,c)$-DDP} instances where $c = \varepsilon k$ and $D$ is a tournament, only doing so for the weaker relation $c = k^\varepsilon$ (cf. \autoref{thm:epsilon_ddp}).

For the \FPT algorithm given in \autoref{thm:FPTinSemicomplete}, the challenge is that now we deal with vertex-disjoint
paths, instead of edge-disjoint paths as in the original work by Fomin and
Pilipczuk~\cite{FominP19}. Much like theirs, our proof is a win/win approach that makes extensive use of \textit{$k$-triples} (cf. \autoref{def:k-triple}), an obstacle for directed pathwidth on tournaments introduced by Fradkin and Seymour~\cite{FradkinS13}, and has two steps:
(\textit{i}) show that, in a minimum solution, almost all vertices (in their case, arcs) of a sufficiently large (i.e., with more than $f(k)$ vertices) triple can be used by other paths, and (\textit{ii}) identify a vertex in the triple that can be safely removed from the instance.
In the edge-disjoint setting, it was extensively used that, in a large enough triple, vertices had large in- or out-degree; consequently, finding available arcs to either shortcut (step \textit{i}) or reroute (step \textit{ii}) a solution was relatively simple.
In the vertex-disjoint case, this does not happen; in fact, Fomin and Pilipczuk~\cite{FominP19} present an infinite family of tournaments that have no irrelevant vertex for \pname{$k$-DDP}. As we show in \autoref{sec:counterexample}, their counterexamples are also valid when $c \leq k/2$ but, when $c > k/2$, we are able to easily find alternative paths, freeing up several vertices and thus making them irrelevant.
In particular, we are able to implement step \textit{i} using, in particular, the pigeonhole principle: any two fully congested vertices have at least one path in common. This is not enough by itself, and we must carefully construct two shortcuts simultaneously, instead of only one, to show that vertices occupied by exactly $c$ paths only amount to $\bigO{\binom{k}{c}}$ (see \autoref{lem:bounded-congested-in-B} elements of the triple. Step \textit{ii} also offers additional challenges; in particular, when dealing with the exterior neighborhood of the triple, we must find large matchings entering and leaving the triple to properly reroute the paths using it.
As such, instead of the polynomial on $k$ used by Fomin and Pilipczuk, we now require that the triple has size of the order of $2^{\bigO{k\log k}}$.

Given the success of the win/win algorithms based on directed pathwidth, it is natural to see how far one can push this approach. The results of Kitsunai, Kobayashi, and Tamaki~\cite{KitsunaiKT15} that triples are also directed pathwidth obstacles for $h$-semicomplete digraphs is a further step in this direction. By~\autoref{thm:w1h_ch} and~\autoref{eq:containment-classes} however, the class of $h$-semicomplete digraphs is essentially as far as it goes, even if we allow for congestion (cf. \autoref{thm:w1h_congestion}). To prove \autoref{thm:w1h_ch}, we take inspiration from Slivkins'~\cite{Slivkins2010} proof that \pname{$k$-Edge Disjoint Paths} is \W[1]-hard parameterized by $k$ on DAGs.
In his reduction from the \pname{Clique} instance $(G,\ell)$, a DAG in a matrix-like format is built in a way that each row corresponds to a copy of $G$ and each column to a vertex. Moreover, only one cell may be available in each row, and this must correspond to a vertex of $G$ in the clique. Requests are added so that $\ell$ of them are used to enforce the uniqueness, while $\binom{\ell}{2}$ are used to encode the edges of the clique (cf. \autoref{fig:slivkins_s3}).
The uniqueness of the available cell is the main obstacle to obtain~\autoref{thm:w1h_ch}. In Slivkins' proof, it is obtained by placing two parallel paths, named $a$ and $b$, with a request from the first vertex of $a$ to the last of $b$, which may only be fulfilled if we perform a jump at the appropriate point: from where the segment of $a$ corresponding to the desired vertex begins to where the corresponding segment in $b$ ends.
We adapt his proof by making each of the paths $a$ and $b$ a tournament (cf. \autoref{fig:slivkins_s1}). This, however, opens us to cheating: an edge-encoding path could randomly walk along a suffix of $a$ or a prefix of $b$ and break the proof.
To overcome this, we require that $a$ and $b$ are completely occupied, except at the segments corresponding to the desired vertex of $G$; this is achieved by introducing several structures and requests in order to enforce a synchronous behavior in each row of the matrix (cf. \autoref{fig:slivkins_s4}).
Using a long path strategy, similarly to the one used in~\autoref{thm:c2_hardness}, we are able to force that every vertex participates in the fulfillment of several requests while keeping both the congestion high ($\congestion = k/d$ for some constant $1 < d \leq k$) and the directed pathwidth equal to two, proving that a statement analogous to \autoref{thm:w1h_ch} (i.e., \autoref{thm:w1h_congestion}) also holds for \pname{$(k,c)$-DDP}.

%% file: sections/prelim.tex
\section{Definitions and preliminaries}
\label{sec:prelim}

In this section we provide basic preliminaries about digraphs (including the definition of directed pathwidth), parametererized complexity, and tournaments and related classes along with the notion of $k$-triple. In \autoref{sec:counterexample}
we present and analyze the counterexample given by  Fomin and Pilipczuk~\cite{FominP19} that shows that the irrelevant vertex technique is not extensible to \pname{$k$-DDP}. 

\subsection{Digraphs, directed pathwidth, and parametrized complexity}
For basic background on graph theory we refer the reader to~\cite{Bondy2008}.
Since in this article we mainly work with digraphs, we focus on basic definitions of
digraphs, often skipping their undirected counterparts.
Given a digraph $G$ we denote by $V(G)$ and $E(G)$ the sets of vertices and arcs of $G$, respectively.
All involved digraphs are simple, i.e., they have neither loops nor parallel arcs.
Given $X \subseteq V(G)$, we denote by $G \setminus X$ the digraph resulting from
removing every vertex of $X$ from $G$.
We denote by $G[X]$ the subgraph of $G$ \emph{induced} by $X$.

If $e$ is an arc of a digraph from a vertex $u$ to a vertex $v$, we say that $e$ has
\emph{endpoints} $u$ and $v$, that $e$ is \emph{incident} to $u$ and $v$, and that $e$ is
\emph{oriented} from $u$ to $v$. We may refer to $e$ as the ordered pair $(u,v)$.
In this case, $u$ is the \emph{tail} of $e$ and $v$ is the \emph{head} of $e$.
We also say that $e$ is \emph{leaving} $u$ and \emph{reaching} $v$, and that $u$ and $v$
are \emph{adjacent}.
A \emph{clique} in a (di)graph $G$ is a set of pairwise adjacent vertices of $G$, and an
\emph{independent set} is a set of pairwise non-adjacent vertices of $G$.
A pair of arcs is adjacent if they share an endpoint.
A \emph{matching} is a set of pairwise non-adjacent arcs.
For $A, B \subseteq V(G)$, we say that $M$ is a matching \emph{from $A$ to $B$} if every arc of $M$ has tail in $A$ and head in $B$.
The \emph{in-degree} (resp. \emph{out-degree}) of a vertex $v$ in a digraph $G$ is the
number of arcs with head (resp. tail) $v$.
We denote the in-degree and out-degree of $v$ by $\deg^-_G(v)$ and $\deg^+_G(v)$, respectively.

The \emph{in-neighborhood} $N^-_D(v)$ of $v$ is the set $\{u \in V(D) \mid (u,v) \in
E(G)\}$, and the \emph{out-neighborhood} $N^+_D(v)$ is the set $\{u \in V(D) \mid (v,u)
\in E(G)\}$.
We say that $u$ is an \emph{in-neighbor} of $v$ if $u \in N^-_D(v)$ and that $u$ is an
\emph{out-neighbor} of $v$ if $u \in N^+_D(v)$.
We extend these notations to sets of vertices: given $X \subseteq V(D)$, we define
$N^-_D(X) = (\bigcup_{v \in X}N^-_D(v)) \setminus X$ and $N^+_D(X) = (\bigcup_{v \in
X}N^+_D(v)) \setminus X$.

A \emph{walk} in a digraph $G$ is an alternating sequence $W$ of vertices and arcs that
starts and ends with a vertex, and such that for every arc $(u,v)$ in the walk, vertex
$u$ (resp. vertex $v$) is the element right before (resp. right after) arc $(u,v)$ in $W$.
If the first vertex in a walk is $u$ and the last one is $v$, then we say this is a
\emph{walk from $u$ to $v$}.
A \emph{path} on $p$ vertices is a digraph formed by $p$ pairwise distinct vertices; to avoid confusion with sets of vertices, $P = \angled{v_1, \dots, v_p}$ always denotes the path $P$ \emph{starting} at $v_1$, \emph{ending} at $v_p$, following the arcs in the given order of the vertices. We denote by $\prec_P$ the \emph{comes before than in $P$} relation, i.e., $v_i \prec_P v_j$ if and only if $i < j$.
When the ordering itself is unimportant, we use $V(P)$ to refer to its vertex set.

Given $X,Y \subseteq V(D)$, we say that $X$ is \emph{complete to} $Y$ if $(u,v) \in E(D)$ whenever $u \in X$ and $v \in Y$.

For a positive integer $k$, we denote by $[k]$ the set of integers $\{1, \ldots, k\}$.

\subparagraph*{Directed pathwidth and $k$-triples.} A \textit{path decomposition} of a digraph $D$ is
a sequence $(X_1, \ldots , X_p)$ of vertex subsets of $D$, called \textit{bags}, such that the
following three conditions are satisfied:
\begin{enumerate}
    \item $\bigcup_{i \in [p]}X_i = V(D)$,
    \item for every arc $(u, v)\in E(D)$, we either have $u,v \in X_i$ for some $X_i$, or $u \in  X_i$ and $v \in X_j$ for some $i \geq j$, and
\item for every vertex $v \in V (D)$, the set $\{i \mid v \in X_i\}$ of indices of the bags containing $v$ forms a single integer interval.    
\end{enumerate}
Similarly to undirected pathwidth, the \textit{width} of a path decomposition $(X_1, \ldots , X_p)$  of a digraph $D$ is defined as $\max_{i \in [p]}|X_i|-1$, and the \textit{directed pathwidth} of $D$, denoted by $\dpw(D)$, is the smallest integer $\ell$ such that there exists a path decomposition of $D$ of width $\ell$. 

\begin{definition}[$k$-triple]\label{def:k-triple}
  Let $D$ be a digraph. For an integer $k \geq 1$, a \emph{$k$-triple} $\mathcal{K}$ of $D$ is formed by an
  ordered triple $(A, B, C)$  of disjoint subsets of $V(D)$ and
  \begin{outline}
    \1 $|A| = |B| = |C| = k$; and
    \1  there are orderings $(a_1, \ldots, a_k)$, $(b_1, \ldots, b_k)$, and $(c_1,
    \ldots, c_k)$ of $A,B$, and $C$, respectively, such that
    \2 for all  $i,j \in [k]$ we have $(a_i,b_j), (b_i,c_j) \in E(D)$, and
    \2 for all $i \in [k]$ we have $(c_i, a_i) \in E(T)$.
  \end{outline}
\end{definition}
When working with a $k$-triple $(A,B,C)$,  it is useful to have an easy way
to refer the associated endpoints of the matching from $C$ to $A$.
Thus we sometimes refer to this matching as a bijective mapping $\Mat$ from $C$
to $A$.
That is, for $i \in [k]$ we have $\Mat(c_i) = a_i$ and $\Mat^{-1}(a_i) = c_i$.
In addition, we extend graph-theoretical notation to $k$-triples in the
following way.
If $\mathcal{K}$ is a $k$-triple, we denote by $V(\mathcal{K})$ the set $A \cup
B \cup C$ and by $E(\mathcal{K})$ the set of all arcs between pairs of
vertices of $V(\mathcal{K})$.
See \autoref{fig:k-triple-example} for an example of a $4$-triple.

\begin{figure}[h]
  \centering
  \begin{tikzpicture}[yscale=.6]
    \foreach \i in {1,...,4} {
      \node[blackvertex, label={[yshift=-.1cm]-90:{$a_\i$}}] (a\i) at (\i-1,-.5) {};
      \node[blackvertex, label={[xshift=-.1cm]180:{$b_\i$}}] (b\i) at (-2,\i-1) {};
      \node[blackvertex, label={[yshift=.1cm]90:{$c_\i$}}] (c\i) at ($(a\i) + (0,4.5)$) {};
      \draw[arrow] (c\i) -- (a\i);
    }
    \foreach \i in {1,...,4}
    \foreach \j in {1,...,4} {
      \draw[arrow,opacity=.5] (a\i) -- (b\j);
      \draw[arrow,opacity=.5] (b\i) -- (c\j);
    }
    \node[rectangle,draw,fit=(a1)(a4), label=0:{$A$}] {};
    \node[rectangle,draw,fit=(b1)(b4), label=90:{$B$}] {};
    \node[rectangle,draw,fit=(c1)(c4), label=0:{$C$}] {};
  \end{tikzpicture}
  \caption{A $4$-triple $(A, B, C)$. We remark that nothing is known about the arcs
  inside $A$, $B$, or $C$ nor about arcs between $C$ and $A$ other than the ones in the matching.}
  \label{fig:k-triple-example}
\end{figure}
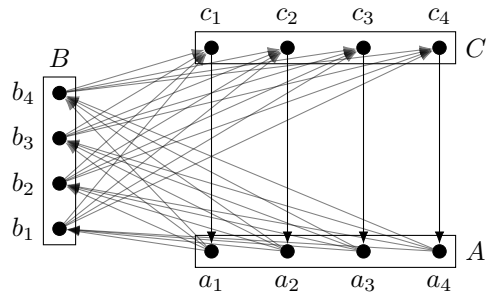

\subparagraph*{Parameterized complexity.}  A \emph{parameterized problem} is a language $L \subseteq \Sigma^* \times \mathbb{N}$, for some finite alphabet $\Sigma$.  For an instance $(x,k) \in \Sigma^* \times \mathbb{N}$, the value~$k$ is called the \emph{parameter}. Such a problem is \emph{fixed-parameter tractable} (\FPT for short) if there is an algorithm that decides membership in $L$ of an instance $(x, k)$ in time $f(k)\cdot {|x|}^{\Ocal(1)}$ for some computable function~$f$, and it is in the class \XP if there is an algorithm that decides membership in $L$ of an instance $(x, k)$ in time $f(k)\cdot {|x|}^{g(k)}$ for some computable functions~$f$ and $g$. Consult the monographs~\cite{CyganFKLMPPS15,Niedermeier06,FlumG06,DoFe13,FominLSZ19} for background on parameterized complexity.

\subsection{Tournaments and relatives}
\label{sec:tournaments-relatives}

As already said in the introduction, we denote by ${\cal T}$ and ${\cal S}$ the classes of tournaments and semicomplete digraphs, respectively.  For every non-negative integer $h$, we consider three classes of digraphs.
The class $\mathcal{S}_h$ contains every digraph $D$ such that every $v \in V(D)$ has at
most $h$ non-neighbors.
The class $\mathcal{C}_h$ contains every digraph whose vertex set can be partitioned into
at most $h$ cliques.
The class $\mathcal{A}_h$ contains every digraph $D$ with $\alpha(D) \leq h$.

Brook's Theorem~\cite{Brooks_1941} states that the vertex set of every graph $G$ with
maximum degree $\Delta(G) \leq k$ can be partitioned into $k + 1$ color classes $X_1,
\ldots, X_{k+1}$ such that each $X_i$ induces an independent set in $G$.
Let $D \in \mathcal{S}_h$.
Applying Brook's Theorem to the complement $\overline{D}$ of $D$, we conclude that $V(D)$
can be partitioned into $h+1$ cliques and thus $\mathcal{S}_h \subseteq \mathcal{C}_{h+1}$.
Since every independent set of a digraph intersects any clique in at most one vertex, it
also holds that $\mathcal{C}_h \subseteq \mathcal{A}_{h}$.
In \autoref{table:summary-of-results} we give a summary of our main results.

\begin{table}
\begin{tabularx}{.99\textwidth}{l >{\centering\arraybackslash}X >{\centering\arraybackslash}X >{\centering\arraybackslash}X}
\toprule
  & $\mathcal{T}$ & $\mathcal{S}$ & $\mathcal{C}_h$\\ 
 \midrule
 $k$-\textsc{DDP} & \NP-c. [\autoref{thm:tournament_nph}] & \textendash & \textendash\\
 \midrule
 $(k,c)$-\textsc{DDP} & \NP-c. for $c = k^{\varepsilon}$, $\varepsilon \in [0,1)$ [\autoref{thm:epsilon_ddp}] & \FPT by $k+h$ when $2c > k$ [\autoref{thm:FPTinSemicomplete}] & \W[1]-h. by $k+h$ with $\dpw(T) \leq 2$ [\autoref{thm:w1h_ch}]\\
\bottomrule 
\end{tabularx}
\caption{Summary of our main results.}
\label{table:summary-of-results}
\end{table}


\subsection{Analyzing the Fomin-Pilipczuk counterexamples}
\label{sec:counterexample}

Let us generalize the counterexamples given in~\cite{FominP19} to other values of the congestion; this will show that the irrelevant vertex technique is not extensible to \pname{$(k,c)$-DDP} for appropriate values of $c$. For each $n \geq 1$ and $c \geq 1$, we build an instance $(T_n, K, c)$ of \pname{$(k,c)$-DDP} in tournaments such that:
\begin{enumerate}
    \item The digraph $T_n$ has $4(n+1)$ vertices.
    \item The vertex set of $T_n$ is partitioned into two sets $U,V$ of the same size, with $U = \{u_i \mid i \in [2n+2]\}$ and $V = \{v_i \mid i \in [2n+2]\}$.
    \item The tournament $T_n[U]$ has a path $P_u = \angled{u_1, \dots, u_{2n+2}}$ and all other arcs going in the opposite direction. Similarly, $T_n[V]$ has a path $P_v = \angled{v_{2n+2}, \dots, v_1}$ and all other arcs going in the opposite direction.
    \item For every pair $(u_i, v_j) \in U \times V$, arc $(u_i, v_j)$ exists if $i = j$, otherwise we have the arc $(v_j, u_i)$.
    \item Digraph $T_n$ has an $n$-triple where $A = \{v_1, v_n\}$, $B = \{u_{n+3}, \dots, u_{2n+2}\}$, and $C = \{u_1, \dots, u_n\}$.
    \item There are $c$ requests $(u_1, u_{2n+2})$ and $c$ requests $(v_{2n+2}, v_1)$.
\end{enumerate}

\begin{figure}[!htb]
  \centering
  \begin{tikzpicture}[xscale=1.0]
    \GraphInit[unit=3,vstyle=Normal]
    \SetVertexNormal[Shape=circle, FillColor=black, MinSize=2pt]
    \tikzset{VertexStyle/.append style = {inner sep = \inners, outer sep = \outers}}
    \SetVertexLabelOut

    \begin{scope}
      \Vertex[x=0, y=0, Lpos=180, Math, L={u_1}]{u1}
      \Vertex[x=2, y=0, Lpos=90, Math, Ldist=0.2cm, L={u_2}]{u2}
      \Vertex[x=4, y=0, Lpos=90, Math, Ldist=0.2cm, L={u_3}]{u3}
      \Vertex[x=6, y=0, Lpos=90, Math, Ldist=0.2cm, L={u_4}]{u4}
      \Vertex[x=8, y=0, Lpos=90, Math, Ldist=0.2cm, L={u_5}]{u5}
      \Vertex[x=10, y=0, Lpos=0, Math, L={u_6}]{u6}

      \Edges[style={-Latex}](u1,u2,u3,u4,u5,u6)

      \node at (-1, 0) {$C$};
      \draw (-0.6, -0.75) rectangle (2.5, 0.75);
      
      \node at (11, 0) {$B$};
      \draw (7.5, -0.75) rectangle (10.6, 0.75);

    \end{scope}
    
    \begin{scope}[yshift=-3cm]
      \Vertex[x=0, y=0, Lpos=180, Math, L={v_1}]{v1}
      \Vertex[x=2, y=0, Lpos=-90, Math, Ldist=0.2cm, L={v_2}]{v2}
      \Vertex[x=4, y=0, Lpos=-90, Math, Ldist=0.2cm, L={v_3}]{v3}
      \Vertex[x=6, y=0, Lpos=-90, Math, Ldist=0.2cm, L={v_4}]{v4}
      \Vertex[x=8, y=0, Lpos=-90, Math, Ldist=0.2cm, L={v_5}]{v5}
      \Vertex[x=10, y=0, Lpos=0, Math, Ldist=0.2cm, L={v_6}]{v6}

      \node at (-1, 0) {$A$};
      \draw (-0.6, -0.75) rectangle (2.5, 0.75);

      \Edges[style={-Latex}](v6,v5,v4,v3,v2,v1)
        \begin{scope}
          \tikzset{VertexStyle/.append style = {shape = rectangle, inner sep = 2.5pt}}
          
          \AddVertexColor{white}{u1,u6}
          \AddVertexColor{gray!50}{v1,v6}
        \end{scope}
    \end{scope}

      \foreach \i in {1,2,3,4,5,6} {
        \Edges[style={-Latex}](u\i,v\i)
      }
      \foreach \i in {1,2,3,4,5} {
        \pgfmathsetmacro{\b}{\i+1}
        \foreach \j in {\b,...,6} {
            \Edges[style={-Latex,opacity=0.1}](v\j,u\i)
            \Edges[style={-Latex,opacity=0.1}](v\i,u\j)
        }
      }
      \foreach \i in {1,2,3,4} {
        \pgfmathsetmacro{\b}{\i+2}
        \foreach \j in {\b,...,6} {
            \Edges[style={-Latex,opacity=0.1, bend right=40}](u\j,u\i)
            \Edges[style={-Latex,opacity=0.1, bend right=40}](v\i,v\j)
        }
      }
      
    \Edges[style={-Latex, bend right=40}](u5,u2)
    \Edges[style={-Latex, bend right=40}](u5,u1)
    \Edges[style={-Latex, bend right=40}](u6,u2)
    \Edges[style={-Latex, bend right=40}](u6,u1)
    \Edges[style={-Latex}](v1,u5)
    \Edges[style={-Latex}](v1,u6)
    \Edges[style={-Latex}](v2,u5)
    \Edges[style={-Latex}](v2,u6)
  \end{tikzpicture}
  \caption{Counterexample for $n = 2$, where non-black vertices of the same color correspond to endpoints of a same request. The rectangles indicate the three sets that make up the $n$-triple.\label{fig:counterexample}}
\end{figure}
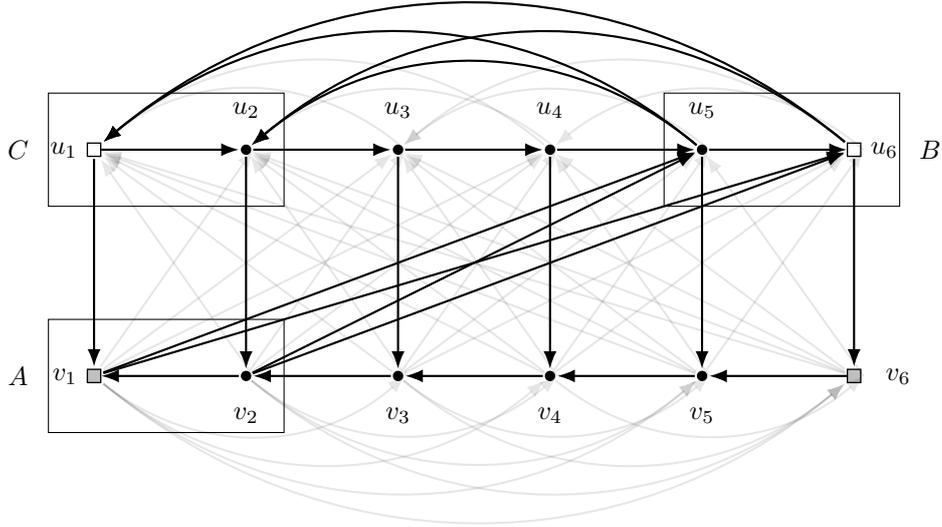

We refer to \autoref{fig:counterexample} for the case $n = 2$. Observe that the requests can only be satisfied by using $P_u$ and $P_v$, $c$ times each. This statement follows by induction on $n$: it suffices to observe that $u_2$ is the unique vertex that may be the second vertex in the $(u_1, u_{2n+2})$-satisfying paths, while $v_2$ is the unique vertex that may be the penultimate vertex in the paths that satisfy $(v_{2n+2}, v_1)$; at this point we can discard $u_1, v_1$ and repeat the analysis.
We can further restrict ourselves to congestion values smaller than $k/2$. For example, if we add $\tau$ copies of $T_n$ and add all arcs from the $i$-th copy to the $j$-th whenever $i < j$, then we can get $\congestion = k/(2\tau)$; it is not hard to come up with other strategies that would yield $\congestion = k/d$ for every $d \in [2, k]$.
At this point, it is natural to think about the interval $(1, 2)$? Interestingly, things break down in this case.
Suppose that, instead of having the same number of requests between the endpoints of $P_u$ and $P_v$, we had two requests $(u_1, u_{2n+2})$, only \textit{one} request for $(v_{2n+2}, v_1)$, and congestion $\congestion = 2$. We could use $P_u$ twice and $P_v$ once, but we could also do the following: begin by satisfying $(v_{2n+2}, v_1)$ using $P_v$; now, add the paths $\angled{u_1, v_1, u_6}$ and $\angled{u_1, u_2, v_2, u_{2n+2}}$ to satisfy the $(u_1, u_{2n+2})$ requests. With this, we have avoided the usage of every vertex in $\{u_3, \dots, u_{2n+1}\}$, making them \textit{irrelevant} to the instance; we could also have avoided using vertices in $V$ using extra steps. 

%% file: sections/hardnesses.tex
\section{Hardness results}
\label{sec:hardnesses}

  While it had been assumed for a long time that \pname{Directed Disjoint Paths} was an \NP-hard problem on tournaments~\cite{chudnovsky_union_of_tournaments,Slivkins2010,KawarabayashiKK14,Amiri2019}, we verified that there is a mistake in the proof given
  in~\cite{thomassen_tournament_nph}. We have reached out to its authors, which have confirmed the flaw in their proof in a
  personal communication~\cite{personal-communication}.
  In this section, we provide a direct \NP-hardness proof for tournaments as well as other lower bounds the $C_h$ superclasses.

\input{sections/tournaments}
\input{sections/ch}

%% file: sections/tournaments.tex
\subsection{Tournaments}
\label{sec:tournament_nph}
  Fortunately, as we show in the following theorem, \pname{$k$-DDP} is
  indeed \NP-hard on tournaments.
  We highlight that our reduction is completely different
  from the approach in~\cite{thomassen_tournament_nph}.
  In particular,
  we reduce from \pname{(3,1)-3-SAT}, a variant of \pname{3-SAT} where each variable appears
  exactly four times -- once negated and three times unnegated -- and that was shown to be
  \NP-complete in~\cite{darmann_monotone_sat}.

  \medskip

\noindent\textbf{Construction.} Let $(X, \mathcal{C})$ be   the input instance to
  \pname{(3,1)-3-SAT}, where $X$ is the set of variables and $\mathcal{C}$ is the set of clauses, and $(T, K)$ be the \pname{$k$-DDP} instance we are
  going to build, with $T$ denoting the digraph, and $K$ denoting the set of requests.
  We begin by picking an arbitrary but fixed order $\angled{x_1, \dots, x_n}$ of $X$.
  Now, for each variable $x_i \in X$, add a copy of the \textit{directed
  butterfly gadget} shown in \autoref{fig:directed_butterfly} to $T$, while $(s_i, t_i)$
  and $(\ol{s}_i, \ol{t}_i)$ are added to $K$; we denote this gadget by $B_i$.
  Intuitively, the paths between the white vertices of $B_i$ must be contained in it
  and, since the gadget has only once center (namely, $\alpha_i$ in
  \autoref{fig:directed_butterfly}), at most one of them may avoid the ``wings'', as shown in \autoref{obs:busy_wing}. This
  allows us to encode the assignment of a variable and only satisfy clauses for which the
  appropriate literal is true, i.e., adding $\angled{s_i, \alpha_i, t_i}$ to the solution is equivalent to setting $x_i = \truet$.

  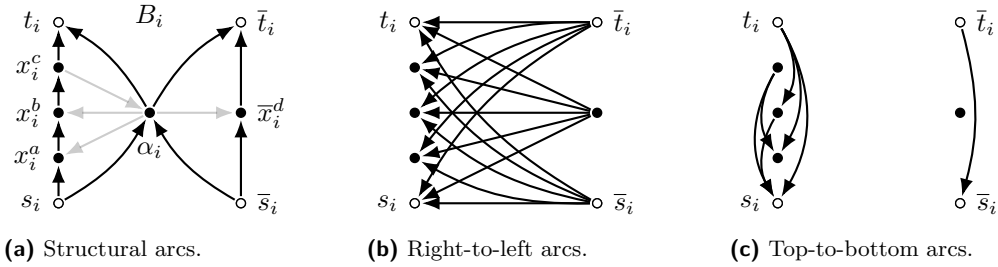
\begin{figure}[!htb]
    \centering
    \begin{subfigure}[t]{0.33\textwidth}
        \begin{tikzpicture}[scale=0.8]
          \GraphInit[unit=3,vstyle=Normal]
          \SetVertexNormal[Shape=circle, FillColor=black, MinSize=2pt]
          \tikzset{VertexStyle/.append style = {inner sep = \inners, outer sep = \outers}}
          \SetVertexLabelOut
          \Vertex[x=0, y=0, Lpos=-90, Math, Ldist=5pt, L={\alpha_i}]{c}
    
          \Vertex[x=1.5, y=1.5, Lpos=0, Math, L={\ol{t}_i}]{tm_i}
          \Vertex[x=-1.5, y=1.5, Lpos=180, Math]{t_i}
    
          \Vertex[x=1.5, y=-1.5, Lpos=0, Math, L={\ol{s}_i}]{sm_i}
          \Vertex[x=-1.5, y=-1.5, Lpos=180, Math]{s_i}
    
          \Vertex[x=-1.5, y=0.75, Lpos=180, Math, L={x^c_i}]{x_i3}
          \Vertex[x=-1.5, y=0, Lpos=180, Math, L={x^b_i}]{x_i2}
          \Vertex[x=-1.5, y=-0.75, Lpos=180, Math, L={x^a_i}]{x_i1}
    
          \Vertex[x=1.5, y=0, Lpos=0, Math, L={\ol{x}^d_i}]{xm_i1}
    
          \Edges[style={-Latex, bend left=15}](sm_i, c, tm_i)
          \Edges[style={-Latex}](sm_i, xm_i1, tm_i)
    
          \Edges[style={-Latex, bend right=15}](s_i, c, t_i)
          \Edges[style={-Latex}](s_i, x_i1, x_i2, x_i3, t_i)
    
          \Edges[style={-Latex, opacity=0.2}](c,xm_i1)
          \Edges[style={-Latex, opacity=0.2}](c,x_i1)
          \Edges[style={-Latex, opacity=0.2}](c,x_i2)
          \Edges[style={-Latex, opacity=0.2}](x_i3,c)
    
          \AddVertexColor{white}{tm_i, t_i, sm_i, s_i}
    
          \node at (0, 1.6) {$B_i$};
        \end{tikzpicture}    
        \subcaption{Structural arcs.}
    \end{subfigure}%
    ~
    \begin{subfigure}[t]{0.33\textwidth}
        \begin{tikzpicture}[scale=0.8]
          \GraphInit[unit=3,vstyle=Normal]
          \SetVertexNormal[Shape=circle, FillColor=black, MinSize=2pt]
          \tikzset{VertexStyle/.append style = {inner sep = \inners, outer sep = \outers}}
          \SetVertexLabelOut
    
          \Vertex[x=1.5, y=1.5, Lpos=0, Math, L={\ol{t}_i}]{tm_i}
          \Vertex[x=-1.5, y=1.5, Lpos=180, Math]{t_i}
    
          \Vertex[x=1.5, y=-1.5, Lpos=0, Math, L={\ol{s}_i}]{sm_i}
          \Vertex[x=-1.5, y=-1.5, Lpos=180, Math]{s_i}

          \SetVertexNoLabel
          \Vertex[x=-1.5, y=0.75, Lpos=180, Math, L={x^c_i}]{x_i3}
          \Vertex[x=-1.5, y=0, Lpos=180, Math, L={x^b_i}]{x_i2}
          \Vertex[x=-1.5, y=-0.75, Lpos=180, Math, L={x^a_i}]{x_i1}
    
          \Vertex[x=1.5, y=0, Lpos=0, Math, L={\ol{x}^d_i}]{xm_i1}

          \foreach \j in {1,2,3} {
            \Edges[style={-Latex, bend right=15}](tm_i, x_i\j)
            \Edges[style={-Latex, bend left=15}](sm_i, x_i\j)
          }
          
          \Edges[style={-Latex, bend right=15}](tm_i, s_i)
          \Edges[style={-Latex, bend left=15}](sm_i, t_i)
          \Edges[style={-Latex}](sm_i, s_i)
          \Edges[style={-Latex}](tm_i, t_i)
          
          \Edges[style={-Latex}](xm_i1, x_i3)
          \Edges[style={-Latex}](xm_i1, t_i)
          \Edges[style={-Latex}](xm_i1, x_i2)
          \Edges[style={-Latex}](xm_i1, x_i1)
          \Edges[style={-Latex}](xm_i1, s_i)
          
          \AddVertexColor{white}{tm_i, t_i, sm_i, s_i}
        \end{tikzpicture}
        \subcaption{Right-to-left arcs.}
    \end{subfigure}%
    ~
    \begin{subfigure}[t]{0.33\textwidth}
        \begin{tikzpicture}[scale=0.8]
          \GraphInit[unit=3,vstyle=Normal]
          \SetVertexNormal[Shape=circle, FillColor=black, MinSize=2pt]
          \tikzset{VertexStyle/.append style = {inner sep = \inners, outer sep = \outers}}
          \SetVertexLabelOut
    
          \Vertex[x=1.5, y=1.5, Lpos=0, Math, L={\ol{t}_i}]{tm_i}
          \Vertex[x=-1.5, y=1.5, Lpos=180, Math]{t_i}
    
          \Vertex[x=1.5, y=-1.5, Lpos=0, Math, L={\ol{s}_i}]{sm_i}
          \Vertex[x=-1.5, y=-1.5, Lpos=180, Math]{s_i}

          \SetVertexNoLabel
          \Vertex[x=-1.5, y=0.75, Lpos=180, Math, L={x^c_i}]{x_i3}
          \Vertex[x=-1.5, y=0, Lpos=180, Math, L={x^b_i}]{x_i2}
          \Vertex[x=-1.5, y=-0.75, Lpos=180, Math, L={x^a_i}]{x_i1}
    
          \Vertex[x=1.5, y=0, Lpos=0, Math, L={\ol{x}^d_i}]{xm_i1}
          
          \Edges[style={-Latex, bend left}](t_i, x_i2)
          \Edges[style={-Latex, bend left}](t_i, x_i1)
          \Edges[style={-Latex, bend left}](t_i, s_i)
          \Edges[style={-Latex, bend right}](x_i3, x_i1)
          \Edges[style={-Latex, bend right}](x_i3, s_i)
          \Edges[style={-Latex, bend right}](x_i2, s_i)
          
          \Edges[style={-Latex, bend left=15}](tm_i, sm_i)
          
          \AddVertexColor{white}{tm_i, t_i, sm_i, s_i}
        \end{tikzpicture}    
        \subcaption{Top-to-bottom arcs.}
    \end{subfigure}
    \caption{Directed butterfly gadget for variable $x_i$, occurring negated in clause $C_d$
      and unnegated in clauses $C_a, C_b,$ and $C_c$.
      White vertices represent terminals.
      The given orientations of the gray arcs will be 
      important when talking about congested versions of the problem. Vertex
      $\alpha_i$ is the \textit{center} of $B_i$, while the two disjoint paths from
      the $s$ vertices to the $t$ vertices that do not use $\alpha_i$ are its \textit{wings}.}
    \label{fig:directed_butterfly}
  \end{figure}

  After building all $n$ butterflies, add an arc from every $u \in B_j$ to every $v \in B_i$ whenever $j > i$.
  To encode our clauses, take each $C_a \in \mathcal{C}$, add new vertices $p_a, q_a$ to $T$,
  the pair $(p_a, q_a)$ to $K$, and the arc $(q_a, p_a)$; arcs between each $p_a$ and $q_b$,
  for $a \neq b$, can be added arbitrarily.
  Now, take $a \in \{1, \dots, m\}$ as an example, and suppose that the $a$-th clause of $\mathcal{C}$ is $C_a = (x_1 \vee x_2 \vee \ol{x}_3)$; in this case, we add
  arcs from $p_a$ to the vertices $x^a_1, x^a_2$, and $\ol{x}^a_3$, and from these to $q_a$.
  For simplicity, the superscript $j$ of
  each $x^j_i$ and $\ol{x}^j_i$ is the same as the subscript of the corresponding clause $C_j$.
  At this point, the only missing arcs to ensure we have a tournament are between vertices of butterfly gadgets and clause
  gadgets; we add them such that $p_a$ has no other outgoing arc and $q_a$ has no other
  incoming arc. This concludes the construction of $(T, K)$, with $K = \{(s_i, t_i), (\ol{s}_i, \ol{t}_i) \mid x_i \in X\} \cup \{(p_a, q_a) \mid C_a \in \mathcal{C}\}$. For the above example, our goal is to witness the satisfiability of $C_a$ with a path $\angled{p_a, y, q_a}$, where $y \in \{x^a_1, x^a_2, \ol{x}^a_3\}$, which will have to be the case as $N^+(p_a) = N^-(q_a) = \{x^a_1, x^a_2, \ol{x}^a_3\}$.

  \begin{lemma}
    \label{lem:forward_tournament_nph}
    If there is a satisfying assignment $\pi$ for $(X, \mathcal{C})$, there is a family
    of disjoint paths $\mathcal{P}$ satisfying $(T, K)$.
  \end{lemma}

  \begin{proof}
    We break our analysis in two. First, for each variable $x_i \in X$, if $\pi(x_i) = 0$,
    then add the paths $\angled{\ol{s}_i, \alpha_i, \ol{t}_i}$ and
    $\angled{s_i, x^a_1, x^b_1, x^c_1, t_i}$ to $\mathcal{P}$; otherwise add $\angled{s_i,
    \alpha_i, t_i}$ and $\angled{\ol{s}_i, \ol{x}^d_i, \ol{t}_i}$ to $\mathcal{P}$.
    Now, let $C_a$ be a clause and $\ell_i$ one of its satisfying literals.
    If $\ell_i = \ol{x}_i$, then $\pi(x_i) = 0$ and, by our choices in $B_i$, vertex
    $\ol{x}^a_i$ is not used by any path in $\mathcal{P}$, so we may now add
    $\angled{p_a, \ol{x}_i, q_a}$ to $\mathcal{P}$.
    A similar choice can be performed if $\ell_i = x_i$; note that, for each vertex in the left
    wing of $B_i$, $a$ is the unique index such that $(p_a, x^a_i), (x^a_i, q_a) \in E(T)$,
    so no other clause gadget could have used $x^a_i$ in its satisfying path.
  \end{proof}

  \begin{observation}
    \label{obs:busy_wing}
    Let $B_i$ be one of the butterfly gadgets of $(T, K)$. If $(T,K)$ admits a
    solution $\mathcal{P}$, then (i) the paths that satisfy $(\ol{s}_i, \ol{t}_i)$ and
    $(s_i, t_i)$ are internal to $B_i$, and (ii) at least one of the wings of $B_i$ is
    entirely occupied by such a path.
  \end{observation}

  \begin{proof}
    Let $\ol{P}_i, P_i \in \mathcal{P}$ be the paths connecting $(\ol{s}_i, \ol{t}_i)$ and
    $(s_i, t_i)$, respectively. Suppose that $\mathcal{P}$ is minimal, i.e., none of its
    paths has an internal shortcut; for example, $\angled{s_i, \alpha_i, \ol{t}_i, t_i} \notin \mathcal{P}$ because arc $(\alpha_i, t_i)$ has both endpoints contained in the path and could be used to shorten it.
    Let us first show that $\ol{P}_i \subseteq B_i$.
    Recall that $N^-(\ol{t}_i) \subseteq \{\ol{x}^d_i, \alpha_i\} \cup \{v \mid v \in
    B_j, j > i\}$ and that no terminal of
    $V(K)$ can be used as an intermediate vertex of a path,
    as they must be saturated by their own paths, so the only arcs incident to $\ol{t}_i$
    that may be used to reach it are the ones we have just listed.
    As such, the penultimate vertex of $\ol{P}_i$ is either one of $\{\ol{x}^d_i, \alpha_i\}$,
    in which case we are done as $\{\ol{x}^d_i, \alpha_i\} \subseteq N^+(\ol{s}_i)$ and
    $\mathcal{P}$ is minimal, or it is in some $B_j$ with $j > i$, but this is impossible
    as there is no path from $\ol{s}_i$ to such a vertex in $G \setminus V(K)$ by design.
    A symmetric analysis can be performed to show that $P_i \subseteq B_i$:
    if $\alpha_i \in P_i$, we are done as $\alpha_i \in N^+(s_i) \cap N^-(t_i)$, so suppose $\alpha_i \notin P_i$.
    Note that $B_j$, with $j > i$ is unreachable from $s_i$ in $G \setminus V(K)$, and so we have that $P_i \cap B_j = \emptyset$.
    Moreover, for $\ell < i$, $t_i$ is unreachable from any vertex of $B_\ell$ in $G \setminus V(K)$.
    Finally, every path between $s_i$ and the right wing of $B_i$ in $G \setminus V(K)$ passes through $\alpha_i$, which is not in $P_i$, so the right wing of $B_i$ does not intersect $P_i$.
    Consequently, if $\alpha_i \notin P_i$, then $P_i \subseteq B_i \setminus \{\alpha_i, \ol{s}_i,\ol{x}^d_i,\ol{t_i}\}$, which is precisely the left wing of $B_i$, completing the proof of (i).    
    Towards proving (ii), if $P_i = \angled{s_i, x_i^a, x_i^b, x_i^c, t_i}$ we are
    done, so suppose that $P_i = \angled{s_i, \alpha_i, t_i}$; by (i), this implies that
    $\ol{P}_i = \angled{\ol{s}_i, \ol{x}^d_i, \ol{t}_i}$.
    A completely symmetric analysis can be performed if $\ol{P}_i = \angled{\ol{s}_i, \alpha_i,
    \ol{t}_i}$ to conclude that $P_i = \angled{s_i, x_i^a, x_i^b, x_i^c}$.
  \end{proof}

  \begin{lemma}
    \label{lem:backward_tournament_nph}
    If there is a family of disjoint paths $\mathcal{P}$ satisfying $(T, K)$, there is
    a satisfying assignment $\pi$ for $(X, \mathcal{C})$.
  \end{lemma}

  \begin{proof}
    Let $\ol{P}_i, P_i \in \mathcal{P}$ be the paths satisfying $(\ol{s}_i, \ol{t}_i)$ and
    $(s_i, t_i)$, respectively.
    By \autoref{obs:busy_wing}, at least one of these two paths occupies a wing.
    If $\ol{P}_i = \angled{\ol{s}_i, \ol{x}^d_i, \ol{t}_i}$, we set $\pi(x_i) = 1$,
    otherwise set $\pi(x_i) = 0$.
    To see that this is a satisfying assignment, take $C_a \in \mathcal{C}$ and the
    corresponding path $P_a = \angled{p_a, v_i, q_a}$, such that $v_i \in B_i$.
    If $v_i = x_i^a$, then $x_i \in C_a$ and $x_i^a \notin P_i$ as $P_i \cap P_a =
    \emptyset$, but this means that $P_i = \angled{s_i, \alpha_i, t_i}$ as this is the unique
    path that can satisfy the request $(s_i, t_i)$ and not use $x_i^a$ by \autoref{obs:busy_wing}.
    The latter in turn implies that $\ol{P}_i = \angled{\ol{s}_i, \ol{x}^d_i, \ol{t}_i}$,
    $\pi(x_i) = 1$, and so $C_a$ is satisfied.
    If, however, $v_i = \ol{x}^a_i$, then $\ol{x}_i \in C_a$, and $\ol{P}_i =
    \angled{\ol{s}_i, \alpha_i, \ol{t}_i}$, in which case we have set $\pi(x_i) = 0$,
    satisfying $C_a$.
  \end{proof}

  \tournamentnph*

  \begin{proof}
    This follows immediately from the above construction and lemmas, as well as the
    straightforward algorithm that, given a solution $\mathcal{P}$ to $(T,K)$, checks
    whether $\mathcal{P}$ satisfies the requests $K$ and its paths are pairwise disjoint.
  \end{proof}

  \subsection{Tournaments with congestion}

  As our main focus for this work is on congested variants of \pname{Directed Disjoint
  Paths}, it is natural to extend the proof of \autoref{thm:tournament_nph} to these
  variants as well. In particular, we are interested in the case where the congestion $\congestion = |K|/d$.

  \subsubsection{Restricted $K$}

  We first consider a variant that we call \pname{Restricted $(k,c)$-DDP}, where we are additionally forbidden from using vertices of $K$ as inner
  vertices of other paths; that is, if $(u,v) \in K$, then a path $P$ has $u,v \in P$ if and
  only if $u$ is the first vertex of $P$ and $v$ is the last. Note that if $\congestion = 1$,
  this is precisely what we have in the standard \pname{$k$-DDP} problem.

  \noindent\textbf{Meta-construction.} We reduce from the instance $(T, K)$ of \pname{$k$-DDP} built in the proof of \autoref{thm:tournament_nph} to construct the
  instance $(T', K', \congestion)$ of \pname{Restricted $(k',c)$-DDP}. Let
  $\angled{B_1, \dots, B_n}$ be the ordering of the butterfly gadgets, that is, $B_j$ is
  complete to $B_i$ for all $j > i$.
  We perform the following modification to $T$: for each $i \in [n-1]$, reverse the arc
  $(x^a_{i+1}, \ol{x}^d_i)$, where $x^a_{i+1}$ is the unique out-neighbor of $s_{i+1}$ in
  the left wing of $B_{i+1}$.
  To conclude the construction of $T'$, we add two vertices $s^*, t^*$ and arcs such that
  $x^a_1$ is the unique out-neighbor of $s^*$ and $\ol{x}^d_n$ is the unique in-neighbor of $t^*$.
  Finally, we add $\congestion-1$ requests $(s^*, t^*)$ and one request $(t^*, s^*)$ to
  $K$, thus obtaining $K'$.
  Note that $\congestion$ is \textit{almost} arbitrary, as it can be set to a constant or whichever
  \textit{fraction} of $|K'|$ that we wish.
  Not every element of $o(|K'|)$, however, is assignable to $\congestion$; for example,
  $\congestion = |K'| -
  \log |K'|$ is not a viable candidate, as this would imply that $|K'|$ is exponential in
  $|K|$ since $|K'| = |K| + \congestion$.
  We illustrate the resulting tournament $T'$ in \autoref{fig:queued_butterfly}.

  \begin{figure}[!htb]
    \centering
    \begin{tikzpicture}[scale=0.8]
      \GraphInit[unit=3,vstyle=Normal]
      \SetVertexNormal[Shape=circle, FillColor=black, MinSize=2pt]
      \tikzset{VertexStyle/.append style = {inner sep = \inners, outer sep = \outers}}
      \SetVertexLabelOut
      \begin{scope}
        \Vertex[x=-2.5, y=-0.75, Lpos=180, Math, L={s^*}]{ss}
        \SetVertexNoLabel

        \Vertex[x=0, y=0, Lpos=-90, Math, Ldist=5pt, L={\alpha_1}]{c}

        \Vertex[x=1.5, y=1.5, Lpos=90, Math, L={\ol{t}_1}]{tm_1}
        \Vertex[x=-1.5, y=1.5, Lpos=90, Math]{t_1}

        \Vertex[x=1.5, y=-1.5, Lpos=-90, Math, L={\ol{s}_1}]{sm_1}
        \Vertex[x=-1.5, y=-1.5, Lpos=-90, Math]{s_1}

        \Vertex[x=-1.5, y=0.75, Lpos=180, Math, L={x^c_1}]{x_13}
        \Vertex[x=-1.5, y=0, Lpos=180, Math, L={x^b_1}]{x_12}
        \Vertex[x=-1.5, y=-0.75, Lpos=180, Math, L={x^a_1}]{x_11}

        \Vertex[x=1.5, y=0, Lpos=0, Math, L={\ol{x}^d_1}]{xm_11}

        \Edges[style={-Latex, bend left=15, opacity=0.2}](sm_1, c, tm_1)
        \Edges[style={-Latex, opacity=0.2}](sm_1, xm_11, tm_1)

        \Edges[style={-Latex, bend right=15, opacity=0.2}](s_1, c, t_1)
        \Edges[style={-Latex, opacity=0.2}](s_1, x_11)
        \Edges[style={-Latex, opacity=0.2}](x_13, t_1)
        \Edges[style={-Latex, opacity=0.2}](c,x_11)
        \Edges[style={-Latex, opacity=0.2}](c,x_12)

        \Edges[style={-Latex}](ss, x_11, x_12, x_13, c, xm_11)

        \AddVertexColor{white}{tm_1, t_1, sm_1, s_1, ss}
      \end{scope}

      \begin{scope}[xshift=4.5cm]
        \SetVertexNoLabel
        \Vertex[x=0, y=0, Lpos=-90, Math, Ldist=5pt, L={\alpha_2}]{c}

        \Vertex[x=1.5, y=1.5, Lpos=90, Math, L={\ol{t}_2}]{tm_2}
        \Vertex[x=-1.5, y=1.5, Lpos=90, Math]{t_2}

        \Vertex[x=1.5, y=-1.5, Lpos=-90, Math, L={\ol{s}_2}]{sm_2}
        \Vertex[x=-1.5, y=-1.5, Lpos=-90, Math]{s_2}

        \Vertex[x=-1.5, y=0.75, Lpos=180, Math, L={x^c_2}]{x_23}
        \Vertex[x=-1.5, y=0, Lpos=180, Math, L={x^b_2}]{x_22}
        \Vertex[x=-1.5, y=-0.75, Lpos=180, Math, L={x^a_2}]{x_21}

        \Vertex[x=1.5, y=0, Lpos=0, Math, L={\ol{x}^d_2}]{xm_21}

        \Edges[style={-Latex, bend left=15, opacity=0.2}](sm_2, c, tm_2)
        \Edges[style={-Latex, opacity=0.2}](sm_2, xm_21, tm_2)

        \Edges[style={-Latex, bend right=15, opacity=0.2}](s_2, c, t_2)
        \Edges[style={-Latex, opacity=0.2}](s_2, x_21)
        \Edges[style={-Latex, opacity=0.2}](x_23, t_2)
        \Edges[style={-Latex, opacity=0.2}](c,x_21)
        \Edges[style={-Latex, opacity=0.2}](c,x_22)

        \Edges[style={-Latex}](xm_11, x_21, x_22, x_23, c, xm_21)

        \AddVertexColor{white}{tm_2, t_2, sm_2, s_2}
      \end{scope}

      \begin{scope}[xshift=9cm]
        \Vertex[x=2.5, y=-0.75, Lpos=0, Math, LabelOut, L={t^*}]{ts}

        \SetVertexNoLabel
        \Vertex[x=0, y=0, Lpos=-90, Math, Ldist=5pt, L={\alpha_3}]{c}

        \Vertex[x=1.5, y=1.5, Lpos=90, Math, L={\ol{t}_3}]{tm_3}
        \Vertex[x=-1.5, y=1.5, Lpos=90, Math]{t_3}

        \Vertex[x=1.5, y=-1.5, Lpos=-90, Math, L={\ol{s}_3}]{sm_3}
        \Vertex[x=-1.5, y=-1.5, Lpos=-90, Math]{s_3}

        \Vertex[x=-1.5, y=0.75, Lpos=180, Math, L={x^c_3}]{x_33}
        \Vertex[x=-1.5, y=0, Lpos=180, Math, L={x^b_3}]{x_32}
        \Vertex[x=-1.5, y=-0.75, Lpos=180, Math, L={x^a_3}]{x_31}

        \Vertex[x=1.5, y=0, Lpos=0, Math, L={\ol{x}^d_3}]{xm_31}

        \Edges[style={-Latex, bend left=15, opacity=0.2}](sm_3, c, tm_3)
        \Edges[style={-Latex, opacity=0.2}](sm_3, xm_31, tm_3)

        \Edges[style={-Latex, bend right=15, opacity=0.2}](s_3, c, t_3)
        \Edges[style={-Latex, opacity=0.2}](s_3, x_31)
        \Edges[style={-Latex, opacity=0.2}](x_33, t_3)
        \Edges[style={-Latex, opacity=0.2}](c,x_31)
        \Edges[style={-Latex, opacity=0.2}](c,x_32)

        \Edges[style={-Latex}](xm_21, x_31, x_32, x_33, c, xm_31, ts)

        \AddVertexColor{white}{tm_3, t_3, sm_3, s_3, ts}
      \end{scope}

    \end{tikzpicture}
    \caption{Queuing of the butterfly gadgets performed for the construction of $T'$. Gray
      arcs do not participate in the critical path of $T'$, while black arcs do. Again,
      within the butterfly gadgets, we add the missing arcs from top-to-bottom,
      right-to-left. All other arcs are also right-to-left.
    \label{fig:queued_butterfly}}
  \end{figure}
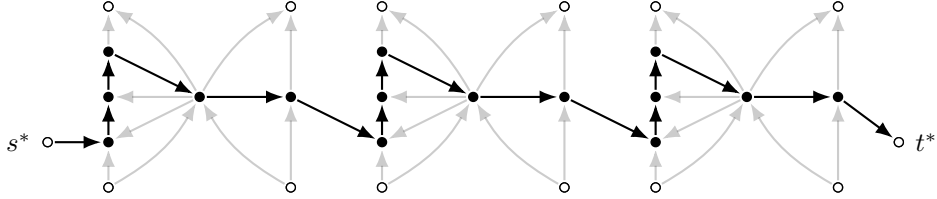

  Observe that the interior vertices of the butterflies' wings are the only vertices of
  $T'$ that do not participate in any element of $K'$.
  Moreover, by construction, they all participate in the unique $s^*$-$t^*$ path of $T'
  \setminus V(K' \setminus \{(s^*, t^*)\})$ in the same order specified by $\angled{B_1,
  \dots, B_n}$.
  We call this path the \textit{critical path} of $T'$, and in any solution to $(T', K',
  \congestion)$, all of its inner vertices participate in at least $\congestion-1$ paths;
  cf. the black arcs in \autoref{fig:queued_butterfly}.

  \begin{theorem}
    \label{thm:restricted_ddp_nph}
    \pname{Restricted $(k',c)$-DDP} remains \NP-complete even when
    restricted to tournaments and when $\congestion = \varepsilon k$ for every $\varepsilon \in (0,1)$. 
  \end{theorem}

  \begin{proof}
    Given the desired fraction $\congestion = |K'|/d$, we begin by restricting ourselves
    to instances of
    \pname{$k$-DDP} such that $|K| - 1$ is divisible by $d - 1$, and then
    constructing $(T', K', \congestion = |K'|/d)$ as above.
    Now, note that, since the critical path of $T'$ must be occupied by $\congestion-1$ paths from
    $s^*$ to $t^*$, the remaining requests of $K'$ in $T'$ are exactly the same requests of
    $K$ in $T$ with the exact same availability for each vertex in $V(T') \cap V(T)$:
    terminals of $K$ can only be used by themselves in both problems, while each non-terminal
    can either be occupied by one request between butterfly terminals or by one request
    from outside a butterfly gadget.
    As such, it follows from \autoref{thm:tournament_nph} that $(T, K)$ admits a
    solution if and only if $(T', K', \congestion)$ admits a solution.
  \end{proof}

  \subsubsection{Unrestricted $K$}

  While we are unable to prove a statement as strong as
  \autoref{thm:restricted_ddp_nph} for \pname{$(k,c)$-DDP}, we
  are able to (conditionally) rule out polynomial-time algorithms for some choices of $\congestion$;
  in particular, if $\congestion \leq |K|^\varepsilon$ for $\varepsilon \in [0,1)$ we can show
  such a lower bound.

  \epsilonddp*

  \begin{proof}
    The construction is simple: let $(T', K', \congestion')$ be an instance of \pname{Restricted $(k',c')$-DDP} obtained in the proof of
    \autoref{thm:restricted_ddp_nph} and $(T, K, \congestion)$ be the instance of
    \pname{$(k,c)$-DDP} we are building.
    We initially set $T \gets T'$, $K \gets K'$, and $\congestion \gets \congestion'$.
    Recall that there are exactly $\congestion$ requests with $s^*, t^*$ as both of its
    endpoints and all other terminals are endpoints of exactly one request.
    Now, let $X$ be the set of all requests of the latter type, and, for each $(u,v) \in
    X$, add $\congestion - 1$ requests of the form $(\bm{v,u})$ to $K$.
    As such, we have that $|K| = |K'| + |X|(\congestion-1) = (\congestion + |X|) +
    |X|(\congestion-1) = (|X| + 1)\congestion$ and, consequently, $|K|/\congestion =
    |K|^{1-\varepsilon} = |X| + 1$ is polynomial in $|X|$.

    The correctness of the reduction follows immediately from the fact that each terminal
    in $(T, K, \congestion)$ must be completely occupied by paths that have it as an
    endpoint, and
    so all that remains to be solved is the original \pname{Restricted $(k', c')$-DDP} instance.
  \end{proof}

  We would like to point out that overcoming the $|K|^\varepsilon$ barrier has proved extremely
  challenging for us, and seems to require a very different approach to the one we employed so far.
  In particular, we can rearrange the butterfly gadgets to be \textit{stacked}
  instead of queued, i.e., we glue the terminals of two butterflies, avoiding the need of
  too many paths to saturate them. This, however, is not enough, as we are unable to
  saturate the clause gadgets with few paths, instead always requiring
  $\bigO{|\mathcal{C}|}$ requests, as done in \autoref{thm:epsilon_ddp}.
  Interestingly, we can prove a weak \NP-hardness result for larger values of the
  congestion, namely if $\congestion = |K|\log^{-\varepsilon} |K|$.

  \begin{theorem}
    \label{thm:log_ratio_ddp}
    \pname{$(k,c)$-DDP} is weakly \NP-hard on tournaments if
    the congestion is of the form $|K|\log^{-\varepsilon} |K|$, for every $\varepsilon > 0$.
  \end{theorem}

  \begin{proof}
    The construction and correctness arguments are precisely the same as the ones in the
    proof of \autoref{thm:epsilon_ddp}. Instead of opting for an explicit
    representation of $K$, we instead encode it as map $K: V(T) \times V(T) \mapsto
    \mathbb{N}$, where $K(u,v)$ represents the number of requests from vertex $u$ to vertex $v$.
    Since $\congestion = |K|\log^{-\varepsilon} |K|$, it follows that $\log^\varepsilon |K| =
    |X| + 1$ and, consequently, each $K(u,v)$ can be encoded in $(|X| + 1)^{1/\varepsilon}$ bits.
  \end{proof}

%% file: sections/ch.tex
\subsubsection{\NP-hardness for graphs in $\mathcal{C}_2$}
\label{sec:ch_nph}
Before proceeding, recall that $G \in \mathcal{C}_h$ if and only if the vertices of its underlying graph can be partitioned by at most $h$ cliques, i.e., its complement has chromatic number at most $h$.
Our construction is similar to the one employed on \autoref{thm:restricted_ddp_nph}
and yields an equivalent statement, however there are some key differences, so we
do not use it as directly as in \autoref{thm:epsilon_ddp}.
As such, we describe our reduction directly from \pname{(3,1)-3-SAT}.
Intuitively, we will use the flexibility of some non-arcs to extend the critical path with
some sinks of the clause requests; at the same time, we can replace the $\bigO{|T|}$ sources
with a single vertex.

\noindent\textbf{Construction} Let once again $(X, \mathcal{C})$ be the given
\pname{(3,1)-3-SAT} instance and take an arbitrary ordering $\angled{C_1, \dots, C_m}$ of
$\mathcal{C}$.
We first show how to construct an instance $(G, K, \congestion)$ of \pname{Directed
Congested Disjoint Paths} with $\congestion = |\mathcal{C}|$, and subsequently show how
to extend it to so $\congestion$ is an arbitrary ratio of $K$.
Initially, add the queued structure of butterflies shown in
\autoref{fig:queued_butterfly} to $G$ and the corresponding requests to $K$, but now
we perform some modifications to it:
(i) take the unique out-neighbor $v$ of $s^*$ and subdivide the arc $(s^*, v)$ exactly
$|\mathcal{C}|$ times, assigning to each generated $q_a$ a unique clause $C_a$;
(ii) for each $C_a$, add an arc from $x_i^a$ to $q_a$ if  $x_i \in C_a$ and from
$\ol{x}_i^a$ to $q_a$ if $\ol{x}_i \in C_a$;
(iii) for each pair $a,b \in [m]$, add the arc $(q_b, q_a)$ if and only if $b > a$;
(iv) now, add a unique vertex $p^*$ to $G$, which has all vertices in the
butterfly gadgets as out-neighbors and $t^*$ as unique in-neighbor.
Finally, add the arc and the request $(t^*, s^*)$, $\congestion-1$ requests of the form
$(s^*, t^*)$ and, for each $C_a \in \mathcal{C}$, add the request $(p^*, q_a)$ to $K$.
We exemplify the above changes in \autoref{fig:queued_butterfly2}, where the
rectangles with the thick vertical line as one of the sides correspond to the cliques
partitioning $G$.

\begin{figure}[!htb]
  \centering
  \begin{tikzpicture}[scale=0.7]
    \GraphInit[unit=3,vstyle=Normal]
    \SetVertexNormal[Shape=circle, FillColor=black, MinSize=2pt]
    \tikzset{VertexStyle/.append style = {inner sep = \inners, outer sep = \outers}}
    \SetVertexLabelOut
    \begin{scope}
      \Vertex[x=-6, y=-0.75, Lpos=180, Math, L={s^*}]{ss}
      \Vertex[x=-5, y=-0.75, Lpos=-90, Math]{q_1}
      \Vertex[x=-4, y=-0.75, Lpos=-90, Math]{q_2}
      \Vertex[x=-3, y=-0.75, Lpos=-90, Math]{q_3}

      \draw (-6.85, -2) rectangle (12.35, 2);
      \draw[line width=2pt] (-2.25, 2) -- (-2.25, -2);
    \end{scope}

    \begin{scope}
      \SetVertexNoLabel

      \Vertex[x=0, y=0, Lpos=-90, Math, Ldist=5pt, L={\alpha_1}]{c}

      \Vertex[x=1.5, y=1.5, Lpos=90, Math, L={\ol{t}_1}]{tm_1}
      \Vertex[x=-1.5, y=1.5, Lpos=90, Math]{t_1}

      \Vertex[x=1.5, y=-1.5, Lpos=-90, Math, L={\ol{s}_1}]{sm_1}
      \Vertex[x=-1.5, y=-1.5, Lpos=-90, Math]{s_1}

      \Vertex[x=-1.5, y=0.75, Lpos=180, Math, L={x^c_1}]{x_13}
      \Vertex[x=-1.5, y=0, Lpos=180, Math, L={x^b_1}]{x_12}
      \Vertex[x=-1.5, y=-0.75, Lpos=180, Math, L={x^a_1}]{x_11}

      \Vertex[x=1.5, y=0, Lpos=0, Math, L={\ol{x}^d_1}]{xm_11}

      \Edges[style={-Latex, bend left=15, opacity=0.2}](sm_1, c, tm_1)
      \Edges[style={-Latex, opacity=0.2}](sm_1, xm_11, tm_1)

      \Edges[style={-Latex, bend right=15, opacity=0.2}](s_1, c, t_1)
      \Edges[style={-Latex, opacity=0.2}](s_1, x_11)
      \Edges[style={-Latex, opacity=0.2}](x_13, t_1)
      \Edges[style={-Latex, opacity=0.2}](c,x_11)
      \Edges[style={-Latex, opacity=0.2}](c,x_12)

      \Edges[style={-Latex}](ss, q_1, q_2, q_3, x_11, x_12, x_13, c, xm_11)

      \AddVertexColor{white}{tm_1, t_1, sm_1, s_1, ss, q_1, q_2, q_3}
      \Edges[style={-Latex, dashed, bend right = 15}](x_13, q_1)
    \end{scope}

    \begin{scope}[xshift=4.5cm]
      \SetVertexNoLabel
      \Vertex[x=0, y=0, Lpos=-90, Math, Ldist=5pt, L={\alpha_2}]{c}

      \Vertex[x=1.5, y=1.5, Lpos=90, Math, L={\ol{t}_2}]{tm_2}
      \Vertex[x=-1.5, y=1.5, Lpos=90, Math]{t_2}

      \Vertex[x=1.5, y=-1.5, Lpos=-90, Math, L={\ol{s}_2}]{sm_2}
      \Vertex[x=-1.5, y=-1.5, Lpos=-90, Math]{s_2}

      \Vertex[x=-1.5, y=0.75, Lpos=180, Math, L={x^c_2}]{x_23}
      \Vertex[x=-1.5, y=0, Lpos=180, Math, L={x^b_2}]{x_22}
      \Vertex[x=-1.5, y=-0.75, Lpos=180, Math, L={x^a_2}]{x_21}

      \Vertex[x=1.5, y=0, Lpos=0, Math, L={\ol{x}^d_2}]{xm_21}

      \Edges[style={-Latex, bend left=15, opacity=0.2}](sm_2, c, tm_2)
      \Edges[style={-Latex, opacity=0.2}](sm_2, xm_21, tm_2)

      \Edges[style={-Latex, bend right=15, opacity=0.2}](s_2, c, t_2)
      \Edges[style={-Latex, opacity=0.2}](s_2, x_21)
      \Edges[style={-Latex, opacity=0.2}](x_23, t_2)
      \Edges[style={-Latex, opacity=0.2}](c,x_21)
      \Edges[style={-Latex, opacity=0.2}](c,x_22)

      \Edges[style={-Latex}](xm_11, x_21, x_22, x_23, c, xm_21)

      \AddVertexColor{white}{tm_2, t_2, sm_2, s_2}
    \end{scope}

    \begin{scope}[xshift=9cm]
      \Vertex[x=2.5, y=-0.75, Lpos=0, Math, LabelOut, L={t^*}]{ts}
      \Vertex[x=2.5, y=1.5, Lpos=0, Math, L={p^*}]{ps}

      \SetVertexNoLabel
      \Vertex[x=0, y=0, Lpos=-90, Math, Ldist=5pt, L={\alpha_3}]{c}

      \Vertex[x=1.5, y=1.5, Lpos=90, Math, L={\ol{t}_3}]{tm_3}
      \Vertex[x=-1.5, y=1.5, Lpos=90, Math]{t_3}

      \Vertex[x=1.5, y=-1.5, Lpos=-90, Math, L={\ol{s}_3}]{sm_3}
      \Vertex[x=-1.5, y=-1.5, Lpos=-90, Math]{s_3}

      \Vertex[x=-1.5, y=0.75, Lpos=180, Math, L={x^c_3}]{x_33}
      \Vertex[x=-1.5, y=0, Lpos=180, Math, L={x^b_3}]{x_32}
      \Vertex[x=-1.5, y=-0.75, Lpos=180, Math, L={x^a_3}]{x_31}

      \Vertex[x=1.5, y=0, Lpos=0, Math, L={\ol{x}^d_3}]{xm_31}

      \Edges[style={-Latex, bend left=15, opacity=0.2}](sm_3, c, tm_3)
      \Edges[style={-Latex, opacity=0.2}](sm_3, xm_31, tm_3)

      \Edges[style={-Latex, bend right=15, opacity=0.2}](s_3, c, t_3)
      \Edges[style={-Latex, opacity=0.2}](s_3, x_31)
      \Edges[style={-Latex, opacity=0.2}](x_33, t_3)
      \Edges[style={-Latex, opacity=0.2}](c,x_31)
      \Edges[style={-Latex, opacity=0.2}](c,x_32)

      \Edges[style={-Latex}](xm_21, x_31, x_32, x_33, c, xm_31, ts)

      \Edges[style={-Latex, dotted}](ts, ps)

      \AddVertexColor{white}{tm_3, t_3, sm_3, s_3, ts, ps}
    \end{scope}

  \end{tikzpicture}
  \caption{Queuing of the butterfly gadgets performed for the construction of $G$ with
    the extended critical path on the left of the figure. Each of the two smaller
    rectangles are cliques in the host graph. Within them, missing arcs are from
    top-to-bottom, right-to-left. Across the thick vertical line, the only arcs are from vertices
    of the form $\ell^a_i \in B_i$ to $q_a$, which exists if and only if $\ell_i \in
    C_a$; these are shown by the dashed arc from the first butterfly to $q_1$. The dotted
    arc is not considered part of the critical path. Arc $(t^*, s^*) \in E(G)$ is omitted
    for simplicity.
  \label{fig:queued_butterfly2}}
\end{figure}
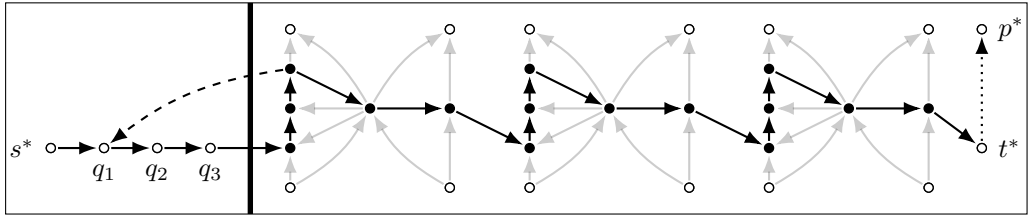

\begin{theorem}
  \label{thm:c2_hardness}
  \pname{$(k,c)$-DDP} is \NP-complete when restricted to graphs of
  $\mathcal{C}_2$.
\end{theorem}

\begin{proof}
  The forward direction should follow immediately from the construction and the arguments
  used in the proof of \autoref{lem:forward_tournament_nph}, with the additional detail
  that the critical path is used to satisfy all $\congestion-1$ requests from $s^*$ to $t^*$.
  For the converse, note that the critical path remains the unique path from $s^*$ to
  $t^*$, and so every vertex of \autoref{fig:queued_butterfly2} incident to a
  continuous black edge must always be occupied by at least $\congestion-1$ paths.
  Moreover, no other path may use $q_a$ as an intermediate vertex since $q_a$ is also an
  endpoint of the request $(p^*, q_a) \in K$.
  As such, we may assume that every path satisfying a request originating from $p^*$ is
  of the form $\angled{p^*, \ell^a_i, q_a}$, as the only way to reach a $q_a$ is through
  some $\ell^a_i$ and $\ell^a_i$ is an out-neighbor of $p^*$.
  Consequently, we have our
  assignment $\pi$ for $(X, \mathcal{C})$: set $\pi(x_i) = 1$ if and only if $\{s_i,
  \ol{x}^d_i, t_i\}$ is the path satisfying $(s_i, t_i) \in K$.
\end{proof}

It is also possible to obtain similar proofs when setting $\congestion < |\mathcal{C}|$.
The proof follows the same steps as the one of \autoref{thm:c2_hardness}, but
requires splitting $p^*$ into sufficiently many vertices to accommodate the fewer
requests that can touch $p^*$. For example, if we want $\congestion = |\mathcal{C}|/2$
and this is an integer value, then we replace $p^*$ with $p^*_1, p^*_2$ and evenly divide
the requests of the form $(p^*, q_a)$ between them. We omit the more technical details for brevity.

\begin{theorem}
  \label{thm:c2_hardness_cte_ratio}
  \pname{$(k,c)$-DDP} is \NP-complete when restricted to graphs of
  $\mathcal{C}_2$ for every constant congestion value $\congestion \geq 1$.
\end{theorem}

We conclude our polynomial-time lower bounds with the following theorem.

\begin{theorem}
  \label{thm:c2_hardness_all_ratio}
  \pname{$(k,c)$-DDP} is \NP-complete when restricted to graphs of
  $\mathcal{C}_2$ even if the congestion satisfies $\congestion = |K|/d$ for $d > 1$.
\end{theorem}

\begin{proof}
  Observe that the construction of \autoref{thm:c2_hardness} gives us $|K| = 2n + m +
  \congestion = 2n + 2m$, where $n = |X|$ and $m = |\mathcal{C}|$. To obtain
  $\congestion$ as a fraction of $|K|$, it suffices to add, for each unit of congestion
  that exceeds $m$, one request
  from $s^*$ to $p^*$. This request must be satisfied by picking the entire critical path
  since $t^*$ is the unique in-neighbor of $p^*$.
  As such, $|K| = 2n + 2m + z$, and we may pick $z = \congestion - m$ so that $|K|/\congestion$ is
  the desired ratio. For example, if we wish for $|K|/\congestion = 2$, then we must pick
  $z = 2n$, since $|K| = 2n + m + \congestion$ this means that $\congestion = 2n + m = |K|/2$.
\end{proof}

\subsection{$\W[1]$-hardness on $k$ + $h$ on graphs of bounded directed pathwidth}
While our previous lower bounds were all in the polynomial-time world, a natural question
is what happens if the numbers $k = |K|$ of requests and $h$ of partitioning cliques are
now taken as parameters.
In \cite{chudnovsky_union_of_tournaments}, Chudnosvky, Scott, and Seymour showed a
$|V(G)|^{\bigO{(hk)^5}}$-time algorithm for \pname{Directed Disjoint Path} on
$\mathcal{C}_h$.
We show that their \XP\ algorithm cannot be significantly improved by proving that \pname{Directed
Congested Disjoint Paths} remains $\W[1]$-hard when parameterized by $k + h$ on graphs of
bounded directed pathwidth.
Our reduction is heavily inspired by, but requires significant increments to, the proof
of Slivkins~\cite{Slivkins2010} of the
$\W[1]$-hardness of \pname{Edge Directed Disjoint Paths} on DAGs, where the source
problem was \pname{Clique} parameterized by the solution size; to simplify a bit our
arguments, we opt to reduce from \pname{Multicolored Clique} parameterized by the solution size.
As in Section~\ref{sec:tournament_nph}, we first deal with the completely disjoint
version of the problem and then show how to extend to the congested case.

\noindent\textbf{Construction.} Let $(G, \mathcal{V})$ be the input instance to
\pname{Multicolored Clique} where $\mathcal{V} = \{V_1, \dots, V_q\}$ are the color
classes and $(D, K)$ denotes the \pname{Directed Disjoint Paths} instance we are going to build.
For simplicity, assume that every $V_i \in \mathcal{V}$ has the same size $n$.
The vertices of $D$ are partitioned in a $q \times (2 + n)$ matrix-like
fashion, with $D_{i,j} \subseteq V(D)$ corresponding to the $j$-th vertex of $V_i \in
\mathcal{V}$ and $D_{i,0}, D_{i, n+1}$ with no correspondence.
For every $i \in [q]$, let us define and add the following sets of vertices to $G$, where
each $D_{i,j}$ is a set of $2q + 3$ vertices that will be specified later:

\begin{align*}
  D_i = &\{\alpha_{i,s}, \alpha_{i,t}, a_{i,s}, a_{i,t}, b_{i,s},  b_{i,t}\}\\
  &\cup \{g^1_{i, s}, g^1_{i, t}, g^2_{i, s}, g^2_{i, t}\}\\
  &\cup \bigcup_{j \in \{0\} \cup [n+1]} D_{i,j}.
\end{align*}

For this proof, we will stick to the convention that vertices with an $s$ subscript or
$t$ subscript are sources or targets of a request of $K$, respectively.
In particular, $K$ now contains the pairs $(\alpha_{i,s}, \alpha_{i,t})$, $(a_{i,s}, a_{i,t})$,
$(b_{i,s}, b_{i,t})$, $(g^1_{i, s}, g^1_{i, t})$, and $(g^2_{i, s}, g^2_{i, t})$ for
every $i \in [q]$.
Each $D_{i,j}$ is further broken down into four tournaments: $D_{i,j}^\alpha, D_{i,j}^a,
D_{i,j}^\beta, D_{i,j}^b$.
We set $D_{i,j}^\alpha = \{\alpha_{i,j}\}$, $D_{i,j}^\beta = \{\beta^1_{i,j},
\beta^2_{i,j}\}$, and, for $x \in \{a,b\}$, we set $D_{i,j}^x = \{x_{i,j}^0, x_{i,j}^1,
\dots, x_{i,j}^q\}$; intuitively, $x^\ell_{i,j}$ is the only vertex that will connect
$D_{i,j}$ to vertices of $D_\ell$, while $x^0_{i,j}$ is used to synchronize the internal
behavior of $D_i$.
Let us divide our construction in steps:
\begin{enumerate}[(i)]
  \item The arcs of $D_{i,j}$ including $D_{i,0}$ and $D_{i,n}$ are built as follows:
    each $D_{i,j}^x$, $x \in \{\alpha, a, \beta, b\}$ has exactly one Hamiltonian path --
    known as the \textit{$x$-path} of
    $D_{i,j}$ -- and all other arcs are in the opposite direction; i.e., if
    $\angled{x^0_{i,j}, \dots, x^\ell_{i,j}}$ is the Hamiltonian path of
    $D_{i,j}^x$, then $(x^z_{i,j}, x^y_{i,j}) \in E(D)$ for every non-consecutive pair
    $z,y$ with $z > y$; w.l.o.g. we assume the visitation order of the
    Hamiltonian path is the same as the one we specified in the definition of
    $D_{i,j}^x$. The only other arcs incident to two vertices of $D_{i,j}$ are a perfect
    matching from $D_{i,j}^a$ to $D_{i,j}^b$, such that $a^\ell_{i,j}$ is matched to
    $b^\ell_{i,j}$ for $\ell \in [n]$.
    We present an example in \autoref{fig:slivkins_s1}.

    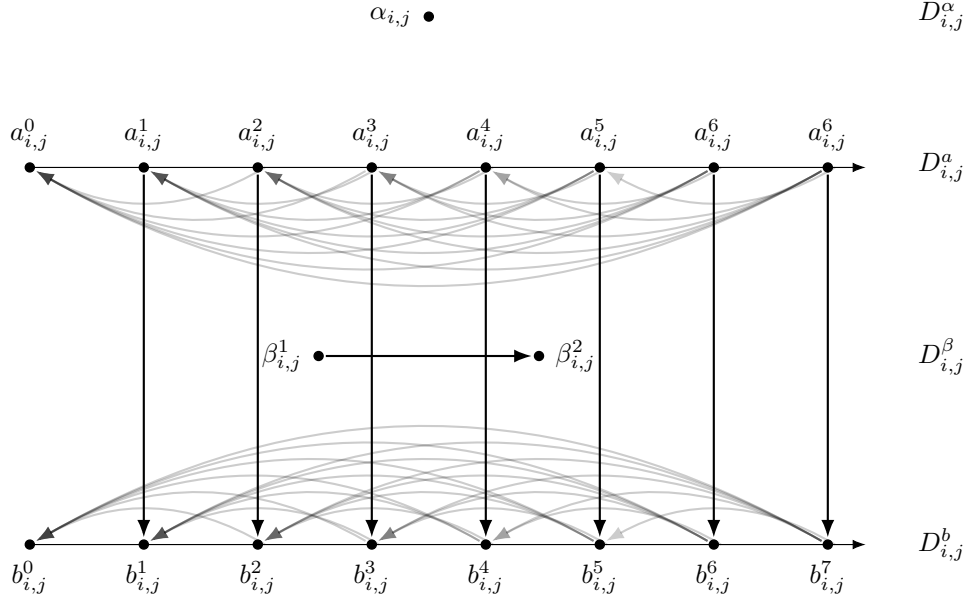
\begin{figure}[!htb]
      \centering
      \begin{tikzpicture}[scale=1]
        \GraphInit[unit=3,vstyle=Normal]
        \SetVertexNormal[Shape=circle, FillColor=black, MinSize=2pt]
        \tikzset{VertexStyle/.append style = {inner sep = \inners, outer sep = \outers}}
        \SetVertexLabelOut

        \Vertex[x=5.25, y=2, Math, Lpos=180, L={\alpha_{i,j}}]{alpha}
        \node at (12, 2) {$D^\alpha_{i,j}$};

        \begin{scope}
          \draw[line width=0.15mm, -Latex] (0,0) -- (11, 0);
          \Vertex[x=0, y=0, Math, Lpos=90, L={a^0_{i,j}}]{a0}
          \Vertex[x=1.5, y=0, Math, Lpos=90, L={a^1_{i,j}}]{a1}
          \Vertex[x=3, y=0, Math, Lpos=90, L={a^2_{i,j}}]{a2}
          \Vertex[x=4.5, y=0, Math, Lpos=90, L={a^3_{i,j}}]{a3}
          \Vertex[x=6, y=0, Math, Lpos=90, L={a^4_{i,j}}]{a4}
          \Vertex[x=7.5, y=0, Math, Lpos=90, L={a^5_{i,j}}]{a5}
          \Vertex[x=9, y=0, Math, Lpos=90, L={a^6_{i,j}}]{a6}
          \Vertex[x=10.5, y=0, Math, Lpos=90, L={a^6_{i,j}}]{a7}
          \foreach \i in {0,...,5} {
            \pgfmathsetmacro{\b}{\i+2}
            \foreach \j in {\b,...,7} {
              \Edges[style={-Latex, bend left, opacity=0.2}](a\j , a\i);
            }
          }
          \node at (12, 0) {$D^a_{i,j}$};
        \end{scope}

        \Vertex[x=3.80, y=-2.5, Math, Lpos=180, L={\beta^1_{i,j}}]{beta1}
        \Vertex[x=6.70, y=-2.5, Math, Lpos=0, L={\beta^2_{i,j}}]{beta2}
        \Edges[style={-Latex}](beta1,beta2);
        \node at (12, -2.5) {$D^\beta_{i,j}$};

        \begin{scope}[yshift=-5cm]
          \draw[line width=0.15mm, -Latex] (0,0) -- (11, 0);
          \Vertex[x=0, y=0, Math, Lpos=-90, L={b^0_{i,j}}]{b0}
          \Vertex[x=1.5, y=0, Math, Lpos=-90, L={b^1_{i,j}}]{b1}
          \Vertex[x=3 , y=0, Math, Lpos=-90, L={b^2_{i,j}}]{b2}
          \Vertex[x=4.5, y=0, Math, Lpos=-90, L={b^3_{i,j}}]{b3}
          \Vertex[x=6 , y=0, Math, Lpos=-90, L={b^4_{i,j}}]{b4}
          \Vertex[x=7.5, y=0, Math, Lpos=-90, L={b^5_{i,j}}]{b5}
          \Vertex[x=9 , y=0, Math, Lpos=-90, L={b^6_{i,j}}]{b6}
          \Vertex[x=10.5, y=0, Math, Lpos=-90, L={b^7_{i,j}}]{b7}
          \foreach \i in {0,...,5} {
            \pgfmathsetmacro{\b}{\i+2}
            \foreach \j in {\b,...,7} {
              \Edges[style={-Latex, bend right, opacity=0.2}](b\j , b\i);
            }
          }

          \foreach \i in {1,...,7} {
            \Edges[style={-Latex}](a\i , b\i);
          }
          \node at (12, 0) {$D^b_{i,j}$};
        \end{scope}

      \end{tikzpicture}
      \caption{Gadget $D_{i,j}$ and the arcs built in Step (i) of the reduction. Grey
        arcs are only used to force that their incident vertices form a tournament.
        Horizontal arrows denote the three $x$-paths of $D_{i,j}$ with at least one edge,
        i.e., $x \in \{a, b, \beta\}$.
      \label{fig:slivkins_s1}}
    \end{figure}

    Intuitively, the $a$- and $b$-paths will be used in the encoding of the
    \pname{Multicolored Clique} instance, while the $\alpha$- and $\beta$-paths are
    used to synchronize the decisions performed in the $a$- and $b$-paths.
  \item The next step in the construction is the connection of the different elements of
    $D_i$. Our goal is to add arcs so that exactly one index $j \in [n]$ has both
    $D^a_{i,j}$ and $D^b_{i,j}$ unoccupied by paths satisfying requests internal to
    gadget $D_i$, while all other vertices of $D^a_i$ and $D_b^i$ are occupied; that is,
    exactly one vertex of each color class, corresponding to that index $j$, may be
    picked in the \pname{Multicolored Clique} instance. This can be accomplished as follows:
    \begin{enumerate}
      \item For each $j \in [n+1]$ and $x \in \{\alpha, a, \beta, b\}$, take the last
        vertex of the $x$-path of $D_{i,j-1}^x$
        and add an arc to the first vertex of the $x$-path of $D_{i,j}^x$, with all other
        arcs going from $D_{i,j}^x$ to $D_{i,j-1}^x$. Doing this, $D_i^x = D_{i,0} \cup
        \bigcup_{j \in [n]} D_{i,j}$ is a tournament with a unique Hamiltonian path
        obtained by stitching the $x$-paths of $D_{i,j}^x$'s to each other.
      \item Add the arcs $(\alpha_{i,s}, \alpha_{i,0})$, $(a_{i,s}, a^0_{i,0})$,
        $(b_{i,s}, b^0_{i,0})$, the arcs $(a^q_{i,n+1}, \alpha_{i,t})$, $(\beta^2_{i,n+1},
        b_{i,t})$, $(b^q_{i,n+1}, a_{i,t})$, and all remaining right-to-left arcs so that
        the vertices in each long horizontal arrow in \autoref{fig:slivkins_s2} form a
        tournament; e.g. $a_{i,s}$ has $D_i^a \setminus \{a^0_{i,0}\}$ as its
        in-neighbors while $\alpha_{i,t}$ has $D_i^a \setminus \{a^q_{i, n+1}\}$ as
        its out-neighbors.
      \item For each $j \in \{0\} \cup [n-1]$, we add the \textit{forward jumping arcs}
        $(\alpha_{i,j}, a^0_{i,
        j+2})$, $(a^q_{i,j}, b^0_{i,j+2})$, and $(b^q_{i,j}, \beta^1_{i,j})$.
      \item For each $j \in [n]$, we add the \textit{backward jumping arcs}
        $(\alpha_{i,j}, \beta^2_{i,j-1})$ and $(a^0_{i,j}, \beta^1_{i,j-1})$
      \item Finally, for each $j \in \{0\} \cup [n]$, we add every arc from
        $g^1_{i,s}$ to $D^a_{i}$, from $g^2_{i,s}$ to $D^\alpha_{i}$, from
        $D^\beta_i$ to $g^1_{i,t}$ and $g^2_{i,t}$, and between the latter two arbitrarily.
        Consequently, $\{g^1_{i,s}\} \cup D^a_{i}$, $\{g^2_{i,s}\} \cup D^\alpha_{i}$,
        and $\{g^1_{i,t}, g^2_{i,t}\} \cup D^\beta_{i}$ are tournaments.
    \end{enumerate}

    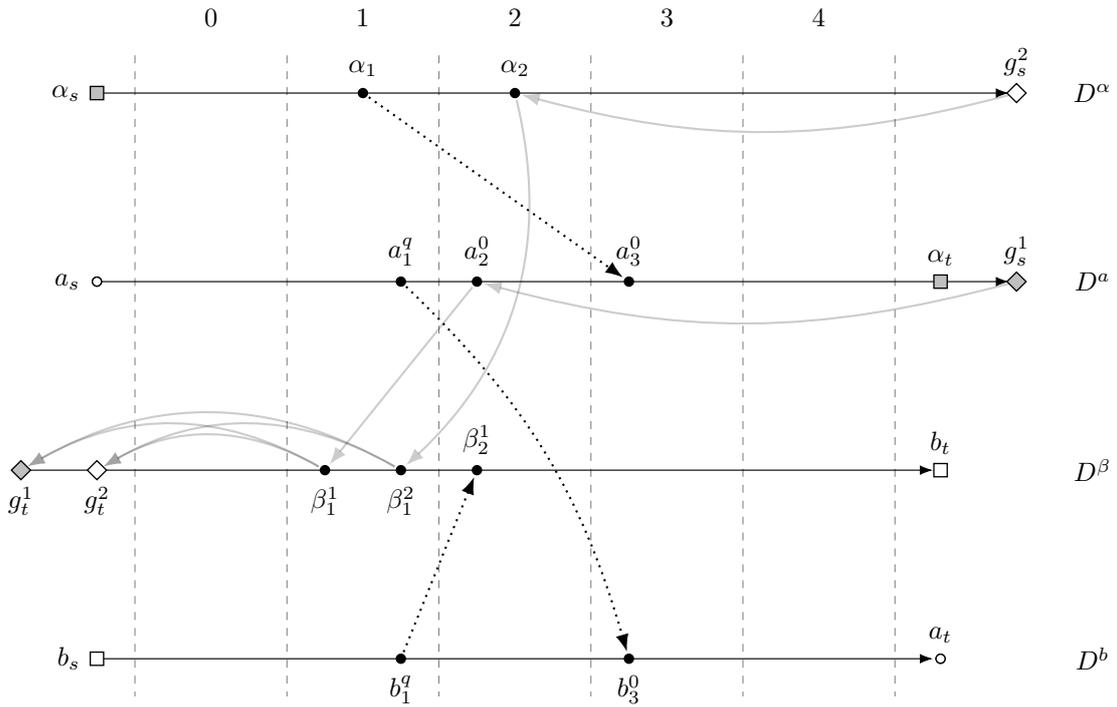
\begin{figure}[!htb]
      \centering
      \begin{tikzpicture}[scale=1]
        \GraphInit[unit=3,vstyle=Normal]
        \SetVertexNormal[Shape=circle, FillColor=black, MinSize=2pt]
        \tikzset{VertexStyle/.append style = {inner sep = \inners, outer sep = \outers}}
        \SetVertexLabelOut

        \begin{scope}[yshift=2.5cm]
          \draw[line width=0.15mm, -Latex] (0,0) -- (12, 0);
          \Vertex[x=0, y=0, Math, Lpos=180, L={\alpha_{s}}]{alphas}
          \begin{scope}[xshift=2.5cm]
            \Vertex[x=1, y=0, Math, Lpos=90, L={\alpha_{1}}]{alpha1}
          \end{scope}

          \begin{scope}[xshift=4.5cm]
            \Vertex[x=1, y=0, Math, Lpos=90, L={\alpha_{2}}]{alpha2}
          \end{scope}

          \Vertex[x=12.1, y=0, Math, Lpos=90, L={g^2_{s}}]{c2s}
          \Edges[style={opacity=0.2, -Latex, bend left=15}](c2s, alpha2)

          \node at (13.1, 0) {$D^\alpha$};

        \end{scope}

        \begin{scope}
          \draw[line width=0.15mm, -Latex] (0,0) -- (12, 0);
          \Vertex[x=0, y=0, Math, Lpos=180, L={a_{s}}]{as}
          \Vertex[x=11.1, y=0, Math, Lpos=90, L={\alpha_{t}}]{alphat}

          \begin{scope}[xshift=2.5cm]
            \Vertex[x=1.5, y=0, Math, Lpos=90, L={a^q_{1}}]{a1}
          \end{scope}

          \begin{scope}[xshift=4.5cm]
            \Vertex[x=0.5, y=0, Math, Lpos=90, L={a^0_{2}}]{a2}
          \end{scope}

          \begin{scope}[xshift=6.5cm]
            \Vertex[x=0.5, y=0, Math, Lpos=90, L={a^0_{3}}]{a3}
            \Edges[style={dotted, -Latex}](alpha1, a3)
          \end{scope}

          \Vertex[x=12.1, y=0, Math, Lpos=90, L={g^1_{s}}]{c1s}
          \Edges[style={opacity=0.2, -Latex, bend left=15}](c1s, a2)
          \node at (13.1, 0) {$D^a$};
        \end{scope}

        \begin{scope}[yshift=-2.5cm]
          \draw[line width=0.15mm, -Latex] (-1,0) -- (11, 0);
          \Vertex[x=11.1, y=0, Math, Lpos=90, L={b_{t}}]{bt}

          \begin{scope}[xshift=2.5cm]
            \Vertex[x=0.5, y=0, Math, Lpos=-90, L={\beta^1_{1}}]{beta11}
            \Vertex[x=1.5, y=0, Math, Lpos=-90, L={\beta^2_{1}}]{beta12}
            \Edges[style={opacity=0.2, -Latex}](a2, beta11)
            \Edges[style={opacity=0.2, -Latex, bend left}](alpha2, beta12)
          \end{scope}

          \begin{scope}[xshift=4.5cm]
            \Vertex[x=0.5, y=0, Math, Lpos=90, L={\beta^1_{2}}]{beta2}

          \end{scope}

          \Vertex[x=0, y=0, Math, Lpos=-90, L={g^2_{t}}]{c2t}
          \Vertex[x=-1, y=0, Math, Lpos=-90, L={g^1_{t}}]{c1t}
          \Edges[style={opacity=0.2, -Latex, bend right}](beta11, c1t)
          \Edges[style={opacity=0.2, -Latex, bend right}](beta12, c1t)
          \Edges[style={opacity=0.2, -Latex, bend right}](beta11, c2t)
          \Edges[style={opacity=0.2, -Latex, bend right}](beta12, c2t)

          \node at (13.1, 0) {$D^\beta$};
        \end{scope}

        \begin{scope}[yshift=-5cm]
          \draw[line width=0.15mm, -Latex] (0,0) -- (11, 0);
          \Vertex[x=0, y=0, Math, Lpos=180, L={b_{s}}]{bs}
          \Vertex[x=11.1, y=0, Math, Lpos=90, L={a_{t}}]{at}

          \begin{scope}[xshift=2.5cm]
            \Vertex[x=1.5, y=0, Math, Lpos=-90, L={b^q_{1}}]{b1}
            \Edges[style={dotted, -Latex}](b1, beta2)
          \end{scope}

          \begin{scope}[xshift=6.5cm]
            \Vertex[x=0.5, y=0, Math, Lpos=-90, L={b^0_{3}}]{b3}
            \Edges[style={dotted, -Latex, bend left=15}](a1, b3)
          \end{scope}

          \node at (13.1, 0) {$D^b$};
        \end{scope}

        \begin{scope}[xshift=0.5cm]
          \draw[dashed, opacity=0.5] (0, 3) -- (0, -5.5);
          \node at (1, 3.5) {0};
        \end{scope}

        \begin{scope}[xshift=2.5cm]
          \draw[dashed, opacity=0.5] (0, 3) -- (0, -5.5);
          \node at (1, 3.5) {1};
        \end{scope}

        \begin{scope}[xshift=4.5cm]
          \draw[dashed, opacity=0.5] (0, 3) -- (0, -5.5);
          \node at (1, 3.5) {2};
        \end{scope}

        \begin{scope}[xshift=6.5cm]
          \draw[dashed, opacity=0.5] (0, 3) -- (0, -5.5);
          \node at (1, 3.5) {3};
        \end{scope}

        \begin{scope}[xshift=8.5cm]
          \draw[dashed, opacity=0.5] (0, 3) -- (0, -5.5);
          \node at (1, 3.5) {4};
        \end{scope}

        \begin{scope}[xshift=10.5cm]
          \draw[dashed, opacity=0.5] (0, 3) -- (0, -5.5);
        \end{scope}

        \AddVertexColor{white}{as, at}
        \begin{scope}
          \tikzset{VertexStyle/.append style = {shape = rectangle, inner sep = 2.5pt}}
          \AddVertexColor{gray!50}{alphas, alphat}
          \AddVertexColor{white}{bs, bt}
        \end{scope}
        \begin{scope}
          \tikzset{VertexStyle/.append style = {shape = diamond, inner sep = 2pt}}
          \AddVertexColor{gray!50}{c1s, c1t}
          \AddVertexColor{white}{c2s, c2t}
        \end{scope}

      \end{tikzpicture}
      \caption{Gadget $D_i$ and the arcs built in Step (ii) of the reduction. Dotted arcs
        correspond to the forward jumping arcs of Step (ii.c), while the gray arcs
        represent the arcs of Steps (ii.d) and (ii.e). We do not illustrate the arcs
        added in Steps (ii.a), (ii.b), (ii.c) and omit the $i$ in the subscripts to
        improve readability. Differently shaped/colored vertices correspond to different
      requests of $K$.\label{fig:slivkins_s2}}
    \end{figure}
  \item While Step (ii) essentially encodes that only one vertex of each $V_i$ may be
    picked, we now must make sure that they collectively indeed form a tournament of
    $G$. Our final set of vertices are obtained by adding the following requests and
    their corresponding vertices $(\delta_{i,\ell,s},
    \delta_{i,\ell,t})$, where $i \in [q]$ and $\ell \in [q] \setminus [i-1]$; this
    implies that $|K| = 5q + \binom{q}{2} + q$.
    We proceed as follows, which is illustrated in \autoref{fig:slivkins_s3}.
    \begin{enumerate}
      \item For each triple $i, \ell \in [q], j \in [n]$ with $i \leq \ell$: add the
        arcs $(\delta_{i,\ell,s}, a^\ell_{i,j}), (b^i_{\ell, j}, \delta_{i,\ell,t})$.
      \item Finally, let $u \in V_i$ and $v \in V_\ell$ be adjacent vertices of $G$
        with $i < \ell$: add the arc $(b^\ell_{i,u}, a^i_{\ell,v})$.
    \end{enumerate}

    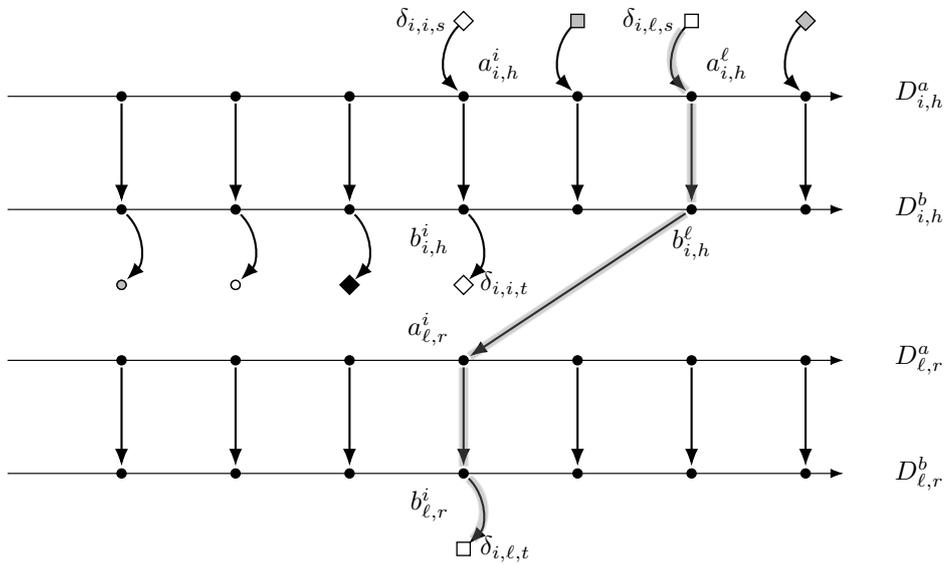
\begin{figure}[!htb]
      \centering
      \begin{tikzpicture}[scale=1]
        \GraphInit[unit=3,vstyle=Normal]
        \SetVertexNormal[Shape=circle, FillColor=black, MinSize=2pt]
        \tikzset{VertexStyle/.append style = {inner sep = \inners, outer sep = \outers}}
        \SetVertexLabelOut

        \begin{scope}
          \draw[line width=0.15mm, -Latex] (0,0) -- (11, 0);

          \Vertex[x=6, y=1.0, Math, Lpos=180, L={\delta_{i,i,s}}]{dis}
          \Vertex[x=6, y=0, Math, Lpos=45, L={a^i_{i,h}}]{a4}
          \Edges[style={-Latex, bend right=45}](dis, a4);

          \Vertex[x=9, y=1.0, Math, Lpos=180, L={\delta_{i,\ell,s}}]{dis2}
          \Vertex[x=9, y=0, Math, Lpos=45, L={a^\ell_{i,h}}]{a6}
          \Edges[style={-Latex, bend right=45}](dis2, a6);

          \SetVertexNoLabel
          \Vertex[x=1.5, y=0, Math, Lpos=90, L={a^1_{i,j}}]{a1}
          \Vertex[x=3, y=0, Math, Lpos=90, L={a^2_{i,j}}]{a2}
          \Vertex[x=4.5, y=0, Math, Lpos=90, L={a^3_{i,j}}]{a3}

          \Vertex[x=7.5, y=1.0, Math, Lpos=90, L={\delta_{i,i,s}}]{dis1}
          \Vertex[x=7.5, y=0, Math, Lpos=90, L={a^5_{i,j}}]{a5}
          \Edges[style={-Latex, bend right=45}](dis1, a5);

          \Vertex[x=10.5, y=1.0, Math, Lpos=90, L={\delta_{i,i,s}}]{dis3}
          \Vertex[x=10.5, y=0, Math, Lpos=90, L={a^7_{i,j}}]{a7}
          \Edges[style={-Latex, bend right=45}](dis3, a7);
          \node at (12, 0) {$D^a_{i,h}$};
        \end{scope}

        \begin{scope}[yshift=-1.5cm]
          \draw[line width=0.15mm, -Latex] (0,0) -- (11, 0);

          \Vertex[x=6 , y=0, Math, Lpos=-135, L={b^i_{i,h}}]{b4}
          \Vertex[x=9 , y=0, Math, Lpos=-90, L={b^\ell_{i,h}}]{b6}

          \Vertex[x=6, y=-1.0, Math, Lpos=0, L={\delta_{i,i,t}}]{dit}

          \Edges[style={-Latex, bend left=45}](b4, dit);

          \SetVertexNoLabel
          \Vertex[x=1.5, y=0, Math, Lpos=-90, L={b^1_{i,v}}]{b1}
          \Vertex[x=1.5, y=-1.0, Math, Lpos=-90, L={\delta_{i,i,t}}]{dit1}
          \Edges[style={-Latex, bend left=45}](b1, dit1);

          \Vertex[x=3 , y=0, Math, Lpos=-90, L={b^2_{i,j}}]{b2}
          \Vertex[x=3, y=-1.0, Math, Lpos=-90, L={\delta_{i,i,t}}]{dit2}
          \Edges[style={-Latex, bend left=45}](b2, dit2);

          \Vertex[x=4.5, y=0, Math, Lpos=-90, L={b^3_{i,j}}]{b3}
          \Vertex[x=4.5, y=-1.0, Math, Lpos=-90, L={\delta_{i,i,t}}]{dit3}
          \Edges[style={-Latex, bend left=45}](b3, dit3);

          \Vertex[x=7.5, y=0, Math, Lpos=-90, L={b^5_{i,j}}]{b5}
          \Vertex[x=10.5, y=0, Math, Lpos=-90, L={b^7_{i,j}}]{b7}

          \foreach \i in {1,...,7} {
            \Edges[style={-Latex}](a\i , b\i);
          }
          \node at (12, 0) {$D^b_{i,h}$};
        \end{scope}

        \begin{scope}[yshift=-3.5cm]
          \draw[line width=0.15mm, -Latex] (0,0) -- (11, 0);
          \Vertex[x=6, y=0, Math, Lpos=135, L={a^i_{\ell,r}}]{ax4}

          \SetVertexNoLabel
          \Vertex[x=1.5, y=0, Math, Lpos=90, L={a^1_{i,j}}]{ax1}
          \Vertex[x=3, y=0, Math, Lpos=90, L={a^2_{i,j}}]{ax2}
          \Vertex[x=4.5, y=0, Math, Lpos=90, L={a^3_{i,j}}]{ax3}
          \Vertex[x=7.5, y=0, Math, Lpos=90, L={a^5_{i,j}}]{ax5}
          \Vertex[x=9, y=0, Math, Lpos=90, L={a^6_{i,j}}]{ax6}
          \Vertex[x=10.5, y=0, Math, Lpos=90, L={a^7_{i,j}}]{ax7}
          \node at (12, 0) {$D^a_{\ell,r}$};
        \end{scope}

        \begin{scope}[yshift=-5cm]
          \draw[line width=0.15mm, -Latex] (0,0) -- (11, 0);

          \Vertex[x=6 , y=0, Math, Lpos=-135, L={b^i_{\ell,r}}]{bx4}
          \Vertex[x=6, y=-1.0, Math, Lpos=0, L={\delta_{i,\ell,t}}]{ditx}
          \Edges[style={-Latex, bend left=45}](bx4, ditx);

          \SetVertexNoLabel
          \Vertex[x=1.5, y=0, Math, Lpos=-90, L={b^1_{i,j}}]{bx1}
          \Vertex[x=3 , y=0, Math, Lpos=-90, L={b^2_{i,j}}]{bx2}
          \Vertex[x=4.5, y=0, Math, Lpos=-90, L={b^3_{i,j}}]{bx3}
          \Vertex[x=7.5, y=0, Math, Lpos=-90, L={b^5_{i,j}}]{bx5}
          \Vertex[x=9 , y=0, Math, Lpos=-90, L={b^6_{i,j}}]{bx6}
          \Vertex[x=10.5, y=0, Math, Lpos=-90, L={b^7_{i,j}}]{bx7}

          \Edges[style={-Latex}](b6, ax4);

          \foreach \i in {1,...,7} {
            \Edges[style={-Latex}](ax\i , bx\i);
          }
          \node at (12, 0) {$D^b_{\ell,r}$};
        \end{scope}

        \AddVertexColor{white}{dit2}
        \AddVertexColor{gray!50}{dit1}

        \begin{scope}
          \tikzset{VertexStyle/.append style = {shape = rectangle, inner sep = 2.5pt}}
          \AddVertexColor{white}{dis2, ditx}
          \AddVertexColor{gray!50}{dis1}
        \end{scope}
        \begin{scope}
          \tikzset{VertexStyle/.append style = {shape = diamond, inner sep = 2pt}}
          \AddVertexColor{white}{dis, dit}
          \AddVertexColor{gray!50}{dis3}
          \AddVertexColor{black}{dit3}
        \end{scope}

        \draw[gray,double=gray, double distance=3pt, opacity=0.2] (dis2) to [bend right=45] (a6);
        \draw[gray,double=gray, double distance=3pt, opacity=0.2] (a6) -- (b6) -- (ax4) -- (bx4);
        \draw[gray,double=gray, double distance=3pt, opacity=0.2] (bx4) to [bend left=45] (ditx);
      \end{tikzpicture}
      \caption{Partial representation of the gadgets $D_{i,h}$ and $D_{\ell, r}$ with the thicker
        arcs being those added in Step (iii) of the reduction. The grayed path encodes one
        edge of the multicolored clique of $G$. Differently shaped/colored vertices
      correspond to different requests of $K$.\label{fig:slivkins_s3}}
    \end{figure}
    \end{enumerate}
    Note that, if we can guarantee that each $i \in [q]$ has exactly one $D_{i,j}$
    unoccupied by the paths discussed in Step (ii), then the paths used to satisfy each
    $\delta_{i, \ell}$ pair must, necessarily, go through arcs of $D$ that encode
    adjacency in $G$, and all such arcs must be incident to the same vertex of $V_i$.
    We now formalize our intuition in the lemmas required to prove our main result for this section.

    \begin{observation}
      \label{obs:slivkins_structure}
      Instance $(D, K)$ of \pname{Directed Disjoint Paths} has $|K| = 6q + \binom{q}{2}$,
      $D$ can be partitioned into $6q + 2\binom{q}{2}$ tournaments and has directed
      pathwidth equal to 2.
    \end{observation}

    \begin{proof}
      The first property has already been proved during the construction.
      The second one follows almost immediately: each $D_i$ has four non-trivial tournaments:
      $D_i^\alpha \cup \{g^2_{i, s}\}$, $D_{i}^a \cup \{g^1_{i, s}\}$, $D_{i}^\beta \cup
      \{g^1_{i, t}, g^2_{i, t}\}$, and $D_{i}^b$; aside from these, $D$ has
      $2\binom{q}{2}$ trivial tournaments of the form
      $\delta_{i, \ell, *}$ with $i < \ell$ plus $2q$ trivial tournaments with $i = \ell$.

      The final property requires a bit more of care, but note that each $x$-path of
      $D_i$ is a digraph of directed pathwidth 2: there is a path decomposition $P_{i}^x$ of width
      2 where each bag contains one arc of the $x$-path and consecutive bags have consecutive arcs.
      As such, we can compose a width-2 decomposition of $D_i$ by ordering its vertices
      as $P_i = \angled{g_{i,t}^1, g_{i,t}^2, P_i^\beta, P_i^b, P_i^a, P_i^\alpha,
      g_{i,s}^1, g_{i,s}^2}$.
      Finally, we order $D$ as:
      \begin{align*}
        \langle&\angled{\delta_{i,\ell,t} \mid i \in [q], \ell \in [q] \setminus[i-1]},\\
        &\angled{P_q, \dots, P_1},\\
        &\angled{\delta_{i,\ell,t} \mid i \in [q], \ell \in [q] \setminus [i-1]} \rangle.
      \end{align*}
      Note that the $P_i$'s are ordered \textit{decreasingly}, while the other two listed
      suborders can be arbitrarily internally ordered.
    \end{proof}

    We could modify our reduction to get a smaller number of cliques, but breaking the
    $\Theta(q^2)$ bound for either $|K|$ or the number of cliques seems very hard with
    this approach.
    As such, we do not further trim down either parameter, opting for the current version
    for ease of exposition.

    \begin{lemma}
      \label{lem:slivkins_forward}
      If there is a solution $Q$ to the \pname{Multicolored Clique} instance $(G,
      \mathcal{V})$, then the \pname{Directed Disjoint Path} instance $(D, K)$ admits a
      solution $\mathcal{P}$.
    \end{lemma}

    \begin{proof}
      Let $u \in V_i \cap Q$. To satisfy the internal paths of $D_i$, denoted by
      $\mathcal{I}_i$, we add the following
      disjoint paths, which are illustrated in \autoref{fig:slivkins_s4}:

      \begin{align*}
        &(\alpha_{i,s}, \alpha_{i,t}): &\angled{\alpha_{i,s}, \alpha_{i,0}, \dots,
        \alpha_{i,u-1}, a^0_{i,u+1}, \dots, a^q_{i, n+1}, \alpha_{i,t}}\\
        &(a_{i,s}, a_{i,t}): &\angled{a_{i,s}, a^0_{i,0}, \dots, a^q_{i,u-1}, b^0_{i,u+1}, \dots,
        b^q_{i, n+1}, a_{i,t}}\\
        &(b_{i,s}, b_{i,t}): &\angled{b_{i,s}, b^0_{i,0}, \dots, b^q_{i,u-1}, \beta^1_{i,u}, \dots,
        \beta^2_{i, n+1}, b_{i,t}}\\
        &(g^2_{i,s}, g^2_{i,t}): &\angled{g^2_{i,s}, \alpha_{i,u}, \beta^2_{i,u-1}, g^2_{i,t}}\\
        &(g^1_{i,s}, g^1_{i,t}): &\angled{g^1_{i,s}, a^0_{i,u}, \beta^1_{i,u-1}, g^1_{i,t}}.\\
      \end{align*}

      \begin{figure}[!htb]
        \centering
        \begin{tikzpicture}[scale=1]
          \GraphInit[unit=3,vstyle=Normal]
          \SetVertexNormal[Shape=circle, FillColor=black, MinSize=2pt]
          \tikzset{VertexStyle/.append style = {inner sep = \inners, outer sep = \outers}}
          \SetVertexLabelOut

          \begin{scope}[yshift=2.5cm]
            \draw[line width=0.15mm, -Latex] (0,0) -- (12, 0);
            \Vertex[x=0, y=0, Math, Lpos=180, L={\alpha_{s}}]{alphas}
            
            \begin{scope}[xshift=0.5cm]
              \node at (1, -0.5) {$P_\alpha$};
            \end{scope}
            
            \begin{scope}[xshift=2.5cm]
              \Vertex[x=1, y=0, Math, Lpos=90, L={\alpha_{1}}]{alpha1}
            \end{scope}

            \begin{scope}[xshift=4.5cm]
              \Vertex[x=1, y=0, Math, Lpos=90, L={\alpha_{2}}]{alpha2}
            \end{scope}

            \Vertex[x=12.1, y=0, Math, Lpos=90, L={g^2_{s}}]{c2s}
            \Edges[style={opacity=0.2, -Latex, bend left=15}](c2s, alpha2)

            \begin{scope}[xshift=8.5cm]
              \node at (1, -1) {$P_{g^2}$};
            \end{scope}

            \node at (13.1, 0) {$D^\alpha$};

          \end{scope}

          \begin{scope}
            \draw[line width=0.15mm, -Latex] (0,0) -- (12, 0);
            \Vertex[x=0, y=0, Math, Lpos=180, L={a_{s}}]{as}
            \Vertex[x=11.1, y=0, Math, Lpos=90, L={\alpha_{t}}]{alphat}

            \begin{scope}[xshift=0.5cm]
              \node at (1, -0.5) {$P_a$};
            \end{scope}

            \begin{scope}[xshift=2.5cm]
              \Vertex[x=1.5, y=0, Math, Lpos=90, L={a^q_{1}}]{a1}
            \end{scope}

            \begin{scope}[xshift=4.5cm]
              \Vertex[x=0.5, y=0, Math, Lpos=90, L={a^0_{2}}]{a2}
            \end{scope}

            \begin{scope}[xshift=6.5cm]
              \Vertex[x=0.5, y=0, Math, Lpos=90, L={a^0_{3}}]{a3}
              \Edges[style={-Latex}](alpha1, a3)
            \end{scope}
            
            \begin{scope}[xshift=8.5cm]
              \node at (1, -1) {$P_{g^1}$};
            \end{scope}

            \Vertex[x=12.1, y=0, Math, Lpos=90, L={g^1_{s}}]{c1s}
            \Edges[style={opacity=0.2, -Latex, bend left=15}](c1s, a2)
            \node at (13.1, 0) {$D^a$};
          \end{scope}

          \begin{scope}[yshift=-2.5cm]
            \draw[line width=0.15mm, -Latex] (-1,0) -- (11, 0);
            \Vertex[x=11.1, y=0, Math, Lpos=90, L={b_{t}}]{bt}

            \begin{scope}[xshift=2.5cm]
              \Vertex[x=0.5, y=0, Math, Lpos=-90, L={\beta^1_{1}}]{beta11}
              \Vertex[x=1.5, y=0, Math, Lpos=-90, L={\beta^2_{1}}]{beta12}
              \Edges[style={opacity=0.2, -Latex}](a2, beta11)
              \Edges[style={opacity=0.2, -Latex, bend left}](alpha2, beta12)
            \end{scope}

            \begin{scope}[xshift=4.5cm]
              \Vertex[x=0.5, y=0, Math, Lpos=90, L={\beta^1_{2}}]{beta2}

            \end{scope}

            \Vertex[x=0, y=0, Math, Lpos=-90, L={g^2_{t}}]{c2t}
            \Vertex[x=-1, y=0, Math, Lpos=-90, L={g^1_{t}}]{c1t}
            \Edges[style={opacity=0.2, -Latex, bend right}](beta11, c1t)
            \Edges[style={opacity=0.2, -Latex, bend right}](beta12, c2t)

            \node at (13.1, 0) {$D^\beta$};
          \end{scope}

          \begin{scope}[yshift=-5cm]
            \draw[line width=0.15mm, -Latex] (0,0) -- (11, 0);
            \Vertex[x=0, y=0, Math, Lpos=180, L={b_{s}}]{bs}
            \Vertex[x=11.1, y=0, Math, Lpos=90, L={a_{t}}]{at}

            \begin{scope}[xshift=0.5cm]
              \node at (1, -0.5) {$P_b$};
            \end{scope}

            \begin{scope}[xshift=2.5cm]
              \Vertex[x=1.5, y=0, Math, Lpos=-90, L={b^q_{1}}]{b1}
              \Edges[style={-Latex}](b1, beta2)
            \end{scope}

            \begin{scope}[xshift=6.5cm]
              \Vertex[x=0.5, y=0, Math, Lpos=-90, L={b^0_{3}}]{b3}
              \Edges[style={-Latex, bend left=15}](a1, b3)
            \end{scope}

            \node at (13.1, 0) {$D^b$};
          \end{scope}

          \begin{scope}[xshift=0.5cm]
            \draw[dashed, opacity=0.5] (0, 3) -- (0, -5.5);
            \node at (1, 3.5) {0};
          \end{scope}

          \begin{scope}[xshift=2.5cm]
            \draw[dashed, opacity=0.5] (0, 3) -- (0, -5.5);
            \node at (1, 3.5) {1};
          \end{scope}

          \begin{scope}[xshift=4.5cm]
            \draw[dashed, opacity=0.5] (0, 3) -- (0, -5.5);
            \node at (1, 3.5) {2};
          \end{scope}

          \begin{scope}[xshift=6.5cm]
            \draw[dashed, opacity=0.5] (0, 3) -- (0, -5.5);
            \node at (1, 3.5) {3};
          \end{scope}

          \begin{scope}[xshift=8.5cm]
            \draw[dashed, opacity=0.5] (0, 3) -- (0, -5.5);
            \node at (1, 3.5) {4};
          \end{scope}

          \begin{scope}[xshift=10.5cm]
            \draw[dashed, opacity=0.5] (0, 3) -- (0, -5.5);
          \end{scope}

          \AddVertexColor{white}{as, at}
          \begin{scope}
            \tikzset{VertexStyle/.append style = {shape = rectangle, inner sep = 2.5pt}}
            \AddVertexColor{gray!50}{alphas, alphat}
            \AddVertexColor{white}{bs, bt}
          \end{scope}
          \begin{scope}
            \tikzset{VertexStyle/.append style = {shape = diamond, inner sep = 2pt}}
            \AddVertexColor{gray!50}{c1s, c1t}
            \AddVertexColor{white}{c2s, c2t}
          \end{scope}

          \draw[goodblue,double=goodblue, double distance=3pt, opacity=0.2] (alphas) --
          (alpha1) -- (a3) -- (alphat);
          
          \draw[goodgreen,double=goodgreen, double distance=3pt, opacity=0.2] (bs) --
          (b1) -- (beta2) -- (bt);
          
          \draw[goodred,double=goodred, double distance=3pt, opacity=0.2] (as) to (a1) to [bend left=15] (b3);
          \draw[goodred,double=goodred, double distance=3pt, opacity=0.2] (b3) -- (at);

          \draw[goodyellow, double=goodyellow, double distance=3pt, opacity=0.2] (c2s) to [bend left=15] (alpha2);
          \draw[goodyellow, double=goodyellow, double distance=3pt, opacity=0.2] (alpha2) to [bend left] (beta12);
          \draw[goodyellow, double=goodyellow, double distance=3pt, opacity=0.2] (beta12) to [bend right] (c2t);
          
          \draw[goodteal, double=goodteal, double distance=3pt, opacity=0.2] (c1s) to [bend left=15] (a2);
          \draw[goodteal, double=goodteal, double distance=3pt, opacity=0.2] (a2) to (beta11);
          \draw[goodteal, double=goodteal, double distance=3pt, opacity=0.2] (beta11) to [bend right] (c1t);
        \end{tikzpicture}
        \caption{Paths used to satisfy the interior requests of a gadget $D_i$ if $j = 2$; the $i$
          subscript is omitted for readability. Shaded contiguous paths represent the
          appropriate paths added to the solution of $(D, K)$ to satisfy the long paths of
          $D_i$, while the gray arcs are used to satisfy its short paths. Differently
        shaped/colored vertices correspond to different requests of $K$.\label{fig:slivkins_s4}}
      \end{figure}
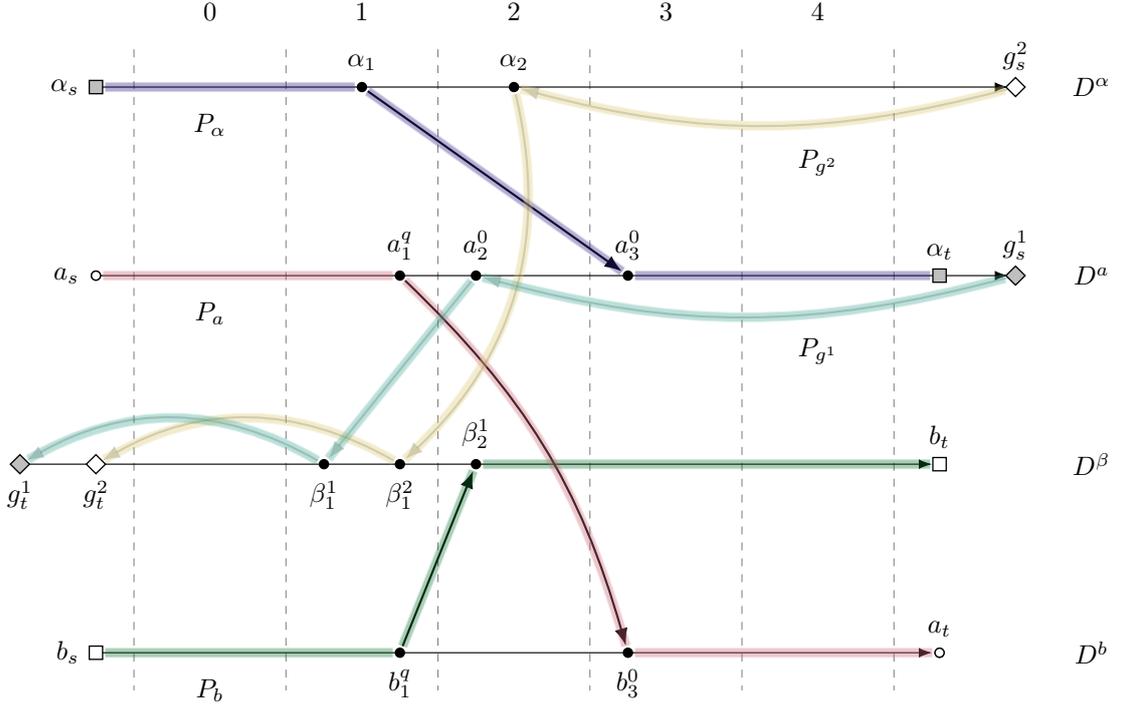

      To satisfy the requests involving the $\delta$ vertices, we add
      $\angled{\delta_{i,i,s}, a^i_{i,u}, b^i_{i,u}, \delta_{i,i,t}}$ to $\mathcal{P}$. Now, take $v
      \in V_\ell \cap Q$ with $i < \ell$. We add the path $P_{i,\ell}$, which is displayed
      in \autoref{fig:slivkins_s3}:
      \begin{equation*}
        P_{i,\ell} = \angled{\delta_{i,\ell,s}, a^\ell_{i,u}, b^\ell_{i,u}, a^i_{\ell,v},
        b^i_{\ell,v}, \delta_{i,\ell,t}}.
      \end{equation*}
      This concludes the construction of $\mathcal{P}$.
      By definition, a path $P_{i,\ell}$ occupies vertices of $D_{i,u}$ if and only if $u
      \in Q \cap V_i$, while paths of $\mathcal{I}_i$ only occupy vertices $a^0_{i,u},
      b^0_{i,u}$, which are never used by $P_{i, \ell}$, so it follows that $P_{i, \ell}$
      and $\mathcal{I}_i$ are disjoint.
      Finally, $P_{i,\ell}$ and $P_{i,r}$ are also disjoint whenever $\ell \neq r$, so
      $\mathcal{P}$ is indeed a collection of $|K|$ disjoint paths satisfying $K$.
    \end{proof}

    The converse direction, as usual, requires additional care in its proof. To this end,
    we first show that minimal solutions in $(D,K)$ must adhere to a specific format, i.e., the
    jumps between the $x$-paths of each $D_i$ happen \textit{synchronously} as in
    \autoref{fig:slivkins_s4}: all forward jumps leave a gap at precisely one index
    $j$ of each $D_i$, and this happens due to how the backward jumps were defined.

    \begin{lemma}
      \label{lem:slivkins_ab_hole}
      Let $\mathcal{P}$ be a solution of $(D, K)$ where every $P \in \mathcal{P}$ is
      induced (i.e., $P$ is minimal), $D_{i} \in D$, $P_{i,i}
      \in \mathcal{P}$ be the path satisfying $(\delta_{i,i,s}, \delta_{i,i,t}) \in K$,
      and $\mathcal{I}_i$ be the set of paths satisfying internal requests of $D_i$.
      It holds that (1) $V(\mathcal{I}_i) \subseteq D_i$ and there is some $u \in [n]$
      such that: (2) all forward jumps performed by paths
      of $\mathcal{I}_i$ originate in $D_{i, u-1}$, and (3) all backward jumps go from
      $D_{i,u} \setminus D^\beta_{i,u-1}$ to $D^\beta_{i,u-1}$.
      Moreover, (4) for every $v \in \{0, \dots, n+1\} \setminus \{u\}$, we have
      $D^a_{i,v} \cup D^b_{i,v} \subseteq V(\mathcal{I}_i)$ and $(D^a_{i,u} \cup
      D^b_{i,u}) \cap V(\mathcal{I}_i) = \{a^0_{i,u}, b^0_{i,u}\}$.
    \end{lemma}

    \begin{proof}
      To prove (1), observe that the vertices of $G \setminus V(K)$ that
      can be reached from a vertex of $D_i$ are precisely those in $D_\ell$ for $\ell >
      i$; that is, paths of $\mathcal{I}_i$ only use vertices of $D_i$.

      Towards proving (2) and (3), take $P_{i,b}$ as the path satisfying request
      $(b_{i,s}, b_{i,t})$ and note that it \textit{must} use a forward jumping arc
      $(b^q_{i,u-1},\beta^1_{i,u})$ for some $u \in [n]$, as these are the
      only arcs that connect the $b$- and $\beta$-paths of $D_i$ and to reach $b_{i,t}$
      we must use part of the $\beta$-path of $D_i$.
      Moreover, this arc is unique, as $D^\beta_i$ has \textit{no} arc leaving it.
      Since the only arcs going from $D^b_i$ to $D_i$ are of this form, it follows that:

      \begin{equation*}
        P_{i,b} = \angled{b_{i,s}, b^0_{i,0}, \dots, b^q_{i, u-1}, \beta^1_{i, u}, \dots,
        \beta^2_{i, n+1}, b_{i,t}}.
      \end{equation*}

      Define $P_{i,a}$ similarly, that is, the path satisfying request $(a_{i,s},
      a_{i,t}) \in K$, and note that it must reach $D^b_{i}$ at some point
      $b^x_{i,v}$ for some $v \geq u$ and, after that, it only occupies vertices of $D^b_i$.
      Consequently, the last vertex of $D^a_i$ in $P_{i,a}$ belongs to the suffix of the
      $a$-path that starts at $a^q_{i,u-1}$ and its superscript is not 0, as these
      vertices have no arc to $D^b_i$.
      Moreover, $P_{i,a}$ contains the prefix of the $a$-path ending at
      $a^{q-1}_{i,u-1}$: these vertices' only neighbors are either in $D^\beta_i$, which
      is a dead-end for $P_{i,a}$, or have already been used by $P_{i,b}$.
      The sole exception to this analysis is vertex (if it exists) $a^q_{i,u-2}$, which
      has $b^0_{i,u}$ not yet accounted for; but if this jump is performed, then
      $D^b_i \subset P_{i,a} \cup P_{i,b}$ and becomes impossible to satisfy the request
      $(\delta_{i,i,s}, \delta_{i,i,t}) \in K$.
      If the last vertex of $D^a_i$ in $P_{i,a}$ was not $a^q_{i,u-1}$ we would have
      $a^0_{i,u} \in P_{i,a}$, but then observe that $(g^1_{i,s}, g^1_{i,t}) \in K$ would
      be unsatisfiable: the only neighbors of $a^0_{i,v}$ with $v > u$, which are the
      only neighbors of $g^1_{i,s}$ not in $P_{i,a}$, would already be occupied by $P_{i,b}$.
      As such, we conclude that:

      \begin{equation*}
        P_{i,a} = \angled{a_{i,s}, a^0_{i,0}, \dots, a^q_{i, u-1}, b^0_{i, u+1}, \dots,
        b^q_{i, n+1}, a_{i,t}}.
      \end{equation*}

      Let us now look at the paths $P_{i,g^x}$ that satisfy the requests
      $(g^x_{i,s}, g^x_{i,t})$. The above also implies that $P_{i,g^1} =
      \angled{g^1_{i,s}, a^0_{i,u}, \beta^1_{i,u-1}, g^1_{i,t}}$.
      Now, observe that the first arc of $P_{i, g^2}$ must have $\alpha_{i,v}$ as an
      endpoint and $v \leq u$; since $P_{i, g^2}$ is minimal, no other $\alpha_i$ vertex
      is in it and $v \not> u$ as $D^\beta_{i,w} \subset P_{i,b}$ whenever $w > u$.
      In fact, it must be the case that $v = u$: the $P_{i,\alpha}$, which satisfies the
      request $(\alpha_{i,s}, \alpha_{i,t})$, cannot use vertices of
      $D^\beta_i$ and it cannot perform its unique jump to $D^a_{i}$ before
      $\alpha_{i,u-1}$, or it would clash with $P_{i,a}$. Together, these imply that:

      \begin{align*}
        P_{i,\alpha} &= \angled{\alpha_{i,s}, \alpha_{i,0}, \dots, \alpha_{i, u-1},
          a^0_{i, u+1}, \dots,
        a^q_{i, n+1}, \alpha_{i,t}}\\
        P_{i,g^2} &= \angled{g^2_{i,s}, \alpha_{i,u}, \beta^2_{i,u-1}, g^2_{i,t}}.
      \end{align*}

      Consequently, (2) and (3) hold. Finally, (4) follows immediately from the proofs of (1)-(3).
    \end{proof}

    Property (4) of \autoref{lem:slivkins_ab_hole} is the key ingredient of our proof
    and the main complication to generalize Slivkins' result. In particular, we need to
    fully use $D^a_i$ and $D^b_i$ to avoid cheating by the $\delta$ paths: without this
    guarantee, we could have a $\delta$ path hitting an $a$ vertex at some index $u$,
    going to an $a$ vertex of index $v$ using an arc added to force the appearance of a
    tournament, moving to its matched $b$ vertex and then moving on to any $D_j$,
    completely breaking down the desired behavior.
    With these bizarre paths safely ruled out, proving the converse becomes quite direct,
    as we show in the following lemma.

    \begin{lemma}
      \label{lem:slivkins_backward}
      If $(D, K)$ admits a solution $\mathcal{P}$, then $(G, \mathcal{V})$ contains a
      multicolored clique $Q$ of appropriate size.
    \end{lemma}

    \begin{proof}
      Suppose that $\mathcal{P}$ is minimal, i.e., none of its paths has a shortcut.
      By \autoref{lem:slivkins_ab_hole}, for each $i \in [q]$ there is exactly one $u
      \in [n]$ that is not completely occupied by $V(\mathcal{I}_i)$. We claim that the
      collection of vertices $Q \subseteq V(G)$ corresponding to these indices forms a
      clique in $G$.
      It suffices to show that $u,v \in Q$ implies $uv \in E(G)$.
      First, note that every $D^a_{i,w}$ for $i \in [q], w \in Q$ must be occupied by a path that
      satisfies a request between $\delta$ vertices, and each one is occupied by a
      different path; this follows directly by induction on $i$, with the base case $i =
      1$ being straightforward.
      Suppose that $u \in V_i$, $v \in V_\ell$ and $i < \ell$ and let $P_{i,\ell}$ be the
      path satisfying $(\delta_{i,\ell,s}, \delta_{i,\ell,t}) \in K$.
      Again by \autoref{lem:slivkins_ab_hole}, $a^\ell_{i,u}$ is the unique out-neighbor
      of $\delta_{i,\ell,s}$ not occupied by another path, so $(\delta_{i,\ell,s},
      a^\ell_{i,u}) \in P_{i,\ell}$. By a similar argument, $(a^\ell_{i,u}, b^\ell_{i,u})
      \in P_{i,\ell}$.
      Since $P_{i, \ell}$ exists and $b^\ell_{i,u}$ is not its terminal, the latter must
      have exactly one out-neighbor in the former. By our construction, this vertex must be
      in $D^a_\ell$ and, by \autoref{lem:slivkins_ab_hole}, this is precisely
      $a^i_{\ell, v} \in D^a_{\ell, v}$. In turn, this arc only exists if $uv \in E(G)$,
      concluding the proof.
    \end{proof}

    \whch*

    \begin{proof}
      The proof follows directly from \autoref{obs:slivkins_structure} and
      Lemmas~\ref{lem:slivkins_forward} and~\ref{lem:slivkins_backward}.
    \end{proof}

    \subsubsection{Into the congestionverse}
    We employ a trick very similar to the one used in \autoref{sec:ch_nph}.
    Essentially, we introduce a path that is \textit{almost} Hamiltonian, missing only
    the $\delta$ vertices of \autoref{thm:w1h_ch}. It is also possible to include
    them in this new path, but for simplicity's sake we do not.
    Formally, we proceed as follows: let $(D', K')$ be the instance of \pname{Directed
    Disjoint Paths} built in the proof of \autoref{thm:w1h_ch} and $\congestion > 1$ be an
    integer; to obtain $(D, K, \congestion)$,
    we add $\congestion-1$ requests of the form $(z_s, z_t)$ to $K'$ and the corresponding vertices
    to $D'$. Now, add the arcs $(z_s, g^1_{q, t})$ and $(g^2_{1, s}, z_t)$ to $D$.
    Then, for each $i \in [q]$, we add the arcs $(b_{i,t}, b_{i,s})$, $(a_{i,t}, a_{i,
    s})$, $(g^1_{i,s}, \alpha_{i,s})$ and, if $i > 1$, $(g^2_{i,s}, g^1_{i-1,t})$.
    This completes the construction of $(D, K)$; note that $|K| = |K'| + (\congestion-1)$ and that
    the directed pathwidth of $D$ remains 2.

    \begin{observation}
      Let $\mathcal{P}$ be a solution to $(D,K)$. There exists a unique path $P_z$ that
      can satisfy $(z_s, z_t)$ and, for every $i \in [q]$, this path is obtained by
      concatenating the $\beta$-, $b$-, $a$- and $\alpha$-paths, in this order.
    \end{observation}

    \begin{proof}
      Our instance $(D', K')$ in \autoref{thm:w1h_ch} was carefully constructed so
      that $D^\beta$ was a path with no outgoing arc; we used this extensively in the
      proof of \autoref{lem:slivkins_ab_hole}.
      This is no longer true for $(D, K)$, but note that the \textit{unique} outgoing arc is from
      $b_{q,t}$; that is, in order for a path to leave $D^\beta_q$ it must necessarily pass
      through this arc, so it follows that $D^\beta_q$ belongs to every path from $z_s$ to $z_t$.
      Consequently, the path now finds itself at $b_{q,s}$ and, aside from $(a_{q, t},
      a_{q,s})$, the only arcs that leave $D^b_q$ go to $D^\beta_q$, \textit{but the
      latter has already been used by the path we are constructing}, implying that $P_z$
      must also traverse the entire $b$-path of $D_q$ and then leave through $(a_{q, t}, a_{q,s})$.
      This happens again with the $a$- and $\alpha$-paths of $D_q$, with $P_z$ using the
      arc $(g^2{q,s}, g^1_{q-1, t})$ to arrive at $D_{q-1}$; note that every vertex
      reachable by $D_{q-1}$ that is not in $D_{q-1}$ has already been visited by $P_z$,
      so the argument extends all the way to $D_1$, when $P_z$ is finally able to reach $z_t$.
    \end{proof}

    As in the proof of \autoref{thm:c2_hardness_cte_ratio}, we are now free to choose
    $\congestion$ to get whichever fraction of $K$ we desire, and so we obtain the
    following theorem.

    \begin{theorem}
      \label{thm:w1h_congestion}
      \pname{$(k,c)$-DDP} jointly parameterized by the number of requests
      $|K|$ and the minimum clique cover of the host graph is \W[1]-hard even if the input instance has directed pathwidth 2 and the congestion $\congestion$ is a fraction of, but not equal to, $|K|$.
    \end{theorem}

%% file: sections/fpt.tex
\section{Algorithmic results}
\label{sec:algo}

In this section we prove our main algorithmic result and discuss its applications.
\begin{theorem}\label{thm:FPT-large-congestion}
  For all integers $k \geq 1$ and $h \geq 0$, there is a function $f(k,h)$ such
  that every instance of $(k,c)$-\textsc{DDP} with $2c > k$ on an
  $h$-semicomplete digraph $D$ containing a $f(k,h)$-triple has an irrelevant
  vertex which can be found in $\Ocal(f^2(k) \cdot n^2)$ time.
  In particular, $f(k,0) = \Ocal(2^{k\cdot \log k})$.
\end{theorem}

Fomin and Pilipzcuk~\cite[Theorem 6.3]{FominP19} showed that \textsc{Directed Edge-Disjoint Paths} is \FPT  on tournaments parameterized by the number of paths.
They solve the problem with a win-win approach.
First they show that the problem is \FPT parameterized by the directed pathwidth $\dpw(T)$ of the input semicomplete digraph $T$ plus the number $k$ of paths.
Then, through a series of results and constructions, they show that there is a computable function $g(k)$ such that every semicomplete digraph $T$ with $\dpw(T) \geq g(k)$ contains a $k$-triple, and that every such triple contains a vertex $v$ that is irrelevant for the given instance.
This second proof essentially shows that only the instances with $\dpw(T)$ bounded by some function of $k$ need to be solved.

Although they leave open whether \textsc{Directed Disjoint Paths} is \FPT on semicomplete digraphs, they mention that an \FPT dynamic programming algorithm, parameterized by the number of paths, can be done for \textsc{Directed Disjoint Paths} to obtain a result analogous to \cite[Theorem 6.3]{FominP19}.

Another approach to obtain such an \FPT algorithm relies on \emph{directed cliquewidth}, which is bounded from above by $\dpw(T) + 2$ on every semicomplete digraph $T$~\cite[Lemma 2.14]{FominP19}.
Then one can apply \cite[Theorem 2.16]{FominP19} which states that, given an MSO$_1$ formula $\phi$ and a semicomplete digraph $T$, there is an algorithm  which checks whether $\phi$ is satisfied by $T$ in \FPT time parameterized by the order of $\phi$ and $\dpw(T)$.
This implies the existence of an \FPT algorithm on semicomplete digraphs for \textsc{Directed Disjoint Paths} parameterized by the number of paths, since this problem can be modeled by MSO$_1$ (see, for instance, \cite[Proposition 4.7]{GANIAN201488}).
In any case, the following is obtained.

\begin{proposition}[Fomin and Pilipzcuk~\cite{FominP19}]\label{prop:ddp-fpt-fomin-pilipzcuk}
{\sc Directed Disjoint Paths} on semicomplete digraphs is \FPT parameterized by the number of paths and the directed pathwidth of the input semicomplete digraph.
\end{proposition}

Applying classical techniques, one can show easily show that \autoref{prop:ddp-fpt-fomin-pilipzcuk} also implies an \FPT algorithm for $(k,c)$-\textsc{DDP}.
Indeed, given an instance of this problem on a semicomplete digraph $T$, it suffices to add copies $v_2, \ldots, v_c$ of each $v \in V(T)$ with the same in- and out-neighborhood as $v$. 
Then, we add arcs among vertices in $\{v, v_2, \ldots, v_c\}$ in order to ensure that this set induces an acyclic tournament.
This procedure easily implies that the directed pathwidth of $T$ increases by at most a factor of $c$. Hence we can apply \autoref{prop:ddp-fpt-fomin-pilipzcuk} to solve $k$-\textsc{DDP} in the newly constructed digraph and transport any solution, whether positive or negative, to the original instance in $T$.
This implies that $(k,c)$-\textsc{DDP} is also \FPT parameterized by $k$ and the directed pathwidth of the input digraph.

When applying their irrelevant vertex rule, Fomin and Pilipzcuk~\cite{FominP19} use a series of results to construct in \FPT time a triple of large order on semicomplete digraphs of sufficiently large directed pathwidth.
Namely, given a semicomplete digraph $T$ and an integer $k$, they first apply \cite[Theorem 4.12]{FominP19} to either produce a directed path decomposition of width bounded by some computable function $g(k)$ or find one of the two certificates that $\dpw(T) > k$: a \emph{degree tangle} or a \emph{matching tangle}.
Then, they show how to produce a \emph{short jungle} from any of those two objects in \cite[Lemmas 3.9 and 3.12]{FominP19}.
We remark that all these constructions run in polynomial time.
They argue that although the proof of how to extract triples from short jungles by Fradkin and Seymour~\cite{FradkinS13} is not explicitly algorithmic, it is easy to extract an algorithm from it.
In any case,  the proof of~\cite[Theorem 9.9]{FominP19} shows how to circumvent the lack of an explicit algorithmic result to obtain triples from jungles, essentially obtaining the following result one way or another when applying the other results mentioned in this paragraph.

\begin{proposition}[Fomin and Pilipzcuk~\cite{FominP19}]\label{prop:path-decomposition-or-triple}
Let $T$ be a semicomplete digraph.
There is a computable function $g(t)$ and an \FPT algorithm parameterized by $t$ that either outputs a directed path decomposition of $T$ with width at most $g(t)$ or finds a $t$-triple in $T$.
\end{proposition}

In light of \cref{prop:path-decomposition-or-triple,prop:ddp-fpt-fomin-pilipzcuk}, it is now easy to apply \autoref{thm:FPT-large-congestion} to obtain an \FPT algorithm for $(k,c)$-\textsc{DDP} when $2c > k$.
We restate \autoref{thm:FPTinSemicomplete} for convenience.

\FPTinSemicomplete* 

\begin{proof}
Let $\mathcal{I}$ be an instance of $(k,c)$-\textsc{DDP} on a semicomplete digraph with $2c > k$.
Let $f(k,h)$ be the function as in the statement of \autoref{thm:FPT-large-congestion}, and $f(k) = f(k, 0)$.
Applying \autoref{prop:path-decomposition-or-triple} with input $T$ and $t = f(k)$, we either receive a directed path decomposition of $T$ with width at most $g(f(k)))$  or an $f(k)$-triple in $T$.
In the first case, we solve the instance applying \autoref{prop:ddp-fpt-fomin-pilipzcuk} (we remind the reader of the discussion immediately after this proposition).
In the second case, we apply \autoref{thm:FPT-large-congestion} to find a vertex $v \in V(T)$ that can deleted without changing the answer to $\mathcal{I}$.

After deleting an irrelevant vertex, we apply \autoref{prop:path-decomposition-or-triple} again.
Each application of this proposition that yields the second possible output results in the deletion of one vertex from $T$.
Thus, after at most $|V(T)|$ iterations the first output is guaranteed to occur, and we solve the problem applying \autoref{prop:ddp-fpt-fomin-pilipzcuk}.
\end{proof}

\subsection{Irrelevant vertices in
\texorpdfstring{$k$}{k}-triples}\label{sec:irrelevant-vertex-in-triples}
This section is dedicated to the proof of \autoref{thm:FPT-large-congestion}.

We first prove the version of \autoref{thm:FPT-large-congestion} restricted to semicomplete digraphs, and at the very end of this section we explain which are the (easy) changes to be made to the proof so that it holds for any $h \geq 1$. We abbreviate $f(k,0)$ by $f(k)$ and assume
that an $f(k)$-triple is given.
The choice of $f(k)$ is discussed later.
Let $\mathcal{I} = (T, K, c)$ be an instance of $(k,c)$-\textsc{DDP} on a
semicomplete digraph $T$ with $2c > |K| = k$.
For the remainder of this section, as well as in \autoref{sec:free_b}, \autoref{sec:free_ac}, and \autoref{sec:finding-irrelevant-vertex}, we assume that $\mathcal{K}$ is an
$f(k)$-triple $(A,B,C)$ of $T$ and that $\mathcal{P}$ is a set of paths satisfying the
following property.
\begin{property}
  \label{prop:minimizing-for-solution}
  $\mathcal{P}$ is a solution $\{P_1, \ldots, P_k\}$ for $\mathcal{I}$ minimizing
  $\sum_{i \in [k]}|V(P_i)|$.
\end{property}

Since congestion is allowed, any vertex is counted in the summation as many
times as it is used by a path.
This per-path measure is necessary as we often shortcut (we formally define shortcuts
later in this section) a path $P \in \mathcal{P}$ through a vertex that is used by other
paths in the collection
$\mathcal{P}$.

To avoid some technicalities, here we make two assumptions about $\mathcal{K}$.\
First, by discarding at most $2k$ vertices of each set $A, B$, and $C$, we assume that the terminals of $\mathcal{I}$ do not appear in $V(\mathcal{K})$.
Second, by deleting parallel arcs of $T$ if necessary, we assume that there are no arcs in $E(\mathcal{K})$ from $C$ to $B$ and no arcs from $B$ to $A$.
The removal of those arcs poses no issue for the irrelevant vertex argument since any irrelevant vertex of $T' \subseteq T$ is also irrelevant in $T$.

The set $B$ plays as distinguished role in this proof since it concentrates the heads of all arcs of $\mathcal{K}$ between $A$ and $B$ and the tails of all arcs of $\mathcal{K}$ between $B$ and $C$.
The goal is to show that there is room to reroute paths in $\mathcal{K}$ in order to show that if $\mathcal{I}$ is a  {\sf yes}-instance, then there is a vertex $b \in B$ and a solution $\mathcal{Q}$ for $\mathcal{I}$ such that every $Q \in \mathcal{Q}$ accessing $b$ does so from a vertex in $A$ and to a vertex in $C$.


We follow the blueprint of the proof for \textsc{Directed Edge Disjoint Paths} on semicomplete
digraphs by Fomin and Pilipzcuk~\cite{FominP19}.
A fundamental property in their case is that any vertex with out-degree (resp. in-degree) at
least $k+1$ in the given triple has at least one arc leaving (resp. entering) it that is not
used by any solution $\mathcal{Q}$ minimizing the sum of the lengths of its paths.
This property allows them to prove that every such solution cannot use more than
two arcs of the matching from $C$ to $A$.
This, in turn, is used to show that every path in $\mathcal{Q}$ uses at most $2k+4$ vertices from
$A$ and from $C$, and at most $4k$ vertices of $B$.
The hard part of their proof is about how to apply these results to prove that an irrelevant
vertex is guaranteed to exist in a sufficiently large triple, and how to find it in polynomial time.

In our case, the existence of a vertex $v \in V(\mathcal{K})$ with large out-degree or
in-degree in the triple does {\sl not} guarantee that we can find a vertex in $N^+(v) \cup
N^-(v)$ that is not used by a path of $\mathcal{P}$.
In fact, in many places of our proof we do not guarantee that at all, since we produce
shortcuts for paths in $\mathcal{P}$ through vertices that have been used by other paths
of $\mathcal{P}$ while but \textsl{at most} $c-1$ of them, hence being careful not to exceed the allowed congestion on each vertex.
The fundamental property that we use comes from the assumption that $2c > k$: if $u,v$
are both used by $c$ paths of $\mathcal{P}$, then by the pigeonhole principle there is a path $P \in \mathcal{P}$ that
uses \textsl{both} $u$ and $v$.
The need to rely on this property is justified by the counterexample provided in  \autoref{sec:counterexample}, which shows that no irrelevant vertex is guaranteed
to exist in arbitrarily large triples whenever $2c \leq k$, and makes our analysis
significantly harder than the one in~\cite{FominP19}.
For example, in order to show that a bounded number of vertices of $A$ and $C$ are used
by paths of $\mathcal{P}$ (\autoref{lem:bounded-intersection-in-A-C}), we first need to show that some vertices of $B$ can be
used for shortcuts (\autoref{lem:B-xor-free-pair}) and these two proofs are already much harder than their counterparts
in~\cite{FominP19}.


\medskip
In addition to the definitions introduced in the beginning of this section, we adopt the following notation.
We denote by $\Mat$ the arcs of the matching from $C$ to $A$ and, as discussed
right below \autoref{def:k-triple}, we may also
refer to $\Mat$ as a bijective mapping from $C$ to $A$ according to the arcs of
the matching defining $\mathcal{K}$.
In addition, for $C' \subseteq C$ (resp. $A' \subseteq A$) we call the set $\Mat(C') =
\{\Mat(v) \mid v \in C'\}$ (resp. $\Mat^{-1}(A') = \{\Mat^{-1}(u) \mid u \in A'\}$) the
\emph{mirror} of $C'$ (resp. of $A'$) in~$\mathcal{K}$.

For $P_i \in \mathcal{P}$ we denote by $\prec_{i}$ the order in which $V(P_i)$ appears in $P_i$.
A \emph{shortcut} for $P_i$ is a walk $R$ such that there is some subpath $R' \subseteq P_i$
where $|V(R)| < |V(R')|$ and $R$ starts and ends in the same vertices as $R'$.
See
\autoref{fig:typical-shortcut-example} for
examples of shortcuts.
We may refer to a shortcut $R$ as the sequence $v_1 \to v_2 \to \ldots \to v_m$
of the vertices of $R$ as they appear in the walk.
The existence of a shortcut for any path in $\mathcal{P}$ respecting the congestion
measure contradicts \autoref{prop:minimizing-for-solution},
and the goal of the first part of the proof of
\autoref{thm:FPT-large-congestion} is to exploit
this fact to prove useful properties on how the paths of $\mathcal{P}$ can
intersect an $f(k,h)$-triple.
To avoid repetition, we refrain from stating that these shortcuts lead to
contradictions in the proofs of this section.
Note that simply replacing $R'$ with $R$ may not lead to a minimal solution: there could be forward arcs that further shorten the resulting object, or cycles could be introduced, depending on the visiting order of the involved vertices. For example, in~\autoref{fig:typical-shortcut-example}, it could be the case that the vertex of $C$ appeared in the depicted path \emph{after} the vertex of $A$, but with the addition of the blue segment we now have a cycle.

\begin{figure}[h!]
\centering
  \begin{tikzpicture}
  \node[blackvertex] (x) at (0,0) {};
  \node[blackvertex] (v) at ($(x) + (5,0)$) {};
  \draw[arrow] (x) -- node[blackvertex, scale=.8, pos = 1/6] (x1) {}
  node[blackvertex, scale=.8, pos = 2/6] {}
  node[blackvertex, scale=.8, pos = 3/6] {}
  node[blackvertex, scale=.8, pos = 4/6] {}
  node[blackvertex, scale=.8, pos = 5/6] (v1) {} (v);
  \draw[arrow, thick, goodblue] (x1) to [bend left=30] (v1);

  \begin{scope}[xshift=8cm, yshift = -2cm]
  \node[blackvertex] (b1) at (0,0) {};
  \node[blackvertex] (b5) at ($(b1) + (0,4)$) {};
  \draw[ultra snake, thick, arrow, dashed, shorten >= 0pt] (b1) -- 
  node[blackvertex, scale =.8, pos = 1/4, yshift = -7pt] {}
  node[blackvertex, scale =.8, pos = 2/4, yshift = -4pt] (mid) {}
  node[blackvertex, scale =.8, pos = 3/4] {}
  (b5);

  \node[rectangle, draw, fit=(b1)(b5), label=90:$B$] {};

  \node at ($(mid) + (-1, 0)$) {\Large$P_i$};
  \node (w1) at ($(b1) + (-2,0)$) {};
  \node (w2) at ($(b5) + (-2,0)$) {};
  \draw[arrow] (w1) -- (b1);
  \draw[arrow] (b5) -- (w2);

  \node[blackvertex] (u) at ($(b5) + (3.5, -1)$) {};
  \node (c1) at ($(u) + (-1,0)$) {};
  \node (c2) at ($(u) + (1,0)$) {};
  \node[rectangle, draw, fit=(c1)(c2), label=0:$C$] {};

  \node[blackvertex] (v) at ($(u) + (0, -2)$) {};
  \node (a1) at ($(v) + (-1,0)$) {};
  \node (a2) at ($(v) + (1,0)$) {};
  \node[rectangle, draw, fit=(a1)(a2), label=0:$A$] {};

  \draw[thick, goodblue, arrow] (b1) to [bend right=10] (u);
  \draw[thick, goodblue, arrow] (u) -- (v);
  \draw[thick, goodblue, arrow] (v) to [bend right = 20] (b5);

  \end{scope}
  \end{tikzpicture}
  \caption{Examples of shortcuts.
    In both cases, the blue path denotes a shortcut.
    On the right, we give an illustration of a typical shortcut, internal to a
    $k$-triple $(A,B,C)$, that is often built in the proof of
    \autoref{thm:FPT-large-congestion}.
  The dashed subpath of $P_i$ is exchanged by the blue path in the figure.}
  \label{fig:typical-shortcut-example}
\end{figure}
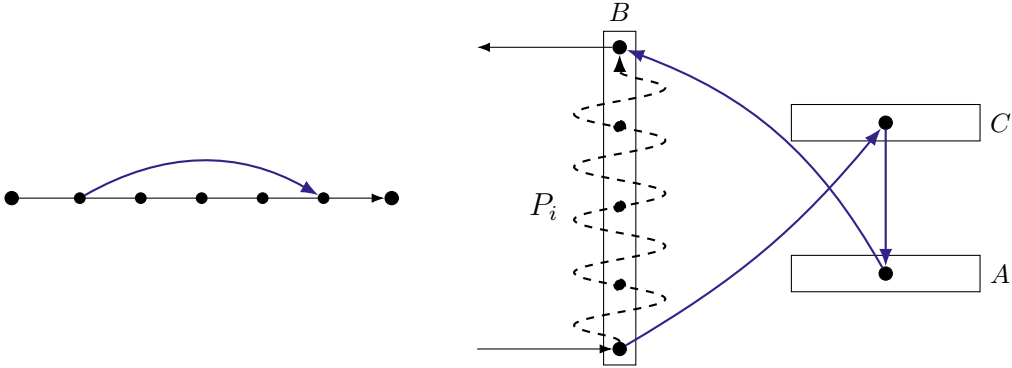

For every $v \in V(\mathcal{K})$, we assign to $v$ a list of indices
$\List(v) \subseteq 2^{[k]}$ representing which paths of $\mathcal{P}$ are
using $v$.
Thus $|\List(v)| \leq c$ holds for any vertex $v$ since $\mathcal{P}$ is a solution for
$\mathcal{I}$.
For $i \in [k]$, we say that $v$ is $\emph{$i$-free}$ if $i \in \List(v)$ or if
$|\List(v)| \leq c - 1$, and that $v$ is \emph{$i$-saturated} otherwise.
Intuitively, an $i$-free vertex can be used to construct shortcuts for path
$P_i$ since either $P_i$ is already using $v$ and thus the route may be replaced with a shorter one without increasing the congestion of $v$, or the congestion of
$v$ is not yet saturated and can be increased by one to generate a better solution.
We denote by $\mathcal{L}$ the set of all possible lists of indices that can be assigned to
a vertex by a solution.
More precisely, $\mathcal{L} = \{L \subseteq 2^{[k]} \mid |L| \leq c\}$.
For $L \in \mathcal{L}$ we denote by $V(L)$ the set $\{v \in V(\mathcal{K})
\mid \List(v) = L\}$ of vertices of $\mathcal{K}$ which are assigned indices in $L$.
Finally, for $u \in C$ and $v = \Mat(u)$ we say that $\{u,v\}$ is an \emph{$i$-free pair}
if both $u$ and $v$ are $i$-free.

\subsection{Freeing \texorpdfstring{$B$}{B}}
\label{sec:free_b}
In this section, we show that only a bounded number of vertices $b \in B$ have $|\List(b)| = c$.
We begin with the following warm up lemma.

\begin{lemma}\label{lem:B-xor-free-pair}
  For all $P_i \in \mathcal{P}$, if $|V(P_i) \cap B| \geq 5$ then there are no
  $i$-free pairs.
\end{lemma}
\begin{proof}
  By contradiction, assume that there is $P_i \in \mathcal{P}$ such that
  $|V(P_i) \cap B| \geq 5$ and there is an $i$-free pair $u,v$ with $u \in C$ and $v = \Mat(u)$.
  Let $\angled{b_1, \ldots, b_5}$ be the first five vertices of $P_i$ that appear in $B$.
  If none of $u,v \not \in V(P_i)$ then we can shortcut the subpath of $P_i$ from $b_1$
  to $b_5$ by following the arcs $b_1 \to u \to v \to b_5$.
  Thus at least one of $u,v$ appears in $P_i$.

  If both $u,v \in V(P_i)$ then $b_4 \prec_i v$, since otherwise we can shortcut
  using $v \to b_5$, and $u \prec_i b_2$, since otherwise $b_1 \to u$ is a
  shortcut.
  As such, we have that $u \prec_i b_2 \prec_i b_4 \prec_i v$, and we shortcut through $b_1 \to u \to v \to b_5$.
  The cases when only one of $u,v$ is not in $P_i$ follow similarly.
\end{proof}

\begin{lemma}\label{lem:bounded-congested-in-B}
  For all $L \in \mathcal{L}$ with $|L| = c$, it holds that $|V(L) \cap B| \leq
  4$.
\end{lemma}
\begin{proof}
  By contradiction, assume that there is $L \in \mathcal{L}$ such that $|L| =
  c$ and $|V(L) \cap B| \geq 5$.
  For all $i \in L$, let $A_i = V(P_i) \cap A$ and $C_i = \Mat^{-1}(A_i)$.

  If all $i \in L$ satisfy $|A_i| \leq 4c$, then $|V(L) \cap A| \leq 4c^2$ and thus,
  since $|L| = c$,  the set $A' = A \setminus V(L)$ has size at least $f(k) - 4c^2$ and
  all its vertices are used by at most $c-1$ paths of $\mathcal{P}$ since $2c > k$.
  By definition, all vertices in $A'$ are $j$-free for all $j \in [k]$.
  If there is $u \in \Mat^{-1}(A')$ with $|L(u)| = c$, then it holds that there is some $j \in L \cap \List(u)$ and, consequently, that $\{u, \Mat(u)\}$ is a $j$-free pair.
  Otherwise, $|\List(u)| \leq c-1$ for all $u \in \Mat^{-1}(A')$.
  Both cases contradict \autoref{lem:B-xor-free-pair} and the result follows.

  Assume now that there is an $i \in L$ such that $|A_i| \geq 4c+1$ and notice that every
  $u \in C_i$ has $|\List(u)| = c$, as otherwise $\{u,\Mat(u)\}$ would be an $i$-free
  pair, contradicting \autoref{lem:B-xor-free-pair} as $|V(P_i) \cap B| \geq |V(L) \cap B| \geq 5$.  
  The goal now is to group all pairs of the form $\{u,\Mat(u)\}$ with $u \in C_i$
  according to which indices appear in $\List(u) \cap \List(\Mat(u))$.
  The size of $|A_i|$ guarantees that at least one index $\ell$ appears in the lists of the endpoints of many of these pairs and we use it to shortcut two paths.
  First, we shortcut $P_\ell$ through some $b \in V(L) \cap B$.
  By doing this, we no longer have a solution for our instance since a vertex of this set
  is used by $c+1$ paths.
  However, the shortcut for $P_\ell$ frees some vertices of $C$, which are now used by at
  most $c-1$ paths.
  Thus we can shortcut $P_i$ in a way that avoids $b$ to generate a new solution
  contradicting \autoref{prop:minimizing-for-solution}.

  Formally, we proceed as follows. For $j \in [k] \setminus L$, let $X_j$ be the set of pairs $\{u, \Mat(u)\}$ such that
  $j \in \List(u) \cap \List(\Mat(u))$ with $u \in C_i$.
  Notice that these sets are not necessarily pairwise disjoint and that the exclusion of $L$ is natural as the occurrence of any of its indices in $\List(u) \cap
  \List(\Mat(u))$ would contradict \autoref{lem:B-xor-free-pair}.
  In addition, the pigeonhole principle implies that at least one set $X_\ell$ satisfies
  $|X_\ell| \geq 5$ as $|A_i| \geq 4c+1$.

  Let $A'$ be the subset of vertices of $A$ appearing in pairs of $X_\ell$ and $C' =
  \Mat^{-1}(A')$, and let $b_1, \ldots, b_q$ be the first $q \geq 5$ vertices of $V(L) \cap B$  ordered as they appear in $P_i$.
  We now consider two cases.

   \begin{enumerate}
      \item Assume first that $P_\ell$ uses a vertex of $A'$ before any vertex of $C'$.
      Denote by $v_1$ the first vertex of $A'$ used by $P_\ell$, and by $u$ the last vertex
      of $C'$ used by $P_\ell$.
      We construct the path $P'_\ell$ by shortcutting $P_\ell$ through $v_1 \to b_3 \to u$.
      This indeed forms a shortcut since at least four vertices used by $P_\ell$ in $C'$ are
      skipped by $P'_\ell$.
      The collection $\mathcal{P}' = (\mathcal{P} \setminus \{P_\ell\}) \cup \{P'_\ell\}$
      improves on the summation in \autoref{prop:minimizing-for-solution} with relation to
      $\mathcal{P}$, but $b_3$ is now used by $c+1$ paths and thus $\mathcal{P'}$ is not a
      solution for the instance $\mathcal{I}$.
      The choice of $P'_\ell$ ensures that at least three vertices of $C' \setminus
      \{\Mat^{-1}(v_1)\}$ are \textsl{not} used by it, and therefore we can chose one of those
      vertices, say $u'$, to shortcut $P_i$ through $b_1 \to u' \to \Mat(u') \to b_q$; by
      using this new path instead of $P_i$ in $\mathcal{P}'$, we remove $i$ from the list of $b_3$ and no other vertex exceeds the allowed congestion, thus contradicting
      \autoref{prop:minimizing-for-solution}.
      See \autoref{fig:proof-of-bounded-in-B} for an example of the shortcuts that are built in the analysis of this case.
    
      \begin{figure}[h]
        \centering
        \begin{tikzpicture}[yscale=.6]
          \foreach \i in {1,...,4} {
            \node[blackvertex] (a\i) at (\i,-.5) {};
            \node[blackvertex] (c\i) at ($(a\i) + (-1,5)$) {};
          }
          \foreach \i in {1,...,4} {
            \node[blackvertex, label={[xshift=-.1cm]180:{$b_\i$}}] (b\i) at (-2,\i-1) {};
          }
          \node[blackvertex, label={[xshift=-.1cm]180:{$b_q$}}] (b5) at (-2,5-1) {};
          \node[blackvertex, label={[yshift=-.1cm]-90:{$v_1$}}] (v1) at (0,-.5) {};
          \node[blackvertex, label={[yshift=.1cm]90:{$u$}}] (u) at ($(a4) + (0,5)$) {};
          \draw[arrow] (u) -- (a4);
          \draw[arrow] (c1) -- (v1);
          \draw[arrow] (c3) -- (a2);
          \draw[arrow] (c4) -- (a3);
    
          \node[rectangle,draw,fit=(v1)(a4), label=0:{$A \cap V(P_i) \cap V(P_\ell)$}] {};
          \node[rectangle,draw,fit=(b1)(b5), label=90:{$B \cap V(L)$}] {};
          \node[rectangle,draw,fit=(c1)(u), label=0:{$C \cap V(P_\ell)$}] {};
    
          \draw[arrow, goodblue, thick] (v1) -- (b3);
          \draw[arrow, goodblue, thick] (b3) -- (u);
          \draw[arrow, goodred, thick] (b1) -- (c2);
          \draw[arrow, goodred, thick] (c2) -- (a1);
          \draw[arrow, goodred, thick] (a1) -- (b5);
        \end{tikzpicture}
        \caption{Shortcuts built in the first case of the proof of
          \autoref{lem:bounded-congested-in-B}. The blue path starting in $v_1$ is the shortcut
          for $P_\ell$, forming $P'_\ell$. The red path starting in $b_1$ is the shortcut for
        $P_i$. The remaining arcs of $\mathcal{K}$ are omitted.}
        \label{fig:proof-of-bounded-in-B}
      \end{figure}
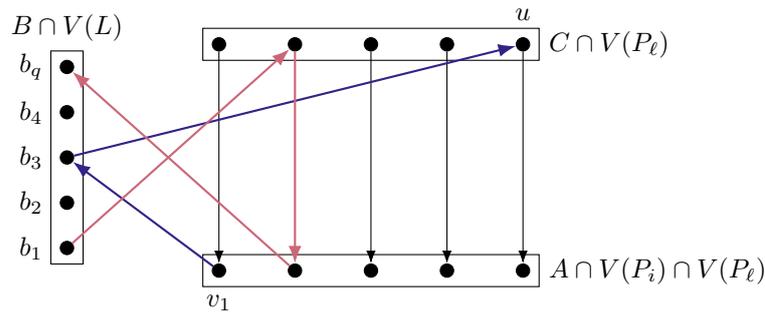
    
      \item Assume now that $P_\ell$ uses a vertex of $C'$ before any vertex of $A'$.
      Denote by $u_1$ and $w$ the first and last vertices of $C'$ used by $P_\ell$, respectively.
      Similarly to the last case, we form $P'_\ell$ by shortcutting $P_\ell$ through $u_1 \to
      \Mat(u_1) \to b_3 \to w$, thus increasing the congestion on the vertex $b_3$ to $c+1$.
      Again, there are at least three vertices of $C'$ not used by $P'_\ell$ and we use one
      of those, say $u'$, to shortcut $P_i$ through $b_1 \to u' \to \Mat(u') \to b_q$ in
      order to avoid $b_3$. This contradicts \autoref{prop:minimizing-for-solution} and the
      result follows.
      Notice that the fact that $|A'| \geq 5$ is needed only in this case due to the size of the shortcut
      for $P_\ell$.
    \end{enumerate}
    \vspace{-\baselineskip}
\end{proof}

\autoref{lem:bounded-congested-in-B} implies that at least $f(k) - 4\binom{k}{c}$
vertices of $B$ are used by at most $c-1$ paths of $\mathcal{P}$.
Thus, all such vertices are $i$-free for any $i \in [k]$ and can be used to build
shortcuts for any path in $\mathcal{P}$.

\subsection{Freeing $A$ and $C$}
\label{sec:free_ac}

We now show that only a bounded number of vertices of $A$ and $C$ are used by paths of
$\mathcal{P}$.
\begin{lemma}\label{lem:bounded-intersection-in-A-C}
  If there is $b \in B$ such that $|\List(b)| \leq c-1$, then for all $\mathcal{P}_i \in
  \mathcal{P}$ it holds that
  \begin{enumerate}
    \item $|V(P_i) \cap A| \leq 8c+4$; and
    \item $|V(P_i) \cap C| \leq 8c+4$.
  \end{enumerate}
\end{lemma}
\begin{proof}
  Let $b \in B$ with $|\List(b)| \leq c-1$. To prove \lipItem{1}, by contradiction let
  $P_i$ be a path of $\mathcal{P}$ using $q \geq 8c + 5$ vertices of $A$.
  Let $v_1 \prec_i v_2 \prec_i \cdots \prec_i v_q$ be the vertices in $V(P_i) \cap A$
  ordered as they appear in $P_i$.
  If $\Mat^{-1}(v_j)$ is $i$-free for some $j \geq 5$, then we can shortcut $P_i$ through
  $v_1 \to b \to \Mat^{-1}(v_j) \to v_j$ to build a new solution.
  Thus, we now assume that every vertex in $C' = \{\Mat^{-1}(v_j) \mid j \geq 5\}$ is $i$-saturated.

  Let $A' = \{v \in \Mat(C') \mid |\List(v)| \leq c-1\}$ and, towards a contradiction, assume that $|A'| \geq 4c+1$.
  Then there is a $P_j \in \mathcal{P}$ using at least five vertices of $\Mat^{-1}(A')$.
  If $u,u' \in \Mat^{-1}(A')$ are the first and last vertices used by $P_j$ in $A'$, respectively, then we can
  shortcut $P_j$ through $u \to \Mat(u) \to b \to u'$, again contradicting
  \autoref{prop:minimizing-for-solution}.
  We conclude that $|A'| \leq 4c$.

  Now, let $A'' = (V(P_i) \cap A) \setminus (A' \cup \{v_1, v_2, v_3, v_4\})$.
  Thus $|A''| \geq 4c+1$ and every $w \in A'' \cup \Mat^{-1}(A'')$ has $|\List(w)| = c$.
  The remaining part of this proof follows similarly to what is done in the proof of
  \autoref{lem:bounded-congested-in-B}, but only one shortcut is needed since $|L(b)| \leq c-1$.
  The size of $A''$ implies that there is some $\ell \in [k]$ and at least five vertices
  $v'_1, \ldots, v'_5 \in A''$ such that $\ell \in \List(\Mat^{-1}(v'_j)) \cap
  \List(v'_j)$ for $j \in [5]$.
  If $P_\ell$ uses a vertices of $A''$ before any vertex of its mirror and
  $v'_1$ and $v'_5$ are the first and last of the vertices used by $P_\ell$ in $A''$,
  respectively, then we can shortcut this path through $v'_1 \to b \to \Mat^{-1}(v'_5) \to v_5$.
  The case when $P_\ell$ uses a vertex of $\Mat^{-1}(A'')$ before any vertex of $A''$ is
  analogous and item \lipItem{1} follows.

  The proof of item~\lipItem{2} is mostly the same as in the first item, but with one small difference.
  By contradiction, let $P_i$ be a path of $\mathcal{P}$ using $q \geq 8c + 5$ vertices of $C$.
  Since all arcs between $B$ and $C$ are orienteed from $B$ to $C$, in this first step we
  need to exclude the four last vertices of $P_i$ in $C$ to say that the mirror of the
  remaining vertices of $C$ contains only vertices which are used by $c$ paths of $\mathcal{P}$.
  Let $u_1 \prec_i u_2 \prec_i \cdots \prec_i u_q$ be the vertices in $V(P_i) \cap C$
  ordered as they appear in $P_i$.
  If $\Mat(u_j)$ is $i$-free for some $j \leq q-4$, then we can shortcut $P_i$ through
  $u_j \to \Mat(u_j) \to b \to u_q$ to create a new solution.
  Thus, every vertex in $A' = \{\Mat(u_j) \mid j \leq q-4\}$ is $i$-saturated.
  The remaining part of the proof follows analogously to the proof of item~\lipItem{1}
  and the result follows.
\end{proof}

\subsection{Finding the irrelevant vertex in polynomial time}
\label{sec:finding-irrelevant-vertex}
Applying \autoref{lem:bounded-intersection-in-A-C} we show an improved version of \autoref{lem:bounded-congested-in-B} which is needed in our proof.

\begin{corollary}\label{cor:bounded-paths-in-B}
For all $i \in [k]$, $|V(P_i) \cap B| \leq 4$.
\end{corollary}
\begin{proof}
Assume that there is a $P_i$ using $z \geq 5$ vertices of $B$.
Let $b_1$ and $b_z$ be the first and last of those vertices under $\prec_i$, respectively.

Applying both items of \autoref{lem:bounded-intersection-in-A-C}, we conclude that there is an $i$-free pair $\{u, \Mat(u)\}$ with $u \in C$ after eliminating at most $2k(8c+4)$ pairs $\{u', \Mat(u')\}$ which can be intersected by $V(P_i)$.
Thus we can shortcut $P_i$ through $b_1 \to u \to \Mat(u) \to b_z$ and the result follows.
\end{proof}


We are now ready to prove \autoref{thm:FPT-large-congestion} restricted to semicomplete digraphs.
The goal is to give a polynomial-time algorithm that finds a vertex $b \in B$ such that if there is a solution for $\mathcal{I}$, then there is another solution whose paths all avoid $b$.

We follow the blueprint of the proof by Fomin and Pilipzcuk~\cite[Theorem 9.1]{FOMIN201978}.
The goal is to first find a large set $X \subseteq B$ such that every path of a solution entering some $b' \in X$ from a vertex not in $A$  can be rerouted to access $b$ from a vertex in $A$.
Then, we must show that there is a $b \in X$ such that every path of the solution leaving $b$ to a vertex not in $C$ can be rerouted to leave $b$ to a vertex in $C$.
As these rerouting steps do not use too many extra vertices in the triple, by \autoref{cor:bounded-paths-in-B}, and our hypothesis that the triple is large, we can then argue that $b$ be can be replaced with another $b^* \in B$, that is unused by the solution, and thus $b$ can be safely removed from the graph.

Thus, at some point in the proof, we need to analyze how a path of a solution for $\mathcal{I}$ can enter and leave $B$, and show how to reroute paths not respecting the desired behavior using vertices \emph{outside} of the triple.
In Fomin and Pilipczuk's proof, the fact that they consider collisions in arcs plays  a major role here in this step.
Informally, in a particular case of their analysis and using our notation, if a path $P_i$ accesses $b' \in B$ through a vertex $v \in V(T) \setminus A$ and $v$ has ``sufficiently many'' out-neighbors in $V(T) \setminus A$ which are in-neighbors of distinct vertices of $B$, then at least one of those out-neighbors, say $v'$, can be used to reroute $P_i$ through $v \to v' \to u \to \Mat(u) \to b$.
This holds due to three main properties: (i) the fact that at most $k$ arcs leaving a vertex are used by any solution, and thus any vertex of out-degree at least $k+1$ has one arc leaving it that is free, (ii) their versions of \cref{lem:bounded-congested-in-B,lem:bounded-intersection-in-A-C}, and (iii) a clever analysis on the behavior of the arcs between in-neighbors of vertices in $B$ outside of $A$.

In our case, an analogous of the first of those three points requires more work since, a priori, it may seem that we have little control on how a solution intersects the in-neighborhood of $B$ outside of $A$.
We prove that this is not the case by applying another shortcutting argument which allows us to bound how many vertices a path $P_i \in \mathcal{P}$ can use in a particular subset of in-neighbors of $B$.
In addition, after the first rerouting round, we can no longer rely on  \autoref{prop:minimizing-for-solution}, as we do increase the length of each path by an $\bigO{1}$ factor.
Thus a small trick is needed to apply a similar shortcutting argument as the one used in the in-neighbors of $B$, this time to the out-neighbors of $B$; this happens when we want to find a vertex of $B$ that, when used by a path, is always followed by a vertex of $C$.

\begin{proof}[Proof of \autoref{thm:FPT-large-congestion} with $h = 0$]
We remind the reader of the assumptions made in the beginning of~\autoref{sec:irrelevant-vertex-in-triples}.
Let:
\begin{align*}
  d_2(k) &= 8k(4k+1) + 8k + c,\\
  m_2(k) &= 8k + c,\\
  x(k) &= 2(d_2(k) \cdot m_2(k)),\\
  d_1(k) &= 7k(4k+1) + 8k + x(k), \text{ and}\\
  m_1(k) &= 8k + x(k).
\end{align*}

Note that $x(k) = \Ocal(k^3)$ is the asymptotically greatest among these five functions.
However, in order to apply \autoref{lem:bounded-congested-in-B}, and consequently obtain \autoref{lem:bounded-intersection-in-A-C} and finally \autoref{cor:bounded-paths-in-B}, we need to start with $f(k) = \Ocal(2^{k \log k})$.
This is the only point of the proof where an exponential dependency on $k$ is requested. \autoref{cor:bounded-paths-in-B} hints that there might be a way to push this bound to a polynomial in $k$, but we do not see how to circumvent the $\binom{k}{c}$ blowup originating from \autoref{lem:bounded-congested-in-B}.

Since the size of $B$ has the largest impact in the requested order of $\mathcal{K}$, in this proof we refrain from repeatedly proving that $i$-free pairs ${u, \Mat(u)}$ with $u \in C$ can be found when they are needed.
By \autoref{lem:bounded-intersection-in-A-C} only $\Ocal(k^2)$ such pairs are initially used by $\mathcal{P}$ and each rerouting done below uses at most one extra pair.

\medskip
\noindent\textbf{First round of reroutings.}
The goal is to first find a set $X \subseteq B$ of size at least $x(k)$ such that if there is a solution for $\mathcal{I}$ then there is another solution $\mathcal{Q} = \{Q_1, \ldots, Q_k\}$ where:
\begin{romanenumerate}
  \item for all $b \in X$, if $b$ is accessed by some $v \in V(Q_i)$ with $i \in [k]$ then $v \in A$,
  \item at most $4k$ new vertices of $B$ and $4k$ new pairs $\{u, \Mat(u)\}$ with $u \in C$ are used by $\mathcal{Q}$ in comparison with $\mathcal{P}$, and
  \item each $Q_i \in \mathcal{Q}$ uses at most one new vertex of $V(T) \setminus V(\mathcal{K})$ in comparison with $\mathcal{P}$.
\end{romanenumerate}

For every $b \in B$, let
\[R_b = \{v \in V(T) \setminus A \mid (v,b) \in E(T) \text{ and } |N^+(v) \cap B| \geq m_1(k)\},\]
\[S_b = \{v \in V(T) \setminus A \mid (v,b) \in E(T) \text{ and } |N^+(v) \cap B| \leq m_1(k) - 1\},\] 
and $B_\emptyset = \{b \in B \mid S_b = \emptyset\}$.
If $|B_\emptyset| \geq x(k)$, we claim that we can choose $X = B_\emptyset$.

To prove this claim, let $b \in B_\emptyset$ and assume that there is a path $P_i \in \mathcal{P}$ that accesses $b$ from a vertex $v \in V(T) \setminus A$.
With the choice of $m_1(k)$, we apply \autoref{cor:bounded-paths-in-B} to conclude that $v$ has at least $4k$ out-neighbors in $B \setminus X$ which are not used by paths of $\mathcal{P}$, and \autoref{lem:bounded-intersection-in-A-C} implies that at least $4k$ pairs $\{u, \Mat(u)\}$ with $u \in C$ are $j$-free for any $j \in [k]$.
Thus we can choose an $i$-free vertex $w \in N^+(v) \cap B \setminus X$ and an $i$-free pair $\{u, \Mat(u)\}$ to reroute $P_i$ through $v \to w \to u \to \Mat(u) \to b$.
By \autoref{cor:bounded-paths-in-B} at most $4k$ reroutings are needed since each $P_i$ can enter $X$ at most four times, and each rerouting can be done through distinct choices of $w$ and $\{u, \Mat(u)\}$.
Note that it is important to push this rerouting to vertices outside of $X$.
If, for instance, $w \in X$, we would not be decreasing how many accesses from outside of the triple are done to $X$, essentially voiding our rerouting.

Assume now that $|B_\emptyset| \leq x(k)-1$ and let $B_S = B \setminus B_\emptyset$.
We build an auxiliary semicomplete digraph $\Gamma$ with vertex set $B_S$ as follows.
For every ordered pair $\{b,b'\} \subseteq B_S$, we add $(b, b')$ to $E(\Gamma)$ if, for all $v \in S_b$, it holds that $v\in S_{b'}$ or there is a $w \in S_{b'}$ such that $(v,w) \in E(T)$.
Intuitively, the goal is to use arcs of $\Gamma$ to identify when we can reroute a path entering $b$ from outside of $\mathcal{K}$ through $b'$.
To see that $\Gamma$ is indeed semicomplete, let $b,b' \in V(\Gamma)$.
If there is no arc from $b$ to $b'$, then some $v \in S_b \setminus S_{b'}$ is an in-neighbor in $T$ of every vertex in $S_{b'} \setminus S_b$. By definition $(b',b) \in E(\Gamma)$ and therefore $\Gamma$ is semicomplete.
The definition of $\Gamma$ immediately implies that it can be built in time $\Ocal(f^2(k) \cdot n^2)$

Now, the goal is to find a set of $x(k)$ vertices of $V(\Gamma)$ which have out-degree at least $d_1(k) \cdot m_1(k)$ in $\Gamma$. 
By discarding $B_\emptyset$ from $B$, we conclude that $|V(\Gamma)| \geq |B| - x(k)+1 = f(k) - x(k) + 1$.
Assume that at most $x(k)-1$ vertices of $\Gamma$ have degree at least $d_1(k) \cdot m_1(k)$.
Then $\Gamma$ has at most $\alpha = (x(k) - 1)(|V(\Gamma)|-1) + (|V(\Gamma)| - x(k) + 1)(d_1(k)\cdot m_1(k) - 1) \geq (x(k) - 1)(f(k)-x(k)) + (f(k) - 2x(k) + 2)(d_1(k)\cdot m_1(k) - 1)$ arcs.
Thus by choosing $f(k)$ large enough to guarantee that $\binom{|V(\Gamma)|}{2} > \alpha$, we guarantee that at least $x(k)$ vertices of $\Gamma$ have the desired out-degree.
Let $X$ contain all $b \in V(\Gamma)$ such that $\deg^+_{\Gamma}(v) \geq d_1(k) \cdot m_1(k)$.
We now prove that paths of $\mathcal{P}$ accessing $X$ from $V(T) \setminus A$ can be rerouted in the desired way.
To achieve this, we introduce another shortcutting argument for $\mathcal{P}$.

\begin{claim}\label{claim:bounded-used-outside-of-B-first-step}
Let $W = \bigcup_{b \in B} S_b$ and $M^+$ be a matching from $W$ to $B$.
Then for all $i \in [k]$, the path $P_i$ uses at most $7(4k+1)-1$ vertices which are tails of arcs in $M^+$.
\end{claim}
\begin{proof}[Proof of the claim] By contradiction, assume that there is a $P_i \in \mathcal{P}$ using at least $7(4k+1)$ vertices which are tails of arcs in $M^+$, and let $W'$ denote the set of all such vertices.
With $r = 4k+1$, we split $W'$ arbitrarily into $4k$ sets $W'_1, \ldots, W'_{r-1}$ of size seven, plus one set $W'_r$ with the remaining vertices of $W'$ not included in any of the other sets; note that $|W'_r| \geq 7$.

By \autoref{cor:bounded-paths-in-B}, there is one set $W'_j$ such that all heads of arcs of $M^+$ with tail in $W'_j$ are $i$-free: at most $4k$ vertices of $B$ are not $i$-free, so it follows that at most $4k$ different $W'_i$'s have at least one head in this set.
Let $v, v'$ be the first and last vertices of $W'_j$ with respect to $\prec_i$ and let $b$ be the head of the arc $(v, b) \in M^+$.
Since $v' \in S_b$ it holds that $|N^-_T(v') \cap B| \geq |B| - m_1(k)-1 \geq 4k + 1$.
Applying \autoref{lem:bounded-intersection-in-A-C} and \autoref{cor:bounded-paths-in-B} we can
shortcut $P_i$ through $v \to b \to u \to \Mat(u) \to b' \to v'$, where $b'$ is an $i$-free vertex of $B$ and $\{u, \Mat(u)\}$ is an $i$-free pair with $u \in C$. Recall that this is a contradiction since we are assuming that $\mathcal{P}$ is of minimum total length, and so we have proved the claim.
See \autoref{fig:shorcut-outside-triple} for an illustration of this shortcut.
\end{proof}

\begin{figure}[h]
  \centering
  \begin{tikzpicture}
    \foreach \i/\lpos/\name in {1/90/b',4/-90/b} {
      \node (a\i) at (\i-1,-.5) {};
      \node (c\i) at ($(a\i) + (0,4.5)$) {};
    }
    \node (b1) at (-2,0) {};
    \node[blackvertex, label={-90:{$b$}}] (b4) at (-2,3) {};

    \node[blackvertex, label=180:{$v$}] (w1) at ($(b4) + (-2,0)$) {};
    \node[blackvertex, label=180:{$v'$}] (wz) at ($(b1) + (-2,0)$) {};
    \draw[arrow, big snake, shorten >= 0pt] (w1) --
    node[pos=1/7,blackvertex, scale=.5] {}
    node[pos=2/7,blackvertex, scale=.5, yshift = .1cm] {}
    node[pos=3/7,blackvertex, scale=.5, yshift = .2cm] {}
    node[pos=4/7,blackvertex, scale=.5, yshift = .3cm] (p) {}
    node[pos=5/7,blackvertex, scale=.5, yshift = .4cm] {} 
    node[pos=6/7,blackvertex, scale=.5, yshift = .5cm] {}(wz) ;
    \node[rectangle,draw,fit=(a1)(a4), label=0:{$A$}] {};
    \node[rectangle,draw,fit=(b1)(b4), label=90:{$B$}] (blab) {};
    \node[rectangle,draw,fit=(c1)(c4), label=0:{$C$}] {};
    \node at ($(p) + (-.75,0)$) {$P_i$};
    \node[blackvertex, label=90:{$b'$}] (b3) at ($(b1) + (0,1.5)$) {};
    \node[blackvertex, label={[yshift=0cm]0:{$u$}}] (u) at ($(c1) + (1.5,0)$) {};
    \node[blackvertex, label={[yshift=0cm]0:{$\Mat(u)$}}] (v) at ($(a1) + (1.5,0)$) {};
    \draw[thick, goodblue,arrow] (w1) -- (b4);
    \draw[thick, goodblue,arrow] (b4) -- (u);
    \draw[thick, goodblue,arrow] (u) -- (v);
    \draw[thick, goodblue,arrow] (v) -- (b3);
    \draw[thick, goodblue,arrow] (b3) -- (wz);

  \end{tikzpicture}
  \caption{In blue, a shortcut for a path $P_i$ using at least $7(4k+1)$ vertices of $W = \bigcup_{b \in B} S_b$.}
  \label{fig:shorcut-outside-triple}
\end{figure}
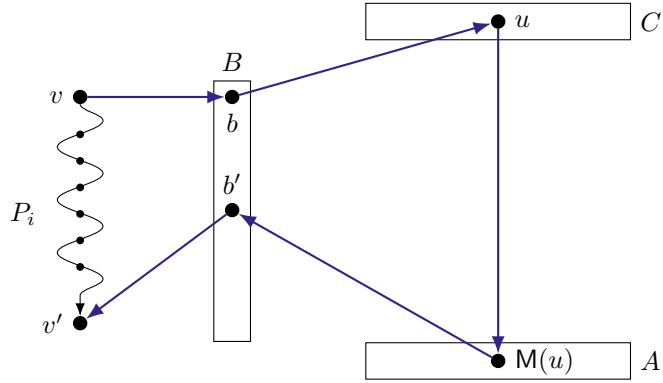

Assume that there is a $P_i \in \mathcal{P}$ such that $P_i$ accesses some $b \in X$ through a vertex $v \in V(T) \setminus A$.
If $v \in R_b$ then we proceed as if $S_b = \emptyset$: the out-degree of $v$ in $B$ is large enough to reroute through a vertex of $B$ unoccupied by any path.
Otherwise, $v \in S_b$.
In this case, we use the bound on the degree of $b$ in $\Gamma$ and apply \autoref{claim:bounded-used-outside-of-B-first-step} to find a large matching in $T$ from $V(T) \setminus A$ to $B$ that can be used to reroute $P_i$. 


We build in $T$ a matching $Y_0$ from $W = \bigcup_{b \in B}S_b$ to $B$ such that $Y_0 =$ $\{(v, b)$, $(v_1, b_1),$ $\ldots,$ $(v_z, b_z)\}$ and:
\begin{enumerate}
\item for $i \in [z]$ it holds that $(v, v_i) \in E(T)$, and
\item $z \geq d_1(k)$. 
\end{enumerate}
We say that $(v, b)$ is the \emph{distinguished} arc of $Y_0$.

Let us show how to build $Y_0$.
We start with a collection $\mathcal{W}$ of all sets $S_{b'}$ with $b' \in N^+_{\Gamma}(b)$ which are the candidates from which we can pick $v_1$; recall that $\deg^+_\Gamma(b) \geq d_1(k)\cdot m_1(k)$.
First we trim from $\mathcal{W}$ all such sets containing $v$, which are at most $m_1(k) - 1$ as $v \in S_b$.
By definition, at most $m_1(k)-1$ sets are discarded in  this way. 
Choose $b_1$ such that $S_{b_1} \in \mathcal{W}$.
The construction of $\Gamma$ and the fact that $(b,b_1) \in E(\Gamma)$ imply that there is a $v_1 \in S_{b_1}$ such that $(v, v_1) \in E(T)$.
We add $(v_1, b_1)$ to $Y_0$ and iterate.
Assume that $\ell$ arcs $(v_1, b_1), \ldots, (v_\ell, b_\ell)$ have been chosen this way, with $\ell \in [d_1(k) - 1]$.
Thus, after discarding from $\mathcal{W}$ every $S_{b'}$ containing any $v_j$ with $j \in [\ell]$, we conclude that at least $m_1(k) \cdot (d_1(k) - \ell) + \ell$ candidates remain in $\mathcal{W}$ for the choice of $v_{\ell+1}$.
Since $\deg^+_{\Gamma}(b) \geq d_1(k) \cdot m_1(k)$, this procedure does not end before $d_1(k)$ arcs are chosen, and therefore we have constructed $Y_0$ with the desired properties.
Notice that $Y_0$ can be built by simply observing the degrees of vertices in $T$, which is given in the input, and $\Gamma$, which can be built in polynomial time, and verifying which vertices are in each of the sets $S_b$, which can be done in time $\Ocal(f^2(k) \cdot n^2)$.

We now show that for some $j \in [d_1(k)]$ the arc $(v_j,b_j) \in Y_0$ has both its endpoints $i$-free.
By \autoref{claim:bounded-used-outside-of-B-first-step}, at most $7k(4k+1)$ vertices which are heads of arcs in $Y_0$ are used by paths of $\mathcal{P}$.
By \autoref{cor:bounded-paths-in-B}, at most $4k$ arcs of $Y_0$ have their endpoints in $B$ being used by a path of $\mathcal{P}$.
Thus, the choice of $d_1(k)$ guarantees that there are at least $4k$ arcs of the form $(v_j, b_j) \in Y_0$ with $j \in [d_1(k)]$ with both its endpoints being $i$-free; the $+8k$ term in the definition of $d_1(k)$ is present precisely so we can claim that these $4k$ arcs exist throughout the first rerouting phase.
This suffices to reroute the path $P_i$ through $v \to v_j \to b_j \to u \to \Mat(u) \to b$, where $\{u, \Mat(u)\}$ is an $i$-free pair with $u \in C$.
See \autoref{fig:rerouting-for-X} for an illustration of this rerouting.

\begin{figure}[h]
  \centering
  \begin{tikzpicture}
      \node[blackvertex, label={[yshift=0.1cm]90:{$b$}}] (b) at (-2, 0) {};
      \node[blackvertex, label=90:{$b_j$}] (bj) at (-2, -2) {};
      \node[blackvertex, label=180:{$v$}] (v) at ($(b) + (-2, 0)$) {};
      \node[blackvertex, label=180:{$v_j$}] (vj) at ($(bj) + (-2, 0)$) {};
      \draw[arrow, thick, dashed] (v) -- node[midway, above] {$P_i$} (b);
    
      \node[blackvertex, label={[yshift=0cm]0:{$u$}}] (u) at ($(b) + (2, 0.5)$) {};
      \node[blackvertex, label={[yshift=0cm]0:{$\Mat(u)$}}] (w) at ($(u) + (0, -3)$) {};
      \node (m1) at ($(u) + (-.75, 0)$) {};
      \node (m2) at ($(u) + (.75, 0)$) {};
      \node[rectangle,draw,fit=(m1)(m2), label=0:$C$] {};
      \node (m1) at ($(w) + (-.75, 0)$) {};
      \node (m2) at ($(w) + (.75, 0)$) {};
      \node[rectangle,draw,fit=(m1)(m2), label=0:$A$] {};
      \node[rectangle, draw, fit=(b)(bj)] {};
      \draw[arrow, thick, goodblue] (v) -- (vj);
      \draw[thick, goodblue,arrow] (vj) -- (bj);
      \draw[thick, goodblue,arrow] (bj) -- (u);
      \draw[thick, goodblue,arrow] (u) -- (w);
      \draw[thick, goodblue,arrow] (w) -- (b);
  \end{tikzpicture}
  \caption{In blue, the rerouting for the path $P_i$ (dashed) accessing $b$ from $V(T) \setminus A$.}
  \label{fig:rerouting-for-X}
\end{figure}
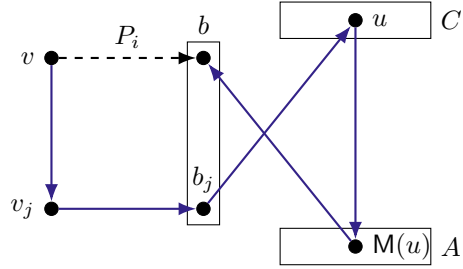

Since $|X| = x(k)$, for each other $b' \in X$ we can do a similar rerouting as described above while the choice of $d_1(k)$ allows us to, in addition, avoid rerouting any paths through distinguished arcs $(v'',b'')$ with $b'' \in X$.
Again by \autoref{cor:bounded-paths-in-B}, at most $4k$ reroutings are done in this case as well, and therefore the choice of $d_1(k)$ guarantees that there is room for all reroutings.

We denote by $\mathcal{Q}$ our new solution built from $\mathcal{P}$ after this first round of reroutings.
Essentially, after the first round of reroutings we have that $\mathcal{Q}$ uses at most $4k$ new vertices of $B$ and $4k$ new pairs $\{u, \Mat(u)\}$, with $u \in C$, when compared with $\mathcal{P}$, and all reroutings are done through vertices outside of $X$.

\medskip
\noindent\textbf{Second round of reroutings.}
The goal now is to find a \emph{special} vertex $b \in X$ such that if there is a solution for $\mathcal{I}$, then there is a solution $\mathcal{Q}' = \{Q'_1, \ldots, Q'_k\}$ such that:
\begin{alphaenumerate}
  \item any path $Q'_i \in \mathcal{Q}'$ leaving $b$ does so through a vertex $u \in C$, and
  \item at most $c$ new vertices of $B$ and at most $c$ new pairs $\{u, \Mat(u)\}$ are used by $\mathcal{Q}'$ in comparison with $\mathcal{Q}$.
\end{alphaenumerate} 

The remaining part of the proof is mostly similar and easier than the previous one.
The main difference is that at this point \autoref{prop:minimizing-for-solution}  no longer holds with respect to $\mathcal{Q}$.
Hence, in order to obtain a similar result as in \autoref{claim:bounded-used-outside-of-B-first-step}, we use the third defining property of $X$: it uses at most one new vertex of $V(T) \setminus V(\mathcal{K})$ in comparison with $\mathcal{P}$ for each of the $k$ paths. 
For every $b \in B$, let:
\[R'_b = \{v \in V(T) \setminus A \mid (v,b) \in E(T) \text{ and } |N^-_T(v) \cap B| \geq m_2(k)\},\]
\[S'_b = \{v \in V(T) \setminus A \mid (v,b) \in E(T) \text{ and } |N^-_T(v) \cap B| \leq m_2(k) - 1\}.\]


\begin{claim}\label{claim:bounded-used-outside-of-B-second-step}
Let $W = \bigcup_{b \in B}S'_b$ and $M^-$ be a matching from $B$ to $W$.
Then for all $i \in [k]$, the path $Q_i$ uses at most $8(4k+1)-1$ vertices which are heads of arcs in $M^-$.
\end{claim}

\begin{proof}[Proof of the claim]
By contradiction, assume that there is a $Q_i \in \mathcal{Q}$ using at least $8(4k+1)$ vertices which are heads of arcs in $M^-$, and let $W'$ denote the set of all such vertices.
With $r = 4k+1$, we split $W'$ arbitrarily into $4k$ sets $W'_1, \ldots, W'_{r-1}$ of size eight, plus one set $W'_r$ with the remaining vertices of $W'$ not included in any of the other sets; note that $|W'_r| \geq 8$.

By \autoref{cor:bounded-paths-in-B}, there is one set $W'_j$ such that all tails of arcs of $M^-$ with heads in $W'_j$ are not used by any path in $\mathcal{Q}$.
We consider two cases.

If $Q_i = P_i$, then we can shortcut $P_i$ similarly to what is done in the proof of \autoref{claim:bounded-used-outside-of-B-first-step}, contradicting \autoref{prop:minimizing-for-solution}.
Denote by $v,v'$ the first and last vertices of $W'_j$ with respect to $\prec_i$, and let $b'$ be the tail of the arc $(b', v') \in M^-$.
Since $v$ is in some $S'_b$, it holds that $|N^+_T(v) \cap B| \geq |B| - m_2(k) + 1 \geq 4k+1$.
Thus, by \autoref{lem:bounded-intersection-in-A-C} and \autoref{cor:bounded-paths-in-B} we can shortcut $P_i$ through $v \to b \to u \to \Mat(u) \to b' \to v'$, where $b$ is an $i$-free vertex in $B$ and $\{u, \Mat(u)\}$ is an $i$-free pair with $u \in C$, and the claim follows.
See \autoref{fig:shorcut-outside-triple-second-case} for an illustration of this shortcut.

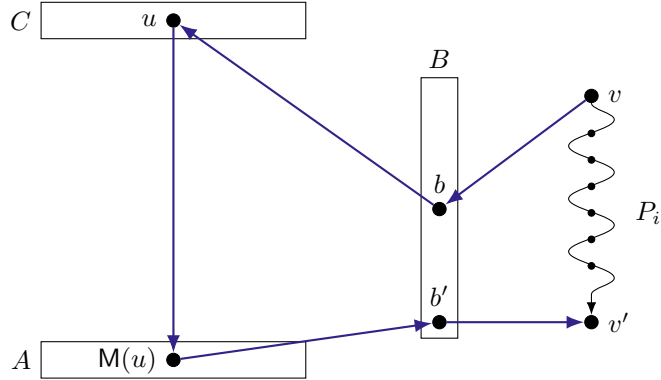
\begin{figure}[h]
  \centering
  \begin{tikzpicture}
    \foreach \i/\lpos/\name in {1/90/b',4/-90/b} {
      \node (a\i) at (-3-\i,-.5) {};
      \node (c\i) at ($(a\i) + (0,4.5)$) {};
    }
    \node[blackvertex, label={$b'$}] (b1) at (-2,0) {};
    \node (b4) at (-2,3) {};

    \node[blackvertex, label=0:{$v$}] (w1) at ($(b4) + (2,0)$) {};
    \node[blackvertex, label=0:{$v'$}] (wz) at ($(b1) + (2,0)$) {};
    
    \draw[arrow, big snake, shorten >= 0pt] (w1) --
    node[pos=1/7,blackvertex, scale=.5] {}
    node[pos=2/7,blackvertex, scale=.5, yshift = .1cm] {}
    node[pos=3/7,blackvertex, scale=.5, yshift = .2cm] {}
    node[pos=4/7,blackvertex, scale=.5, yshift = .3cm] (p) {}
    node[pos=5/7,blackvertex, scale=.5, yshift = .4cm] {} 
    node[pos=6/7,blackvertex, scale=.5, yshift = .5cm] {} (wz) ;

    \node[rectangle,draw,fit=(a1)(a4), label=180:{$A$}] {};
    \node[rectangle,draw,fit=(b1)(b4), label=90:{$B$}] {};
    \node[rectangle,draw,fit=(c1)(c4), label=180:{$C$}] {};
    \node at ($(p) + (.75,0)$) {$P_i$};
    \node[blackvertex, label=90:{$b$}] (b3) at ($(b1) + (0,1.5)$) {};
    \node[blackvertex, label={[yshift=0cm]180:{$u$}}] (u) at ($(c1) + (-1.5,0)$) {};
    \node[blackvertex, label={[yshift=0cm]180:{$\Mat(u)$}}] (v) at ($(a1) + (-1.5,0)$) {};

    \draw[thick, goodblue,arrow] (w1) -- (b3); 
    \draw[thick, goodblue,arrow] (b3) -- (u);
    \draw[thick, goodblue,arrow] (u) -- (v);
    \draw[thick, goodblue,arrow] (v) -- (b1);
    \draw[thick, goodblue,arrow] (b1) -- (wz);
  \end{tikzpicture}
  \caption{In blue, a shortcut for a path $P_i$ associated with a path $Q_i$ which uses at least $8(4k+1)$ vertices of $W = \bigcup_{b \in B} S'_b$.}
  \label{fig:shorcut-outside-triple-second-case}
\end{figure}

Assume now that $Q_i$ was built by rerouting the path $P_i$.
By property \lipItem{(iii)} of $X$, at least seven vertices of $W'_j$ are used by $P_i$.
In this case we can shortcut $P_i$ similarly to the first case, and the claim follows.
\end{proof}

Assume that there is a $b \in B$ such that $S'_b = \emptyset$ and that there is a path $Q_i$ leaving $b$ through a vertex $v \in V(T) \setminus C$.
Thus $v \in R'_b$ and therefore $|N^-_T(v) \cap B| \geq m_2(k)$.
By \autoref{cor:bounded-paths-in-B} this implies that $v$ has in-neighbors $b_1, \ldots, b_c \in B$ such that no $b_j$ with $j \in [c]$ is used by a path of $\mathcal{P}$ and was not used by any of the reroutings from the first round.
Thus, due to \autoref{lem:bounded-intersection-in-A-C}, we can reroute $Q_i$ through $b \to u \to \Mat(u) \to b_1 \to v$, where $\{u, \Mat(u)\}$ is an $i$-free pair that was not used by any of the reroutings from the first round.
Since at most $c$ paths of $Q_i$ can use $b$, there is room in $b_1, \ldots, b_c$ to reroute all paths of $\mathcal{Q}$ using $b$.
In this case, $b$ is a good choice for the special vertex and the second round of reroutings is finished.

Assume now that every $b \in B$ has $S'_b \neq \emptyset$.
We construct an auxiliary semicomplete digraph $\Gamma'$ with vertex set $X$ by adding to it the arc $(b', b)$ if, for all $v \in S'_b$, either $v \in S'_{b'}$ or there is an arc from a $v' \in S'_{b'}$ to $v$.
The proof that $\Gamma'$ is indeed semicomplete is analogous to the proof for $\Gamma$, and it can also be constructed in time $\Ocal(f^2(k) \cdot n^2)$.

Since $|X| = 2(m_2(k)\cdot d_2(k))$, it holds that there is a vertex $b \in V(\Gamma')$ such that $\deg^-_{\Gamma'}(b) \geq m_2(k)\cdot d_2(k)$.
We show that $b$ is a good choice for the special vertex.
To do so, assume that there is a path $Q_i \in \mathcal{Q}$ that leaves $b$ through a vertex $v \in V(T) \setminus C$.
If $v \in R'_{b}$ we proceed as in the previous case of this round.
Assume now that $v \in S'_{b}$.

We build in $T$ a matching $Y'$ from $B$ to $W' = \bigcup_{b \in B}S'_b$ such that $Y' = \{(v, b)$, $(v_1, b_1)$, $\ldots$, $(v_z, b_z)\}$ and:
\begin{enumerate}
  \item for $i \in [z]$ it holds that $(v_i, v) \in E(T)$, and 
  \item $z \geq d_2(k)$.
\end{enumerate}

We start with a collection $\mathcal{W'}$ of all sets $S'_{b'}$ with $b' \in N^-_{\Gamma'}(b)$ which are the candidates from which we can pick $v_1$.
We trim from $\mathcal{W'}$ all such sets containing $v$.
By definition, at most $m_2(k)-1$ sets are discarded in this way.
Choose any $b_1$ such $S'_{b_1} \in \mathcal{W'}$.
The construction of $\Gamma'$ and the fact that $(b_1, b) \in E(\Gamma')$ implies that there is a $v_1 \in S'_{b_1}$ such that $(v_1, v) \in E(T)$.
We add $(b_1, v_1)$ to $Y'$.

For the iterative argument, assume that $\ell$ arcs $(b_1, v_1), \ldots, (b_{\ell}, v_{\ell})$ have been chosen this way, with $\ell \in [d_2(k) - 1]$.
Thus, after discarding from $\mathcal{W}'$ every $S'_{b'}$ containing any $v_j$ with $j \in [\ell]$, we conclude that at least $m_2(k) \cdot (d_2(k) - \ell) + \ell$ candidates for the choice of $v_{\ell+1}$ remain in $\mathcal{W}'$.
Since $\deg^-_{\Gamma'}(b) \geq d_2(k) \cdot m_2(k)$, this procedure does not end before $d_2(k)$ arcs are chosen this way.
Therefore, we terminate with the desired set $Y'$ which, similar to the similar procedure in the first round of reroutings, can be built in $\Ocal(f^2(k) \cdot n^2)$ time.

We show that the choice of $d_2(k)$ implies that at least one $(b_j, v_j) \in Y'$ has both its endpoints available to reroute $Q_i$.
By \autoref{claim:bounded-used-outside-of-B-second-step} at most $8k(4k+1)$ vertices which are tails of arcs in $Y'$ are used by paths of $\mathcal{Q}$.
We eliminate from $Y'$ all arcs associated with those vertices.
By \autoref{cor:bounded-paths-in-B}, at most $4k$ arcs of $Y'$ have their endpoints in $B$ being used by paths of $\mathcal{P}$.
By property \lipItem{(iii)} of $X$, the endpoints in $B$ of at most $4k$ arcs of $Y'$ were used by reroutings from the first round.
We eliminate from $Y'$ all arcs from those two sets.
In the end, at least $c$ arcs remain in $Y'$.
Let $(b_j, v_j)$ be one of those remaining arcs.
Then we can reroute $Q_i$ through $b \to u \to \Mat(u) \to b_j \to v_j \to v$, as shown in \autoref{fig:rerouting-for-b-in-X}.
Since at least $c$ arcs remained in $Y'$ and at most $c$ paths of $\mathcal{Q}$ use $b$, we can reroute any path leaving $b$ to a vertex in $V(T) \setminus A$  similarly.
After the second round of reroutings, we choose $b$ as the special vertex.

\begin{figure}[h]
  \centering
  \begin{tikzpicture}

  \node[blackvertex, label={[yshift=.1cm]90:{$b$}}] (b) at (-2, 0) {};
  \node[blackvertex, label=90:{$b_j$}] (bj) at (-2, -2) {};
  \node[blackvertex, label=0:{$v$}] (v) at ($(b) + (2, 0)$) {};
  \node[blackvertex, label=0:{$v_j$}] (vj) at ($(bj) + (2, 0)$) {};
  \draw[arrow, thick, dashed] (b) -- node[midway, above] {$Q_i$} (v);

  \node[blackvertex, label={[yshift=0cm]180:{$u$}}] (u) at ($(b) + (-2, 0.5)$) {};
  \node[blackvertex, label={[xshift=1pt]180:{$\Mat(u)$}}] (w) at ($(u) + (0, -3)$) {};
  \node (m1) at ($(u) + (-.75, 0)$) {};
  \node (m2) at ($(u) + (.75, 0)$) {};
  \node[rectangle,draw,fit=(m1)(m2), label=180:$C$] {};
  \node (m1) at ($(w) + (-.75, 0)$) {};
  \node (m2) at ($(w) + (.75, 0)$) {};
  \node[rectangle,draw,fit=(m1)(m2), label=180:$A$] {};
  \node[rectangle, draw, fit=(b)(bj)] {};
  \draw[arrow, thick, goodblue] (b) -- (u);
  \draw[thick, goodblue,arrow] (u) -- (w);
  \draw[thick, goodblue,arrow] (w) -- (bj);
  \draw[thick, goodblue,arrow] (bj) -- (vj);
  \draw[thick, goodblue,arrow] (vj) -- (v);
  \end{tikzpicture}
  \caption{In blue, the rerouting for the path $Q_i$ (dashed) leaving $b$ to $V(T) \setminus A$.}
  \label{fig:rerouting-for-b-in-X}
\end{figure}
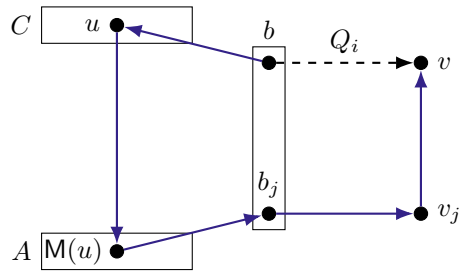

\medskip
\noindent\textbf{Finding the irrelevant vertex.}
It is easy to see that $b$ is an irrelevant vertex.
By property \lipItem{(b)} there are $c$ vertices in $B$ which are not used by any path of $\mathcal{Q}'$.
Since any path $Q'_i$ using $b$ reaches this vertex from a $v \in A$ and leaves it through a $u \in C$, it suffices to reroute $Q'_i$ using $v \to b' \to u$, where $b'$ is any vertex of $B$ not used by paths of $\mathcal{Q}'$.
Since there are at least $c$ such vertices, all paths of $\mathcal{Q}'$ can be rerouted away from $b$ and the result follows.
\end{proof}

Note that \cref{lem:B-xor-free-pair,lem:bounded-congested-in-B,lem:bounded-intersection-in-A-C} and \autoref{cor:bounded-paths-in-B} only deal with arcs and reroutings {\sl inside} of the $f(k)$-triple $\mathcal{K}$.
Thus, in order to adapt the proof to $h$-semicomplete digraphs, it suffices to add a factor of $h$ to the steps above where vertices of large in- or out-neighborhood in $B$ or in the heads or tails of the matchings of the form $Y_0$ (after \autoref{claim:bounded-used-outside-of-B-first-step}) and $Y'$ (after \autoref{claim:bounded-used-outside-of-B-second-step}), for example.
This implies that at each of those steps, the number of non-neighbors of the observed vertex is taken into account, and thus the same ideas work for $h \geq 1$.

%% file: sections/future.tex
\section{Future research}
\label{sec:conclusions}

This work touched on several common strategies used to cope with the hardness of the \pname{Directed Disjoint Paths} problem.
We investigated its (parameterized) complexity on tournaments and some of its superclasses, both with and without the typical relaxation of allowing congestion in the vertices.
One of our key contributions is a fix (\autoref{thm:tournament_nph}) for a gap in the literature originating in a flaw in the \NP-hardness proof for \pname{$k$-DDP} on tournaments given in~\cite{thomassen_tournament_nph}; we also as adapt a win/win approach based on directed pathwidth first used in an \FPT algorithm for \pname{Directed Edge-Disjoint Paths} parameterized by the number of paths on tournaments~\cite{FominP19} to \pname{$(k,c)$-DDP} when $2c > k$.
We note, however, that this latter problem is \textit{not} known to be \NP-complete; while we are able to show that \pname{$(k,k^\varepsilon)$-DDP} is hard for every $\varepsilon \in [0,1)$, we are unable to extend this to $c = \varepsilon k$, or even $c = k/2 + 1$. Alongside the already challenging problems listed in \autoref{sec:intro}, we consider this the main open question related to our work.
We do not provide all the details for the proof of \autoref{thm:FPT-large-congestion} in its full generality, only for the case of semicomplete graphs, i.e., $h=0$. Extending our arguments to $h > 0$ is essentially going deeper into technicalities: we must increase the size of the triple taking $h$ into account and, when discussing the exterior neighborhood of the triple, increase the thresholds to classify the in- and out-neighbors of $B$ as having ``too many'' neighbors in $B$ or not.
This extension, however, is not enough to give an \FPT algorithm for \pname{$(k,c)$-DDP} on $h$-semicomplete graphs when $2c > k$.
In particular, two challenges remain: (\textit{i}) computing a triple for elements of this class in \FPT-time, as the only known algorithm being an \XP one introduced by Kitsunai, Kobayashi, and Tamaki~\cite{KitsunaiKT15}; and (\textit{ii}) devising an \FPT algorithm for the joint parameterization by $k$, $h$, and the directed pathwidth, which would also improve upon the \XP algorithm shown in~\cite{KitsunaiKT15}.
We recall that, among the classes studied in our work, an \FPT algorithm for $h$-semicomplete digraphs is the best we can hope for, as we have shown that \pname{$k$-DDP} and \pname{$(k,c)$-DDP} are \W[1]-hard when parameterized by $k$ and the number of covering tournaments on digraphs of directed pathwidth two (\autoref{thm:w1h_ch} and \autoref{thm:w1h_congestion}).

%% file: sections/flaw.tex
\section{Flaw in the \NP-completeness proof of Bang-Jensen and Thomassen}
\label{sec:flaw}

In this appendix we sketch the flaw in the proof of Bang-Jensen and Thomassen claiming that {\sc Directed Disjoint Paths} is \NP-complete on tournaments when  $k$ is part of the input.

Their approach is to first prove the \NP-completeness of a related problem on tournaments, and then reduce from this problem to {\sc Directed Disjoint Paths}. Namely, it is proved in~\cite[Theorem 6.1]{thomassen_tournament_nph} that the following problem $Q$ is \NP-complete on tournaments: given a tournament $T$ and a set of arcs $A \subseteq E(T)$, decide whether $T$ contains a directed cycle visiting all the arcs in $A$ (in any order). The proof of~\cite[Theorem 6.1]{thomassen_tournament_nph} is correct, and consists in a simple reduction from {\sc Hamiltonian Cycle} on general digraphs to $Q$ on tournaments.

The problem comes later: right after the proof of~\cite[Theorem 6.1]{thomassen_tournament_nph}, it is claimed that ``This proves that, if $k$ is not fixed, then the $k$-{\sc DDP} problem is \NP-complete on
tournaments''. This statement is not clear at all, the main issue being that in problem~$Q$ the set of arcs $A$ to be visited by the cycle is \emph{unordered}. Indeed, in order to construct an instance of $k$-{\sc DDP} starting from an instance $(T,A)$ of problem $Q$, the natural strategy would be to consider an arbitrary ordering $\sigma=(e_0, \ldots, e_{k-1})$ of the arcs in $A$, and construct a set $K$ of $k$ requests in $T$ as follows. For every $i \in \{0, \ldots, k-1\}$, add to $K$ a request from the head of $e_i$ to the tail of $e_{i+1}$, where indices are taken modulo $k$. One may hope that the desired cycle $C$ in $T$ visiting all the arcs in $A$ would translate to the existence of $k$ pairwise vertex-disjoint paths in $T$ satisfying all the requests in $K$. But this is not true, as $C$ may exist in $T$, but may yield a visiting ordering of the arcs in $A$ that is different from the ordering $\sigma$ that we have fixed arbitrarily. One could try to fix this issue by guessing all the possible orderings $\sigma$ of the arcs in $A$, but this would result in $k!$ choices, which is not allowed in a polynomial-time reduction since $k$ is considered as part of the input.

We have shared the above issue with the authors of~\cite{thomassen_tournament_nph} and they have confirmed to us~\cite{personal-communication} that it does not seem to be easily fixable.